\documentclass[11pt,onesided,reqno]{amsart}
\usepackage[foot]{amsaddr}
\usepackage{amsmath}
\usepackage{amssymb}
\usepackage{amsthm}
\usepackage{bbm}
\usepackage{graphicx}
\usepackage{enumerate}
\usepackage{hyperref}
\usepackage[margin=1in]{geometry}
\usepackage[section]{placeins}
\usepackage{algorithmic}
\usepackage{algorithm}
\usepackage{multirow}
\usepackage{url}

\newtheorem{theorem}{Theorem}[section]

\newtheorem{proposition}[theorem]{Proposition}
\newtheorem{lemma}[theorem]{Lemma}
\newtheorem{corollary}[theorem]{Corollary}
\theoremstyle{definition}
\newtheorem{remark}[theorem]{Remark}
\newtheorem*{terminology*}{Terminology}
\newtheorem{definition}[theorem]{Definition}

\newtheorem{example}[theorem]{Example}

\newcommand{\R}{\mathbb{R}}
\renewcommand{\P}{\mathbb{P}}
\newcommand{\eps}{\varepsilon}

\newcommand{\Z}{\mathbb{Z}}
\newcommand{\E}{\mathbb{E}}
\newcommand{\argmin}{\operatorname{argmin}}

\newcommand{\sign}{\operatorname{sign}}

\newcommand{\bS}{\mathbf{S}}
\newcommand{\bs}{\mathbf{s}}

\newcommand{\1}{\mathbbm{1}}

\newcommand{\supp}{\operatorname{supp}}
\newcommand{\Id}{\operatorname{Id}}
\newcommand{\diag}{\operatorname{diag}}
\renewcommand{\r}{\mathrm{r}}
\newcommand{\w}{\mathrm{w}}
\newcommand{\Lz}{\mathrm{L0}}
\newcommand{\W}{\mathrm{W}}
\newcommand{\TV}{\mathrm{TV}}

\newcommand{\ST}{\mathcal{ST}}
\newcommand{\Normal}{\operatorname{Normal}}
\newcommand{\s}{\mathsf{s}}
\renewcommand{\t}{\mathsf{t}}
\renewcommand{\a}{\mathsf{a}}

\newcommand{\sB}{\mathsf{B}}

\title{Approximate $l_0$-penalized estimation of piecewise-constant signals on
graphs}
\author{Zhou Fan}
\author{Leying Guan}
\address{Department of Statistics, Stanford University}
\email{zhoufan@stanford.edu, lguan@stanford.edu}
\thanks{ZF was supported by a Hertz Foundation Fellowship and an NDSEG
Fellowship (DoD, Air Force Office of Scientific Research, 32 CFR
168a).}

\begin{document}
\maketitle
\begin{abstract}
We study recovery of piecewise-constant signals on graphs by the
estimator minimizing an $l_0$-edge-penalized objective. Although exact
minimization of this objective may be computationally intractable, we 
show that the same statistical risk guarantees are achieved by the
$\alpha$-expansion algorithm which computes an approximate minimizer in
polynomial time. We establish that for graphs with small average vertex degree,
these guarantees are minimax rate-optimal over classes of edge-sparse
signals. For spatially inhomogeneous graphs, we propose minimization of an
edge-weighted objective where each edge is weighted by its effective resistance
or another measure of its contribution to the graph's connectivity. We establish
minimax optimality of the resulting estimators over corresponding edge-weighted
sparsity classes. We show theoretically that these risk guarantees are not
always achieved by the estimator minimizing the $l_1$/total-variation
relaxation, and empirically that the $l_0$-based estimates are more accurate in 
high signal-to-noise settings.
\end{abstract}

\section{Introduction}
Let $G=(V,E)$ be a known (undirected) graph, with vertices $V:=\{1,\ldots,n\}$
and edge set $E$.
At each vertex $i \in \{1,\ldots,n\}$, an unknown signal value $\mu_{0,i}$ is
observed with noise:
\[Y_i=\mu_{0,i}+\eps_i.\]
For simplicity, we assume $\eps_1,\ldots,\eps_n \overset{iid}{\sim}
\Normal(0,\sigma^2)$ and $G$ is fully connected. This paper studies the
problem of estimating the true signal vector
$\mu_0:=(\mu_{0,1},\ldots,\mu_{0,n})$ from observed data $Y:=(Y_1,\ldots,Y_n)$,
when $\mu_0$ is (or is well-approximated by) a piecewise-constant signal over
$G$. Informally, this will mean that the set of edges $\{i,j\} \in E$ where
$\mu_{0,i} \neq \mu_{0,j}$ is a small subset of all edges.

Examples of this problem occur in a number of applications:
\begin{itemize}
\item {\bf Multiple changepoint detection.} The graph $G$ is a linear chain
with $n$ vertices and $n-1$ edges, which identifies a sequential order to the
observations. The signal $\mu_0$ is piecewise constant
in the sense $\mu_{0,i} \neq \mu_{0,i+1}$ for a small number of
changepoints $i$.
\item {\bf Image segmentation.} The graph $G$ is a 2-D (or 3-D)
lattice graph, and $\mu_0$ corresponds to the pixels (or voxels) of a digital
image. The assumption of piecewise-constancy implies that $\mu_0$ has regions
of approximately constant pixel value.
\item {\bf Anomaly identification in networks.} The graph $G$ represents a
network. The signal $\mu_0$ indicates locations of anomalous
clusters of vertices, for example representing individuals infected by a
disease or a computer virus. Piecewise-constancy of $\mu_0$ reflects the
assumption that the anomaly spreads along the network connections.
\end{itemize}
Early and pioneering works include
\cite{chernoffzacks,yao1984,barryhartigan} on multiple changepoint detection
and \cite{gemangeman,besag} on image segmentation. For general
graphs and networks, 
\cite{ariascastroetal2008,addarioberryetal,ariascastroetal2011,ariascastrogrimmett,sharpnacketallovasz,sharpnacketalwavelet}
studied related hypothesis testing problems, and 
\cite{hutterrigollet,wangetal,hernanmadridpadillaetal} also recently considered
estimation. We discuss some connections of our work to this existing
literature in Section \ref{subsec:literature}.

The focus of our paper is the method of ``$l_0$-edge-denoising'', which seeks to
estimate $\mu_0$ by the values $\mu \in \R^n$ minimizing the 
residual squared error $\frac{1}{2}\|Y-\mu\|^2$ plus a penalty for each
edge $\{i,j\} \in E$ where $\mu_i \neq \mu_j$. (Here and
throughout, $\|\cdot\|$ without a subscript denotes the standard Euclidean
norm.) More formally, this estimate minimizes the objective function
\begin{equation}\label{eq:objl0}\tag{L0}
F_0(\mu):=\frac{1}{2}\|Y-\mu\|^2+\lambda \|D\mu\|_0,
\hspace{0.1in} \|D\mu\|_0:=\sum_{\{i,j\} \in E} \1\{\mu_i \neq \mu_j\}.
\end{equation}
Here, $D:\R^n \to \R^E$ denotes a vertex-edge incidence matrix with entries
in $\{-1,0,1\}$ that maps $\mu \in \R^n$
to the vector of edge differences $(\mu_i-\mu_j)_{\{i,j\} \in E}$ (with an
arbitrary sign for each edge). The penalty term
$\|D\mu\|_0$ denotes the usual ``$l_0$-norm'' of $D\mu$, and $\lambda$ is
a user-specified tuning parameter that controls the magnitude of this penalty.

For reasons to be discussed, we will also consider procedures that seek to
minimize a more general weighted version of the above objective function,
\begin{equation}\label{eq:objw}\tag{W}
F_\w(\mu):=\frac{1}{2}\|Y-\mu\|^2+\lambda \|D\mu\|_\w,
\hspace{0.1in} \|D\mu\|_\w:=\sum_{\{i,j\} \in E} \w(i,j)\1\{\mu_i \neq \mu_j\},
\end{equation}
where $\w:E \to \R_+$ assigns a non-negative weight to each edge. This allows
possibly different penalty values to be applied to different edges of the graph.

The combinatorial nature of (\ref{eq:objl0}) and (\ref{eq:objw})
render exact minimization of these objectives computationally intractable for
general graphs. A primary purpose of this paper is to show, however, that
approximate minimization is sufficient to obtain statistically rate-optimal
guarantees. We study one such approximation algorithm by Boykov, Veksler, and
Zabih \cite{boykovetal}, suggest its use in minimizing (\ref{eq:objw}) for
applications involving inhomogeneous networks, and provide a unified analysis of
minimax squared-error risk for this estimation problem over edge-sparse signal
classes on general graphs.

We summarize our results as follows:
\begin{enumerate}[1.]
\item A polynomial-time algorithm using the $\alpha$-expansion
procedure of \cite{boykovetal} yields approximate minimizers $\hat{\mu}$ of
(\ref{eq:objl0}) and (\ref{eq:objw}) that achieve
the same statistical risk guarantees as the exact minimizers, up to constant
factors. Computation of $\hat{\mu}$ is reasonably efficient in practice and
yields good empirical signal recovery in our tested examples. In this sense,
inference based on minimizing (\ref{eq:objl0}) or (\ref{eq:objw})
is computationally tractable, even for large graphs.
\item For any graph $G$, the
estimate $\hat{\mu}$ (exactly or approximately) minimizing (\ref{eq:objl0})
with $\lambda \asymp \sigma^2\log |E|$ satisfies an ``edge-sparsity'' oracle
inequality
\begin{equation}\label{eq:oracleproperty}
\E[\|\hat{\mu}-\mu_0\|^2] \lesssim \inf_{\mu \in \R^n} \|\mu-\mu_0\|^2+
\sigma^2 \max(\|D\mu\|_0,1)\log|E|.
\end{equation}
This bounds the squared-error risk of $\hat{\mu}$ in terms of the
approximability of $\mu_0$ by any piecewise-constant signal $\mu$.
If it is known that $\|D\mu_0\|_0 \leq s$, then setting instead
$\lambda \asymp \sigma^2(1+\log\frac{|E|}{s})$ yields
\begin{equation}\label{eq:minimaxbound}
\E[\|\hat{\mu}-\mu_0\|^2] \lesssim \sigma^2s\left(1+\log\tfrac{|E|}{s}\right).
\end{equation}
The risk bound (\ref{eq:minimaxbound})
is rate-optimal in a minimax sense over the edge-sparse signal class
$\{\mu_0:\|D\mu_0\|_0 \leq s\}$ up to a multiplicative factor depending on
the mean vertex degree of $G$.
\item An alternative to minimizing (\ref{eq:objl0}) is to minimize its
$l_1$/total-variation relaxation,
\begin{equation}\label{eq:objTV}\tag{TV}
F_1(\mu):=\frac{1}{2}\|Y-\mu\|^2+\lambda \|D\mu\|_1,
\hspace{0.2in} \|D\mu\|_1:=\sum_{\{i,j\} \in E} |\mu_i-\mu_j|.
\end{equation}
One advantage of this approach is that (\ref{eq:objTV}) is convex
and can be exactly minimized in polynomial time. However, whether the risk
guarantees (\ref{eq:oracleproperty})
and (\ref{eq:minimaxbound}) hold for $\hat{\mu}$ minimizing
(\ref{eq:objTV}) depends on properties of the graph. In particular,
they do not hold for the linear chain graph, where instead
\begin{equation}\label{eq:minimaxTVlower}
\inf_{\lambda \geq 0} \sup_{\mu_0:\|D\mu_0\|_0 \leq s}
\E[\|\hat{\mu}^\lambda-\mu_0\|^2]\gtrsim \sigma^2(\log n)^{-5}\sqrt{sn},
\end{equation}
$\hat{\mu}^\lambda$ denoting the minimizer of (\ref{eq:objTV})
for each $\lambda$. This
result is connected to the ``slow rate'' of convergence in prediction risk
for the Lasso \cite{tibshirani,chenetal} in certain linear regression settings
with correlated predictors.
\item When $G$ has regions of differing connectivity,
$\|D\mu\|_0$ is not a spatially homogeneous
measure of complexity, and it may be more appropriate to minimize the
edge-weighted objective (\ref{eq:objw}) where $\w(i,j)$ measures the
contribution of edge $\{i,j\}$ to the connectivity of the graph. One
such weighting, inspired by the analyses in
\cite{sharpnacketallovasz,sharpnacketalwavelet}, weighs each edge by its
effective resistance when $G$ is viewed as an electrical resistor network.
In simulations on real networks, this weighting can yield
a substantial reduction in error over minimizers of the unweighted objective
(\ref{eq:objl0}). For general
weightings $\w:E \to \R_+$ belonging to the spanning tree polytope of $G$,
the guarantee (\ref{eq:minimaxbound}) holds
over the larger class $\{\mu_0:\|D\mu_0\|_\w \leq s\}$
for $\hat{\mu}$ minimizing (\ref{eq:objw}), and this guarantee is
minimax rate-optimal up to a graph-independent constant factor, for
all graphs.
\end{enumerate}

We provide a more detailed discussion of these results in
Sections \ref{sec:algorithm} to \ref{sec:weighting}.
Simulations comparing minimization of (\ref{eq:objl0}), (\ref{eq:objw}),
and (\ref{eq:objTV}) over several graphs are presented in 
Section \ref{sec:simulations}. Proofs are deferred to the appendices.

\subsection{Related work}\label{subsec:literature}
For changepoint problems where $G$ is the linear chain, (\ref{eq:objl0}) may
be exactly minimized by dynamic programming in quadratic time
\cite{augerlawrence,winklerliebscher,jacksonetal}. Pruning ideas may
reduce runtime to be near-linear in practice \cite{killicketal}.
Correct changepoint recovery and distributional properties of $\hat{\mu}$
minimizing
(\ref{eq:objl0}) were studied asymptotically in \cite{yao1988,yaoau} when the
number of true changepoints is fixed. Non-asymptotic risk bounds similar 
to (\ref{eq:oracleproperty}) and (\ref{eq:minimaxbound}) were established for
estimators minimizing similar objectives in
\cite{lebarbier,birgemassart2007}; we discuss this further below. 
Extension to the recovery of piecewise-constant functions over a continuous
interval was considered in \cite{boysenetal}. 

In image applications where $G$ is the 2-D lattice, (\ref{eq:objl0}) is closely
related to the Mumford-Shah functional \cite{mumfordshah} and Ising/Potts-model
energies for discrete Markov random fields \cite{gemangeman}. In the latter
discrete setting, where each $\mu_i$ is allowed to take value in a finite set
of ``labels'', a variety of algorithms seek to minimize (\ref{eq:objl0}) using
minimum s-t cuts on augmented graphs; see \cite{kolmogorovzabih} and the
contained references for a review. For an Ising model with only two distinct
labels, \cite{greigetal} showed that the exact minimizer may be computed via 
a single minimum s-t cut. For more than two distinct
labels, exact minimization of (\ref{eq:objl0}) is NP-hard \cite{boykovetal}.
We analyze a graph-cut algorithm from
\cite{boykovetal} that applies to more than two labels, where the exact
minimization property is replaced by an approximation guarantee. We show that
the deterministic guarantee of this algorithm implies rate-optimal statistical
risk bounds, for the 2-D lattice as well as for general graphs.

For an arbitrary graph $G$,
the estimators $\hat{\mu}$ exactly minimizing (\ref{eq:objl0}) and
(\ref{eq:objw}) are examples of general model-complexity penalized estimators
studied in \cite{barronetal,birgemassart2007}. The penalties we
impose may be smaller than those needed for the analyses of
\cite{barronetal,birgemassart2007} by
logarithmic factors, and we instead control the supremum of a certain Gaussian
process using an argument specialized to our graph-based problem.
A theoretical focus of \cite{barronetal,birgemassart2007} was on adaptive
attainment of minimax rates over families of models---for example, for
the linear chain graph, \cite{lebarbier,birgemassart2007} considered
penalties increasing but concave in the number
of changepoints, and the resulting estimates achieve the guarantee
(\ref{eq:minimaxbound}) simultaneously for all $s$. Instead of using such a
penalty, which poses additional computational challenges, we will
apply a data-driven procedure to choose $\lambda$, although we will
not study the adaptivity properties of the procedure in this paper.

The method of $l_0$-edge-denoising and the characterization
of signal complexity by $\|D\mu_0\|_0$ are ``nonparametric'' in the sense of
\cite{ariascastroetal2005,ariascastroetal2011}. This is in contrast to methods
that employ additional prior knowledge about $\mu_0$, for instance
that its constant regions belong to parametric classes of shapes
\cite{ariascastroetal2005}, are thick and blob-like in nature
\cite{ariascastroetal2011}, or have sufficiently smooth boundaries when $G$ is
embedded in a
Euclidean space \cite{korostelevtsybakov,donoho}. In this regard, our study is
more closely related to the hypothesis testing work of
\cite{ariascastrogrimmett,sharpnacketallovasz,sharpnacketalwavelet} in similar
nonparametric contexts. An advantage of this perspective is that the inference
algorithm is broadly applicable to general
graphs and networks, where appropriate notions of boundary smoothness or support
constraints for $\mu_0$ are less naturally defined. A disadvantage is
that such an approach may not be statistically optimal in more specialized
settings when such prior assumptions hold true.

A connection between this problem,
effective edge resistances, and graph spanning trees emerged in the analyses of
\cite{sharpnacketallovasz,sharpnacketalwavelet}.
In \cite{sharpnacketalwavelet}, a procedure was proposed to construct an
orthonormal wavelet basis over a spanning tree of $G$ and to perform inference 
by thresholding in this basis. Our proposal to minimize (\ref{eq:objw}) for
$\w:E \to \R_+$ in the
spanning tree polytope of $G$ may be viewed as a derandomization of this idea
when the spanning tree is chosen at random; we discuss this
connection in Remark \ref{remark:wavelettree}. Sampling edges by effective
resistances is also a popular method of graph sparsification
\cite{spielmanteng,spielmansrivastava}, and effective-resistance edge
weighting may be viewed as a derandomization of procedures such as in 
\cite{sadhanalaetalsparsification} that operate on a randomly sparsified graph.

There is a large body of literature on the $l_1$-relaxation (\ref{eq:objTV}).
This method and generalizations were suggested in different
contexts and guises for the linear chain graph in
\cite{landfriedman,mammenvandegeer,davieskovac,chenetal,tibshiranietal}
and also studied theoretically in
\cite{rinaldo,harchaouilevyleduc,dalalyanetal,linetal,guntuboyinaetal}.
For 2-D lattice graphs in image denoising, variants of (\ref{eq:objTV})
were proposed and studied in
\cite{rudinetal,chambollelions,goldsteinosher}. For more general graphs, this
method and generalizations have been studied in
\cite{hoefling,kovacsmith,tibshiranitaylor,sharpnacketalsparsistency,tanseyscott,hutterrigollet,wangetal,sadhanalaetal},
among others. In particular, \cite{chambolle,darbonsigelle,xinetal} developed
algorithms for minimizing (\ref{eq:objTV}) and related objectives also using
iterated graph cuts, although these algorithms yield exact solutions and are
different from the algorithm we study.
A body of theoretical work establishes that $\hat{\mu}$ minimizing
(\ref{eq:objTV}) is (or is nearly)
minimax rate-optimal over signal classes of bounded variation, 
$\{\mu_0:\|D\mu_0\|_1 \leq s\}$, for the linear chain graph
and higher-dimensional lattices
\cite{mammenvandegeer,sadhanalaetal,wangetal,hutterrigollet}. Several risk
bounds over the exact-sparsity classes $\{\mu_0:\|D\mu_0\|_0 \leq s\}$
that we consider were also established for the linear chain graph
in \cite{dalalyanetal,linetal,guntuboyinaetal} and for general graphs in
\cite{hutterrigollet,hernanmadridpadillaetal}; we discuss some of
these results in Section \ref{sec:TV}.
We believe that benefits of using effective resistance weighted
edge penalties may also apply to the $l_1$/TV setting, and we leave further
exploration of this to future work.

\subsection{Notation and conventions}
We assume throughout that $G$ is fully connected with $n \geq 3$
vertices. Theoretical results are non-asymptotic, in the sense that they are
valid for all finite $n$ and $s$ with universal
constants $C,c>0$ independent of $n$, $s$, and the graph $G$. 
For positive $a$ and $b$, we write informally $a \lesssim b$ if $a \leq Cb$ and
$a \asymp b$ if $ca \leq b \leq Ca$ for universal constants $C,c>0$ and all $n
\geq 3$.

For a vector $v$, $\|v\|:=(\sum_i v_i^2)^{1/2}$ is the Euclidean norm,
$\|v\|_0:=|\{i:v_i \neq 0\}|$ the ``$l_0$-norm'', $\|v\|_1:=\sum_i |v_i|$
the $l_1$-norm, and $\|v\|_\infty:=\max_i |v_i|$ the $l_\infty$-norm.
For vectors $v$ and $w$,
$\langle v,w \rangle=\sum_i v_iw_i$ is the Euclidean inner-product.

$\R_+$ denotes the non-negative reals.
For an edge weighting $\w:E \to \R_+$, $\w(i,j)$ is shorthand for $\w(\{i,j\})$,
and we denote $\w(E'):=\sum_{\{i,j\} \in E'} \w(i,j)$ for any edge subset
$E' \subseteq E$. For two edge weightings $\w,\r:E \to \R_+$, we write $\w \geq
\r$ if $\w(i,j) \geq \r(i,j)$ for all edges $\{i,j\} \in E$. For $v \in \R^E$, 
$\|v\|_\w:=\sum_{\{i,j\} \in E} \w(i,j)v_{\{i,j\}}$ denotes the
$l_0$-norm weighted by $\w$.

$\1\{\cdot\}$ denotes the indicator function, i.e.\ $\1\{\mathcal{E}\}=1$ if
condition $\mathcal{E}$ is true and 0 otherwise.

\section{Approximation algorithm}\label{sec:algorithm}
\begin{algorithm}[t]
\caption{Algorithm to compute a
$(\tau,\delta\Z)$-local-minimizer of (\ref{eq:objw})}
\label{alg:boykovouter}
\begin{algorithmic}[1]
\STATE{Let $\bar{Y}$, $Y_{\min}$, and $Y_{\max}$ be the mean, minimum,
and maximum of $Y$, rounded to $\delta\Z$.}
\STATE{Initialize $\hat{\mu} \in \R^n$ by setting $\hat{\mu}_i=\bar{Y}$
for all $i \in \{1,\ldots,n\}$.}
\LOOP
  \FOR{each $c \in \delta\Z \cap [Y_{\min},Y_{\max}]$}
     \STATE{Compute the best $\delta\Z$-expansion $\tilde{\mu}$ of $\hat{\mu}$
with new value $c$ using Algorithm \ref{alg:boykovinner}.}
     \STATE{If $F_\w(\tilde{\mu}) \leq F_\w(\hat{\mu})-\tau$, then
set $\hat{\mu}=\tilde{\mu}$.}
  \ENDFOR
  \STATE{If $\hat{\mu}$ was unchanged, then return $\hat{\mu}$.}
\ENDLOOP
\end{algorithmic}
\end{algorithm}
\begin{algorithm}[t]
\caption{$\alpha$-expansion sub-routine \cite{boykovetal}}
\label{alg:boykovinner}
\begin{algorithmic}[1]
     \STATE{Construct the following edge-weighted augmentation
$G_{c,\hat{\mu}}$ of $G$:}
\begin{ALC@g}
\STATE{Introduce a source vertex $\s$ and a sink vertex $\t$.}
\STATE{Connect $\s$ to each $i \in \{1,\ldots,n\}$ with weight
$\frac{1}{2}(Y_i-c)^2$.}
\STATE{Connect $\t$ to each $i \in \{1,\ldots,n\}$ with weight
$\frac{1}{2}(Y_i-\hat{\mu}_i)^2$ if
$\hat{\mu}_i \neq c$, or weight $\infty$ if $\hat{\mu}_i=c$.}
\FOR{each edge $\{i,j\} \in E$}
\IF{$\hat{\mu}_i=\hat{\mu}_j$}
\STATE{Assign weight $\lambda\w(i,j)\1\{\hat{\mu}_i \neq c\}$ to $\{i,j\}$.}
\ELSE
\STATE{Introduce a new vertex $\a_{i,j}$.}
\STATE{Replace edge $\{i,j\}$ by the three edges $\{i,\a_{i,j}\}$,
$\{j,\a_{i,j}\}$, and $\{\t,\a_{i,j}\}$, with weights
$\lambda\w(i,j)\1\{\hat{\mu}_i \neq c\}$,
$\lambda\w(i,j)\1\{\hat{\mu}_j \neq c\}$,
and $\lambda\w(i,j)$, respectively.}
\ENDIF
\ENDFOR
\end{ALC@g}
     \STATE{Find the minimum s-t cut $(S,T)$ of $G_{c,\hat{\mu}}$ such
that $\s \in S$, $\t \in T$.}
\STATE{For each vertex $i \in \{1,\ldots,n\}$, set $\tilde{\mu}_i=c$ if $i$
belongs to $T$ and $\tilde{\mu}_i=\hat{\mu}_i$ otherwise.}
\STATE{Return $\tilde{\mu}$.}
\end{algorithmic}
\end{algorithm}

As discussed in Section \ref{subsec:literature}, whether (\ref{eq:objl0}) and
(\ref{eq:objw}) may be minimized exactly in polynomial time depends on the graph
$G$. However, good approximation of the solution
is tractable for any graph. We review in this section one approach
that achieves such an approximation,
based on discretizing the range of values of the entries of $\mu$ and applying
the $\alpha$-expansion local move of \cite{boykovetal} for discrete Markov
random fields. We describe the algorithm for (\ref{eq:objw}),
as (\ref{eq:objl0}) is a special case.

The fundamental property of this algorithm will be that its output
is a $(\tau,\delta\Z)$-local-minimizer for the objective function
(\ref{eq:objw}), defined as follows:
\begin{definition}\label{def:localmin}
For $\delta>0$, denote by
\[\delta\Z:=\{\ldots,-3\delta,-2\delta,-\delta,0,
\delta,2\delta,3\delta,\ldots\}\]
the set of all integer multiples of $\delta$.
For any $\mu \in \R^n$, a {\bf $\pmb{\delta\Z}$-expansion} of
$\mu$ is any other vector $\tilde{\mu} \in \R^n$ such that there exists a
single value $c \in \delta\Z$ for which, for every $i \in \{1,\ldots,n\}$,
either $\tilde{\mu}_i=\mu_i$ or $\tilde{\mu}_i=c$.
For $\delta>0$ and $\tau \geq 0$,
a {\bf $\pmb{(\tau,\delta\Z)}$-local-minimizer} of (\ref{eq:objw}) is any
$\mu \in \R^n$ such that for every $\delta\Z$-expansion $\tilde{\mu}$ of $\mu$,
\[F_\w(\mu)-\tau \leq F_\w(\tilde{\mu}).\]
\end{definition}

More informally, a $\delta\Z$-expansion of $\mu$ can replace any subset of
vertex values by a single new value $c \in \delta\Z$,
and a $(\tau,\delta\Z)$-local-minimizer is
such that no further $\delta\Z$-expansion reduces the objective value by more
than $\tau$. This definition does not require $(\tau,\delta\Z)$-local-minimizers
to have all entries belonging to $\delta\Z$---hence, in
particular, a global minimizer of (\ref{eq:objw}) is also a
$(\tau,\delta\Z)$-local-minimizer for any $\delta>0$ and $\tau \geq 0$.
We define analogously $(\tau,\delta\Z$)-local-minimizers for (\ref{eq:objl0}).

The $\alpha$-expansion procedure of \cite{boykovetal} may be used to compute a
$(\tau,\delta\Z)$-local-minimizer efficiently with graph cuts. We review this
procedure and how we apply it to our problem in Algorithm \ref{alg:boykovouter}.
We will use a small discretization $\delta$ so as to
yield a good solution to the original continuous problem.

The following propositions verify that this algorithm returns a
$(\tau,\delta\Z)$-local-minimizer in polynomial time; proofs are contained in 
Appendix \ref{appendix:algorithm}.
\begin{proposition}\label{prop:runtime}
Algorithm \ref{alg:boykovouter}, using Edmonds-Karp or Dinic's algorithm for
solving minimum s-t cut, has worst-case
runtime $O(n|E|^3(Y_{\max}-Y_{\min})^3/(\delta\tau))$.
\end{proposition}
\begin{proposition}\label{prop:minimum}
Among all $\delta\Z$-expansions of $\hat{\mu}$ with new value $c$,
the vector $\tilde{\mu}$ returned by Algorithm \ref{alg:boykovinner} has
lowest objective value in (\ref{eq:objw}). The estimate $\hat{\mu}$
returned by Algorithm \ref{alg:boykovouter} is a
$(\tau,\delta\Z)$-local-minimizer of (\ref{eq:objw}).
\end{proposition}

In particular, if $Y_{\max}-Y_{\min}$, $1/\delta$, and
$1/\tau$ are polynomial in $n$, then Algorithm \ref{alg:boykovouter}
is polynomial-time in $n$. We will use the Boykov-Kolmogorov algorithm
\cite{boykovkolmogorov} instead of Edmonds-Karp or Dinic to solve minimum s-t
cut. This has slower worst-case runtime but is much faster in practice on our
tested examples. We have found Algorithm \ref{alg:boykovouter} to be fast in
practice, even with $\tau=0$, and we discuss empirical
runtime in Section \ref{subsec:runtime}.

The $l_2$ vertex-cost $(Y_i-\mu_i)^2$
and $l_0$ edge cost $\1\{\mu_i \neq \mu_j\}$ of (\ref{eq:objl0}) and
(\ref{eq:objw}) are not intrinsic to this
algorithm, and the same method may be applied to approximately minimize
\[F(\mu)=\sum_{i=1}^n c_i(Y_i,\mu_i)+\sum_{\{i,j\} \in E} c_{i,j}(\mu_i,\mu_j)\]
for any vertex cost functions $c_i$ and edge cost functions $c_{i,j}$ such that
each $c_{i,j}$ satisfies a triangle inequality. Thus the algorithm is 
easily applicable to other likelihood models and forms of edge penalties.

\section{Theoretical guarantees for $l_0$ denoising}\label{sec:l0}
In this section, we describe squared-error risk guarantees for
$\hat{\mu}$ (exactly or approximately) minimizing (\ref{eq:objl0}).
Although these results are corollaries of
those in Section \ref{sec:weighting} for the weighted objective
(\ref{eq:objw}), we state them here separately as they are simpler to understand
and also provide a basis for comparison with total-variation denoising
discussed in the next section. We defer discussion of the proofs to Section
\ref{sec:weighting}.

Recall Definition \ref{def:localmin} of $(\tau,\delta\Z)$-local-minimizers,
which include both the exact minimizer and the estimator
computed by Algorithm \ref{alg:boykovouter}.
A sparsity-oracle inequality for any
such minimizer holds when the penalty $\lambda$ in (\ref{eq:objl0})
is set to a ``universal'' level $C\sigma^2\log |E|$:
\begin{theorem}\label{thm:oraclel0}
Let $\delta \leq \sigma/\sqrt{n}$ and $\tau \leq \sigma^2$. For any $\eta>0$,
there exist constants $C_\eta,C_\eta'>0$ depending only on $\eta$ such that
if $\lambda \geq C_\eta\sigma^2\log |E|$ and $\hat{\mu}$ is any
$(\tau,\delta\Z)$-local-minimizer of (\ref{eq:objl0}), then
\begin{equation}\label{eq:oraclel0}
\E[\|\hat{\mu}-\mu_0\|^2] \leq \inf_{\mu \in \R^n}
(1+\eta)\|\mu-\mu_0\|^2+C_\eta'\lambda\max(\|D\mu\|_0,1).
\end{equation}
\end{theorem}
The upper bound in (\ref{eq:oraclel0}) trades off the edge-sparsity of $\mu$
and its approximation of the true signal $\mu_0$. Setting
$\lambda=C_\eta\sigma^2 \log |E|$ yields the guarantee (\ref{eq:oracleproperty})
described in the introduction. If $\mu_0$ is exactly edge-sparse with
$\|D\mu_0\|_0=s$, then evaluating (\ref{eq:oraclel0})
at $\mu=\mu_0$ yields a risk bound
of order $\sigma^2 s\log |E|$. When $s$ is known, we may obtain the tighter
guarantee (\ref{eq:minimaxbound}) by using a smaller penalty:
\begin{theorem}\label{thm:minimaxupperl0}
Let $\delta \leq \sigma/\sqrt{n}$ and $\tau \leq \sigma^2$.
There exist universal constants $C,C'>0$ such that for any $s \in [1,|E|]$,
if $\lambda \geq C\sigma^2(1+\log\frac{|E|}{s})$ and $\hat{\mu}$ is any
$(\tau,\delta\Z)$-local-minimizer of (\ref{eq:objl0}), then
\begin{equation}\label{eq:minimaxupperl0}
\sup_{\mu_0:\|D\mu_0\|_0 \leq s} \E[\|\hat{\mu}-\mu_0\|^2] \leq C'\lambda s.
\end{equation}
\end{theorem}

Theorems \ref{thm:oraclel0} and \ref{thm:minimaxupperl0} are analogous to
estimation guarantees in the sparse normal-means problem:
For estimating a signal $\mu_0 \in \R^n$ with at most $k$ nonzero entries,
asymptotically if $n \to \infty$ and $k/n \to 0$, then
\begin{equation}\label{eq:minimaxnormalmeans}
\inf_{\hat{\mu}} \sup_{\mu_0:\|\mu_0\| \leq k}
\E[\|\hat{\mu}-Y\|^2] \sim 2\sigma^2k\log \frac{n}{k}.
\end{equation}
This risk is achieved by $\hat{\mu}=\argmin_\mu \frac{1}{2}\|Y-\mu\|^2
+\lambda\|\mu\|_0$ for $\lambda=\sigma^2\log \frac{n}{k}$, corresponding to
entrywise hard-thresholding at level $\sqrt{2\lambda}$
\cite[Theorem 8.20]{johnstone}. Setting $\lambda=\sigma^2 \log n$
hard-thresholds instead at the universal level $\sqrt{2\sigma^2\log n}$, and
Lemma 1 of \cite{donohojohnstone} implies an oracle bound
\[\E[\|\hat{\mu}-\mu_0\|^2] \leq \inf_{\mu \in \R^n}
1.2\|\mu-\mu_0\|^2+\sigma^2(2\log n+1)(\|\mu\|_0+1)\]
for any true signal $\mu_0 \in \R^n$.

When there is an underlying graph $G$, the sparsity condition
$\|\mu_0\|_0 \leq k$ is a notion of
vertex sparsity, in contrast to our notion of edge-sparsity. 
The edge-sparsity of a ``typical'' piecewise-constant signal may be
graph-dependent---for example, if $G$ is a $K$-dimensional lattice graph with
side length $n^{1/K}$ and
$\mu_0$ consists of two constant pieces separated by a smooth boundary, then
$s \asymp n^{1-1/K}$. For such choices of $\mu_0$ and for $K \geq 2$,
the risk in (\ref{eq:minimaxupperl0}) grows polynomially in $n$ and does not
represent a parametric rate. On the other hand, vertex-sparse signals are also
edge-sparse for low-degree graphs. This containment may be used to show, when
$G$ has bounded average degree, that the above nonparametric rate is optimal
in a minimax sense over $\{\mu_0:\|D\mu_0\|_0 \leq s\}$:
\begin{theorem}\label{thm:minimaxlowerl0}
Suppose $G$ has average vertex degree $d$.
There exists a universal constant $c>0$ such that for any $s \in [4d,|E|]$,
\begin{equation}\label{eq:minimaxlowerl0}
\inf_{\hat{\mu}} \sup_{\mu_0:\|D\mu_0\|_0 \leq s} \E[\|\hat{\mu}-\mu_0\|^2]
\geq c\sigma^2\frac{s}{d}\left(1+\log \frac{|E|}{s}\right),
\end{equation}
where the infimum is taken over all possible estimators
$\hat{\mu}:=\hat{\mu}(Y)$.
\end{theorem}
\noindent When the average degree $d$ is not small, there is a gap
between (\ref{eq:minimaxupperl0}) and (\ref{eq:minimaxlowerl0}) of order $d$,
which we will discuss in Section \ref{sec:weighting}.

\begin{remark}\label{remark:smalls}
We assume $s \geq 4d$ for (\ref{eq:minimaxlowerl0}) so
that the result does not depend on the exact structure of near-minimum cuts
in $G$. For example, if vertices $\{1,\ldots,n-1\}$ are
connected in a single cycle and vertex $n$ is
connected to vertex 1 by a single edge, then for $s=1$, any $\mu_0$ with
$\|D\mu_0\|_0 \leq 1$ must be constant over vertices
$\{1,\ldots,n-1\}$ and take a possibly different value on vertex $n$. The
minimax risk over this class is then $2\sigma^2$, rather than
order $\sigma^2\log n$. Considering the graph tensor
product of this example with the complete graph on $d$ vertices, a
similar argument shows that a general lower bound must restrict to
$s \geq cd$ for some small constant $c>0$.
\end{remark}

While our main focus is estimation, let us
state a result relevant to testing:
\begin{theorem}\label{thm:nulll0}
Let $\delta \leq \sigma/\sqrt{n}$ and $\tau \leq \sigma^2$,
and suppose $\mu_0$ is constant over $G$. There
exist universal constants $C,C'>0$ such that if $\lambda \geq C\sigma^2\log |E|$
and $\hat{\mu}$ is any $(\tau,\delta\Z)$-local-minimizer of (\ref{eq:objl0}),
then
\[\P[\hat{\mu} \text{ is constant over } G] \geq 1-C'n^{-3}.\]
\end{theorem}
This implies that we may test the null hypothesis
\begin{equation}\label{eq:null}
H_0:\mu_0 \text{ is constant}
\end{equation}
by setting $\lambda \asymp \sigma^2\log |E|$ and rejecting $H_0$ if
$\hat{\mu}$ is not constant.
Denoting by $P^\perp$ the orthogonal projection onto the space orthogonal to
the all-1's vector, since
\[\min_{\mu:\mu \text{ is constant}} \|\mu-\mu_0\|^2=\|P^\perp \mu_0\|^2,\]
the risk bound (\ref{eq:minimaxupperl0}) (or more precisely, Lemma
\ref{lemma:oracleinequalityprob}(b) which establishes an analogous bound in
probability) implies that this test can distinguish a non-constant alternative
$\mu_0$ with probability approaching 1 as long as
$\|P^\perp \mu_0\|^2 \geq C\sigma^2\|D\mu_0\|_0 \log |E|$, for a universal
constant $C>0$.

\section{Comparison with $l_1$/total-variation denoising}
\label{sec:TV}
We compare the guarantees of the preceding section with those attainable by
$\hat{\mu}$ minimizing (\ref{eq:objTV}). Theoretical risk bounds for the
TV-penalized estimator have been established for both piecewise-constant classes
$\{\mu_0:\|D\mu_0\|_0 \leq s\}$ and bounded-variation classes
$\{\mu_0:\|D\mu_0\|_1 \leq s\}$, and we focus our comparison on the
former. We will empirically explore in Section \ref{sec:simulations} some
trade-offs between the $l_0$ and TV approaches for signals that are both
piecewise-constant and of small total-variation norm.

One general risk bound for $\hat{\mu}$ minimizing (\ref{eq:objTV})
was established in \cite{hutterrigollet}. For an arbitrary graph $G$, let
$D:\R^n \to \R^E$ be its vertex-edge incidence matrix, $S=D^\dagger$ the
Moore-Penrose pseudo-inverse of $D$,
and $\rho$ the maximum Euclidean norm of any column
of $S$. Theorem 2 of \cite{hutterrigollet}
implies, for the estimator $\hat{\mu}$
minimizing (\ref{eq:objTV}) with the choice $\lambda=\sigma\rho\sqrt{2(1+\log
(|E|/\delta))}$, and for any $\mu \in \R^n$, with
probability at least $1-2\delta$,
\begin{equation}\label{eq:oracleTV}
\|\hat{\mu}-\mu_0\|^2 \leq 
\|\mu-\mu_0\|^2+8\sigma^2\left(\frac{\rho^2\|D\mu\|_0}{\kappa^2}\left(1+\log
\frac{|E|}{\delta}\right)+\log \frac{e}{\delta}\right),
\end{equation}
where $\kappa$ is a compatibility constant bounded
as $\kappa^{-2} \leq 4\min(d,\|D\mu\|_0)$ and $d$ is the mean vertex degree of
$G$. (The result of \cite{hutterrigollet} is more general, involving
both $\|D\mu\|_1$ and $\|D\mu\|_0$, and we have specialized to the
``pure-$l_0$'' setting.)

An important difference between this result and Theorem \ref{thm:oraclel0}
is the appearance of
$\rho^2$, which is graph-dependent. Assuming $G$ has small average
degree $d$, the above guarantee is similar to Theorem \ref{thm:oraclel0}
if $\rho^2$ is small. It is shown in \cite{hutterrigollet}
that $\rho^2 \lesssim 1$ for 3-D (and higher-dimensional) lattice graphs and
$\rho^2 \lesssim \log n$ for 2-D lattice graphs, indicating that
$\hat{\mu}$ is nearly rate-optimal over $\{\mu_0:\|D\mu_0\|_0 \leq s\}$ for
these graphs. However, for example, when $G$ is the linear chain,
$\rho^2 \asymp n$ and the bound (\ref{eq:oracleTV}) is larger than
those of the preceding section by a factor of $n$.

More specialized analyses were performed for the linear chain
in \cite{dalalyanetal,linetal}, where sharper results were obtained that depend
on the minimum spacing $\Delta(\mu_0)$
between two changepoints of $\mu_0$. More precisely,
denoting by $1 \leq i_1<\ldots<i_s<n$ the values $i$ for which $\mu_{0,i}
\neq \mu_{0,i+1}$ and letting $i_0:=0$ and $i_{s+1}:=n$, define
$\Delta(\mu_0):=\min_{0 \leq r \leq s} i_{r+1}-i_r$. Then Theorem 4 of
\cite{linetal} shows, if $\|D\mu_0\|_0=s$ and
$\lambda=\sigma(n\Delta(\mu_0))^{1/4}$, then
\[\E[\|\hat{\mu}-\mu_0\|^2] \lesssim \sigma^2s\left((\log s+\log\log n)\log n
+\sqrt{n/\Delta(\mu_0)}\right).\]
If $\Delta(\mu_0) \gtrsim n/(s+1)$ so that changepoints are
nearly equally-spaced, then this bound is of order $s^{3/2}$ times logarithmic
factors, and furthermore this has been improved to the optimal bound
$\E[\|\hat{\mu}-\mu_0\|^2] \lesssim \sigma^2s\log(1+n/s)$ in
\cite{guntuboyinaetal}. However,
if $\Delta(\mu_0) \lesssim n^\alpha$ for any $\alpha<1$, then the above
bound differs from the guarantee of Theorem \ref{thm:minimaxupperl0} by a
factor of roughly $n^{(1-\alpha)/2}$, and in the worst case this
suboptimality is of order $\sqrt{n}$.

It has been conjectured, for example in Remark 3 of \cite{linetal} and
Remark 2.3 of \cite{guntuboyinaetal}, that this suboptimality is not an
artifact of the theoretical analysis, but rather that the TV-penalized estimate
$\hat{\mu}$ exhibits a slower rate of convergence when the equal spacing
condition $\Delta(\mu_0) \gtrsim n/(s+1)$ is not met. We provide in this
section a theoretical validation of this conjecture; proofs are given in
Appendix \ref{appendix:TV}.

First, suppose the true signal $\mu_0$ is constant and equal to zero:
\begin{theorem}\label{thm:TVnull}
Let $G$ be the linear chain graph with $n$ vertices, and
suppose $\mu_0=0$. There exists a constant $c>0$
such that for any fixed $\lambda \in [0,\sigma\sqrt{cn/\log n}]$,
if $\hat{\mu}$ is the minimizer of (\ref{eq:objTV}), then the following hold:
\begin{enumerate}[(a)]
\item For some constants $C,c'>0$, letting $\hat{k}:=\|D\hat{\mu}\|_0+1$ be 
the number of constant intervals of $\hat{\mu}$,
\[\P\left[\hat{k}>\frac{cn}{\max(\lambda^2/\sigma^2,1)\log n}\right]
 \geq 1-Ce^{-c'n/\max(\lambda^2/\sigma^2,1)}.\]
\item For some constant $c'>0$, the squared-error risk of $\hat{\mu}$ satisfies
\[\E[\|\hat{\mu}\|^2] \geq \frac{c'\sigma^2n}{\max(\lambda^2/\sigma^2,1)
(\log n)^4}.\]
\end{enumerate}
\end{theorem}
Hence if $\mu_0=0$ and
$\lambda \asymp \sigma n^\alpha$ for any $\alpha<1/2$, then the number of 
changepoints and the squared-error risk of the TV-penalized estimator
$\hat{\mu}$ are (up to logarithmic factors) at least of order $n^{1-2\alpha}$
and $\sigma^2 n^{1-2\alpha}$, respectively.
As a consequence, we obtain the following lower bound in a minimax sense:
\begin{theorem}\label{thm:TValternative}
Let $G$ be the linear chain graph with $n$ vertices, and let
$\Delta(\mu_0)$ denote the minimum distance between two changepoints in $\mu_0$.
For each fixed $\lambda \geq 0$, let $\hat{\mu}^\lambda$ denote the minimizer
of (\ref{eq:objTV}) for this $\lambda$. Then there exists a constant
$c>0$ such that for any $s \in [2,n-1]$ and $\Delta \leq n/(s+1)$,
\[\inf_{\lambda \geq 0} \sup_{\mu_0:\|D\mu_0\|_0 \leq s,\,\Delta(\mu_0) \geq 
\Delta} \E[\|\hat{\mu}^\lambda-\mu_0\|^2] \geq
\frac{c\sigma^2}{(\log n)^5}\sqrt{\frac{sn}{\Delta}}.\]
\end{theorem}
\noindent In particular, setting $\Delta=1$ removes restrictions on the
minimum spacing between changepoints and yields (\ref{eq:minimaxTVlower}) stated
in the introduction.

Theorem \ref{thm:TValternative} may be re-interpreted in the context of the
Lasso estimate for sparse linear regression: Setting $\beta_0:=D\mu_0$ and
\[X:=D^\dagger=\begin{pmatrix}
-\frac{n-1}{n} & -\frac{n-2}{n} & -\frac{n-3}{n} & \cdots & -\frac{1}{n} \\
\frac{1}{n} & -\frac{n-2}{n} & -\frac{n-3}{n} & \cdots & -\frac{1}{n} \\
\frac{1}{n} & \frac{2}{n} & -\frac{n-3}{n} & \cdots & -\frac{1}{n} \\
\vdots & \vdots & \vdots & \ddots & \vdots \\
\frac{1}{n} & \frac{2}{n} & \frac{3}{n} & \cdots & \frac{n-1}{n}
\end{pmatrix} \in \R^{n \times (n-1)},\]
minimizing (\ref{eq:objTV}) is equivalent to minimizing the Lasso objective
\[\frac{1}{2}\|\tilde{Y}-X\beta\|^2+\lambda \|\beta\|_1\]
over $\beta \in \R^{n-1}$,
where $\tilde{Y}=(Y_1-\bar{Y},\ldots,Y_n-\bar{Y})$ denotes $Y$ centered by its
mean. The maximum column norm of $X$ is $\sqrt{n/4}$, the
error $\|\hat{\mu}-\mu_0\|^2$ corresponds to $n$ times the
``prediction loss'' $n^{-1}\|X\hat{\beta}-X\beta_0\|^2$, and in this context
Theorem \ref{thm:TValternative} (with $\Delta=1$) implies
\[\inf_{\lambda \geq 0} \sup_{\beta_0:\|\beta_0\|_0 \leq s}
\E\left[\frac{1}{n}\|X\hat{\beta}^\lambda-X\beta_0\|^2\right]
\geq \frac{c\sigma^2}{(\log n)^5} \sqrt{\frac{s}{n}}.\]
Hence the minimax prediction risk for the Lasso estimate over
$\{\beta_0:\|\beta_0\|_0 \leq s\}$
decays essentially no faster than order $n^{-1/2}$. This is in contrast to the
faster rate of $n^{-1}$ that is achievable when $X$ has well-behaved restricted
eigenvalue constants (see, for example,
\cite{bickeletal,vandegeerbuhlmann} and the references contained therein).

More generally, for any connected graph $G$, noting that $D^\dagger \in \R^{n
\times E}$ is of rank $n-1$ with range orthogonal to the all-1's
vector, minimizing (\ref{eq:objTV}) is equivalent to minimizing
\[\frac{1}{2}\|\tilde{Y}-D^\dagger \beta\|^2+\lambda \|\beta\|_1
\text{ subject to } \beta \in \operatorname{range}(D),\]
where $\operatorname{range}(D)$ denotes the column span
of $D$ in $\R^E$. The results of the two preceding sections imply that whenever
$G$ has small average degree, the ``fast'' optimal rate for
prediction risk over the class
$\{\beta_0 \in \operatorname{range}(D):\|\beta_0\|_0 \leq s\}$ for the above
problem is attainable in polynomial time,
even if it is not achieved by the $l_1$ estimator. This
may be contrasted with the negative results of \cite{zhangetal1,zhangetal2},
which show that there exist adversarial design matrices $X$ for sparse
regression where such fast rates are not achieved by a broad class of
$M$-estimators or by any polynomial-time algorithm returning an $s$-sparse
output.

\section{Edge-weighting for inhomogeneous graphs}\label{sec:weighting}
In this section, we generalize the results of Section \ref{sec:l0} by
considering (exact or approximate) minimizers $\hat{\mu}$ of the edge-weighted
objective (\ref{eq:objw}). Proofs are contained in Appendix
\ref{appendix:proofs}, with a brief summary of proof ideas at the end of this
section.

We motivate our discussion by the following
example, which examines the factor-$d$ gap between the upper and lower
bounds of (\ref{eq:minimaxupperl0}) and (\ref{eq:minimaxlowerl0}):
\begin{example}\label{ex:tadpole}
Let $G$ be the complete graph on $n$ vertices. Then the
average vertex degree of $G$ is $d=n-1$, and (\ref{eq:minimaxlowerl0}) implies
\begin{equation}\label{eq:lowerboundtadpole}
\inf_{\hat{\mu}} \sup_{\mu_0:\|D\mu_0\|_0 \leq s}
\E[\|\hat{\mu}-\mu_0\|^2] \gtrsim \sigma^2
\frac{s}{n}\left(1+\log \frac{|E|}{s}\right).
\end{equation}
This lower bound is in fact tight, and the upper bound of
(\ref{eq:minimaxupperl0}) is loose by a factor of $n$:
Theorem \ref{thm:minimaxw} below will imply that setting $\lambda \asymp
\frac{\sigma^2}{n}(1+\log \frac{|E|}{s})$ in (\ref{eq:objl0}) achieves the above
level of risk, when $G$ is the complete graph.

On the other hand, let $G$ be a ``tadpole'' graph consisting of a linear chain
of $n/2$ vertices with one endpoint connected by an edge to a clique of
$n/2$ remaining vertices. The average vertex degree of $G$ in this case is
$d=(n+1)/2$, so a direct application of (\ref{eq:minimaxlowerl0}) still yields
(\ref{eq:lowerboundtadpole}). However,
by restricting to the sub-class of signals that take a constant value on the
$n/2$-clique, it is clear that the minimax risk over
$\{\mu_0:\|D\mu_0\|_0 \leq s\}$ is at least that of estimating the signal over
only the linear chain portion of $G$ with $n/2$ vertices.
The lower bound (\ref{eq:minimaxlowerl0}) applied to only this sub-graph
implies that in this case, the upper bound (\ref{eq:minimaxupperl0}) is tight
up to a constant factor, and the lower bound (\ref{eq:lowerboundtadpole}) is
loose by a factor of $1/n$.
\end{example}

This example highlights the problem that the complexity measure $\|D\mu_0\|_0$
is not necessarily spatially homogeneous over $G$. For example, when $G$ is the
tadpole graph, a signal $\mu_0$ that is constant over all but one vertex
belonging to the $n/2$-clique has
$\|D\mu_0\|_0=n/2-1$, but a signal taking a different value at each of the
vertices of the linear chain also has $\|D\mu_0\|_0=n/2-1$. A theoretical
consequence is that the minimax risk over $\{\mu_0:\|D\mu_0\|_0 \leq s\}$ is
controlled by the least well-connected portion of the graph. A practical
consequence is that any choice of $\lambda$ will either oversmooth the signal
over the $n/2$-clique or undersmooth the signal over the linear chain, and no
single choice of $\lambda$ leads to good signal recovery in both of these
regions.

While the tadpole graph is an extreme example, the same phenomenon arises in
any graph with regions of varying connectivity. In such applications, we propose
to consider the weighted objective function (\ref{eq:objw}) where each edge is
weighted by a measure of its contribution to the connectivity of $G$. We
believe both that minimizing this weighted objective is usually a more 
reasonable procedure in practice and that the value
$\|D\mu_0\|_\w$ provides a better indication of the complexity of the
piecewise-constant signal $\mu_0$.

One specific weighting $\w:E \to \R_+$ that implements this idea is to weigh
each edge by its effective resistance.
\begin{definition}\label{def:resistance}
Let $G$ be a connected graph and $\{i,j\}$ an edge in $G$. The {\bf effective
resistance} $\r(i,j)$ of this edge has the following four equivalent
definitions:
\begin{enumerate}
\item $\r(i,j)$ is the effective electrical resistance measured across vertices
$i$ and $j$ when $G$ represents an electrical network where each edge is a
resistor with resistance 1.
\item Let $L$ be the (unweighted) Laplacian matrix of $G$, $L^\dagger$ the
pseudo-inverse of $L$, and $e_i$ the basis vector with $i^\text{th}$
entry 1 and remaining entries 0. Then $\r(i,j)=(e_i-e_j)L^\dagger(e_i-e_j)$.
\item Consider a simple random walk on $G$ starting at vertex $i$, and let $t$
be the number of steps taken to reach vertex $j$ and then return to vertex $i$
for the first time. Then $\r(i,j)=\frac{1}{2|E|}\E[t]$.
\item Let $T$ be (the edges of) a random spanning tree of $G$ chosen uniformly
from the set of all spanning trees of $G$. Then $\r(i,j)=\P[\{i,j\} \in T]$.
\end{enumerate}
\end{definition}
For verification of the equivalence of these definitions, see
\cite{lovasz,ghoshetal}. In practice, $\r(i,j)$ may be computed via the second
characterization using fast Laplacian solvers \cite{spielmanteng,livnebrandt}.

The fourth characterization above describes one sense
in which $\r(i,j)$ measures the ``contribution'' of edge $\{i,j\}$ to the
connectivity of $G$: For example, if removing $\{i,j\}$ breaks $G$ into two 
disconnected components, then
every spanning tree $T$ of $G$ must contain $\{i,j\}$, so $\r(i,j)=1$.
Conversely, if there are many short alternative paths from $i$ to $j$ not using
edge $\{i,j\}$, then $\r(i,j)$ is much smaller than 1.

More generally, the contribution of each edge to the graph
connectivity may be measured by any weighting belonging to the spanning tree
polytope of $G$.
\begin{definition}
A weighting $\w:E \to \R_+$ is a \emph{spanning tree weighting} if there exists
a spanning tree $T$ of $G$ such that $\w(i,j)=1$ if $\{i,j\} \in T$ and
$\w(i,j)=0$ otherwise. The {\bf spanning tree polytope} $\ST(G)$ is the convex
hull of all spanning tree weightings.
\end{definition}
(With a slight abuse of notation, we will henceforth denote general
weightings by $\w$ and any weighting in $\ST(G)$ by $\r$.)
Thus, if $\r \in \ST(G)$, then there exist spanning trees $T_1,\ldots,T_M$ of
$G$ and $\lambda_1,\ldots,\lambda_M>0$ with $\sum_{m=1}^M
\lambda_m=1$ such that
\[\r(i,j)=\sum_{m=1}^M \lambda_m \1\{\{i,j\} \in T_m\}\]
for every edge $\{i,j\} \in E$. The weighting $\r$ is thus identified with the
probability distribution of a random spanning tree $T$, where
$T=T_m$ with probability $\lambda_m$. This distribution satisfies the property,
for all $\{i,j\} \in E$,
\[\r(i,j)=\P[\{i,j\} \in T].\]
For any subset of edges $E' \subseteq E$ and weighting $\w:E \to \R_+$, let us
denote $\w(E')=\sum_{\{i,j\} \in E'} \w(i,j)$ as the total weight of these
edges. Then the above implies, for the random spanning tree $T$ associated to
$\r \in \ST(G)$,
\begin{equation}\label{eq:expectedtreeedges}
\r(E')=\E[|E' \cap T|].
\end{equation}
The effective resistance weighting of Definition \ref{def:resistance}
corresponds to the uniform distribution for $T$.

The results below describe the squared-error risk of the
estimator $\hat{\mu}$ that (exactly or approximately)
minimizes (\ref{eq:objw}) for any
edge-weighting $\w:E \to \R_+$. We derive, for all graphs $G$,
minimax upper and lower bounds on this risk over the class
$\{\mu_0:\|D\mu_0\|_\w \leq s\}$. The tightness of these bounds will depend on
how close $\w$ is to the spanning tree polytope $\ST(G)$---for effective
resistance weighting, or more generally for any $\r \in \ST(G)$, these
bounds are tight up to a universal constant factor independent of the graph.

Let us make a remark about scaling, which is important for the interpretation of
the below results: As rescaling $\w$ by $c>0$ and $\lambda$ by $1/c$
leads to the same penalty in (\ref{eq:objw}), we will state all of our results,
for simplicity and without loss of generality, under a scaling such that
$\w \geq \r$ for some $\r \in \ST(G)$, meaning $\w(i,j) \geq \r(i,j)$ for every
edge. For any $\w$ (where $G$ remains connected by edges with positive weight),
there is a smallest constant $c>0$ for which this property
holds for $\tilde{\w}=c\w$; the below results yield the tightest risk bounds
when applied to $\tilde{\w}$ scaled in this way. Whenever $\w \geq \r$ for some
$\r \in \ST(G)$, (\ref{eq:expectedtreeedges}) implies that the total
weight of all edges satisfies
\begin{equation}\label{eq:wE}
\w(E) \geq \r(E)=n-1,
\end{equation}
since every spanning tree has $n-1$ edges.
The ratio $\w(E)/(n-1)$ provides a measure of the distance of $\w$ to $\ST(G)$.
Furthermore, if $E'$ is any subset of edges whose removal 
disconnects $G$ into $k+1$ connected components,
then every spanning tree contains at least $k$ edges
of $E'$, so $\w(E') \geq \r(E') \geq k$. In particular,
\begin{equation}\label{eq:wcut}
\mu \in \R^n \text{ has } k+1 \text{ distinct values} \Rightarrow
\|D\mu\|_\w \geq k.
\end{equation}

The following results generalize Theorems \ref{thm:oraclel0},
\ref{thm:minimaxupperl0}, and \ref{thm:minimaxlowerl0}:
\begin{theorem}\label{thm:oraclew}
Let $\w:E \to \R_+$ be such that $\w \geq \r$ for some $\r \in \ST(G)$, and let
$\delta \leq \sigma/\sqrt{n}$ and $\tau \leq \sigma^2$. For any $\eta>0$,
there exist constants $C_\eta,C_\eta'>0$ depending only on $\eta$ such that if
$\lambda \geq C_\eta\sigma^2\log \w(E)$ and $\hat{\mu}$ is any
$(\tau,\delta\Z)$-local-minimizer of (\ref{eq:objw}), then
\begin{equation}\label{eq:oraclew}
\E[\|\hat{\mu}-\mu_0\|^2] \leq \inf_\mu
(1+\eta)\|\mu-\mu_0\|^2+C_\eta'\lambda\max(\|D\mu\|_\w,1).
\end{equation}
\end{theorem}
\begin{theorem}\label{thm:minimaxw}
Let $\w:E \to \R_+$ be such that $\w \geq \r$ for some $\r \in \ST(G)$, and let
$\delta \leq \sigma/\sqrt{n}$ and $\tau \leq \sigma^2$.
There exist universal constants $C,C'>0$ such that
for any $s \in [1,\w(E)]$, if $\lambda \geq C\sigma^2(1+\log\frac{\w(E)}{s})$
and $\hat{\mu}$ is any $(\tau,\delta\Z)$-local-minimizer of (\ref{eq:objw}),
then
\begin{equation}\label{eq:minimaxupperw}
\sup_{\mu_0:\|D\mu_0\|_\w \leq s} \E[\|\hat{\mu}-\mu_0\|^2]
\leq C'\lambda s.
\end{equation}
Conversely, there exists a universal constant $c>0$ such that for any
$s \in [\frac{4\w(E)}{n},\w(E)]$,
\[\inf_{\hat{\mu}} \sup_{\mu_0:\|D\mu_0\|_\w \leq s}
\E[\|\hat{\mu}-\mu_0\|^2] \geq c\sigma^2
s\frac{n}{\w(E)}\left(1+\log \frac{\w(E)}{s}\right),\]
where the infimum is taken over all possible estimators
$\hat{\mu}:=\hat{\mu}(Y)$.
\end{theorem}
\noindent The restriction to $s \gtrsim \w(E)/n$ in the
lower bound is necessary for generality of the result to all graphs $G$,
for the same reason as in Remark \ref{remark:smalls}.

The minimax upper and lower bounds above differ by the factor
$n/\w(E)$. Recall from (\ref{eq:wE}) that $\w(E) \geq n-1$, with
$\w(E)=n-1$ precisely when $\w \in \ST(G)$. Hence the above immediately implies
the following corollary:
\begin{corollary}\label{cor:minimaxw}
If $\w=\r$ where $\r(i,j)$ is the effective resistance of each edge $\{i,j\}$,
or more generally where $\r \in \ST(G)$, then for any $s \in [4,n-1]$,
\begin{equation}\label{eq:minimaxSTG}
\inf_{\hat{\mu}} \sup_{\mu_0:\|D\mu_0\|_\w \leq s} \E[\|\hat{\mu}-\mu_0\|^2]
\asymp \sigma^2 s\left(1+\log \frac{n}{s}\right).
\end{equation}
\end{corollary}

\begin{remark}\label{remark:wavelettree}
One may compare (\ref{eq:minimaxupperw}) with a guarantee achieved by the
wavelet spanning tree method of \cite{sharpnacketalwavelet}: In this method,
for a fixed spanning tree $T$ of $G$, an orthonormal basis of Haar-like wavelet
functions is constructed over $T$ such that a signal $\mu_0$ cutting
$s$ edges of $T$ has a representation of sparsity $s(\log
d_{\max}(T))(\log n)$ in this basis, where $d_{\max}(T)$ is the maximal vertex
degree of $T$. The corresponding wavelet thresholding estimator then satisfies
\[\E[\|\hat{\mu}-\mu_0\|^2] \lesssim \sigma^2s(\log d_{\max}(T))(\log n)^2.\]
If $T$ is chosen at random from the spanning tree distribution corresponding to
any weighting $\w \in \ST(G)$, then bounding $d_{\max}(T) \leq d_{\max}(G)$
and averaging over the random choice of $T$ yields
\[\E[\|\hat{\mu}-\mu_0\|^2] \lesssim \sigma^2(\log d_{\max}(G))
(\log n)^2\|D\mu_0\|_\w,\]
which agrees with (\ref{eq:minimaxupperw}) up to extra logarithmic factors.
Whereas this defines a randomized
algorithm and the above risk is averaged also over the algorithm execution,
minimizing (\ref{eq:objw}) for $\w \in \ST(G)$ directly penalizes the number of
edges cut by $\hat{\mu}$ in the average spanning tree, and thus may be
interpreted loosely as a derandomization of this wavelet approach.
\end{remark}

Finally, we state a result of relevance to testing the null
hypothesis (\ref{eq:null}):
\begin{theorem}\label{thm:nullw}
Let $\w:E \to \R_+$ be such that $\w \geq \r$ for some $\r \in \ST(G)$,
and let $\delta \leq \sigma/\sqrt{n}$ and $\tau \leq \sigma^2$. There exist
universal constants $C,C'>0$ such that if $\mu_0$ is constant over $G$,
$\lambda \geq C\sigma^2\log \w(E)$, and $\hat{\mu}$ is
any $(\tau,\delta\Z)$-local-minimizer of (\ref{eq:objw}), then
\[\P[\hat{\mu} \text{ is constant over } G] \geq 1-C'n^{-3}.\]
\end{theorem}
Thus we may test $H_0$ in (\ref{eq:null}) by setting
$\lambda \asymp \sigma^2\log \w(E)$ and rejecting $H_0$ if $\hat{\mu}$
minimizing (\ref{eq:objw}) is not constant. Denoting by $P^\perp$ the projection
orthogonal to the all-1's vector, the risk bound (\ref{eq:minimaxupperw})
(or more precisely, the probability guarantee of Lemma
\ref{lemma:oracleinequalityprob}(b))
implies that this test can distinguish a non-constant alternative $\mu_0$ with
probability approaching 1 as long as $\|P^\perp \mu_0\|^2 \geq C\sigma^2
\|D\mu_0\|_\w \log \w(E)$, for a universal constant $C>0$.
When $\w:E \to \R_+$ is the effective resistance
weighting, this recovers a similar detection threshold as established for
the tests in \cite{sharpnacketalwavelet,sharpnacketallovasz}.

In the case of uniform edge weights $\w \equiv 1$, it is clear that $\w(E)=|E|$
and $\w \geq \r$ for all $\r \in \ST(G)$. Then Theorems \ref{thm:oraclel0},
\ref{thm:minimaxupperl0}, \ref{thm:minimaxlowerl0}, and \ref{thm:nulll0}
follow directly by specializing these
results. If there exists $\r \in \ST(G)$ such that $\r(i,j)<1$
for every edge, then the results of Section \ref{sec:l0} are trivially 
strengthened by rescaling $\lambda$ by $\max_{\{i,j\} \in E} \r(i,j)$. For
example, if $G$ is the complete graph, then every edge has effective resistance 
$\r(i,j)=2/n$, and Theorems \ref{thm:oraclew} and \ref{thm:minimaxw} imply that 
$\lambda$ may in fact be set to $Cn^{-1}\log |E|$ and $Cn^{-1}\log |E|/s$ in
Theorems \ref{thm:oraclel0} and \ref{thm:minimaxupperl0} respectively, as
claimed in Example \ref{ex:tadpole}.

We prove Theorems \ref{thm:oraclew}, \ref{thm:minimaxw}, and \ref{thm:nullw}
in Appendix \ref{appendix:proofs}. The upper bound in Theorem \ref{thm:minimaxw}
uses the idea of \cite[Theorem 3.3]{karger} for bounding the number of small
graph cuts by controlling the
number of cut edges in a given spanning tree. We apply this idea in
Lemma \ref{lemma:badcut}, controlling a supremum over all small cuts by
selecting a random spanning tree according to the weighting $\r$ and taking a
union bound over cuts of this tree. In conjunction with a Chernoff bound and a
standard Cauchy-Schwarz argument,
this establishes (\ref{eq:oraclew}) and (\ref{eq:minimaxupperw}) for
the exact minimizer of (\ref{eq:objw}) with high probability. We obtain
bounds in expectation using Holder's inequality to control the risk on the
complementary low-probability event. The extension to approximate minimizers
uses the factor-2 approximation guarantee for the alpha-expansion
algorithm established in \cite{boykovetal}. However, whereas the optimal
objective value for (\ref{eq:objw}) is usually dominated by the squared-error
term, we verify in Lemma \ref{lemma:factor2approx} that the approximation
factor applies not to this term but only to the $l_0$ penalty, and it holds not
only with respect to the global minimizer of (\ref{eq:objw})
but also with respect to any candidate vector $\mu$. This yields
(\ref{eq:oraclew}) and (\ref{eq:minimaxupperw}) for local
minimizers. Theorem \ref{thm:nullw} uses the preceding risk bounds together with
the observation that the optimal constant estimate is within one alpha-expansion
from any vector $\mu$. Finally, the lower bound in
Theorem \ref{thm:minimaxw} follows from an embedding of vertex-sparse vectors
into $\{\mu_0:\|D\mu_0\|_\w \leq s\}$
and a standard lower bound for sparse normal-means; similar arguments were used
in \cite{sharpnacketalwavelet,sadhanalaetal}.

\section{Simulations}\label{sec:simulations}
We study empirically the squared-errors of the approximate minimizers of
(\ref{eq:objl0}) and (\ref{eq:objw}) as returned by Algorithm
\ref{alg:boykovouter},
as well as the exact minimizer of (\ref{eq:objTV}) (computed using the
pygfl library \cite{tanseyscott}). We denote these
estimates by $\hat{\mu}^{\Lz}$, $\hat{\mu}^{\W}$, and $\hat{\mu}^{\TV}$.
We consider piecewise-constant signals over various graphs,
corrupted by Gaussian noise for various noise levels $\sigma$.
We report in each setting the standardized mean-squared-error
\begin{equation}\label{eq:stMSE}\tag{st.MSE}
\frac{1}{n\sigma^2}\|\hat{\mu}-\mu_0\|^2.
\end{equation}
Due to this normalization by $\sigma^2$, one may equivalently interpret these
results as for a fixed noise level $\sigma$ under various rescalings of the
true signal $\mu_0$.

\subsection{Parameter tuning}\label{subsec:tuning}
For Algorithm \ref{alg:boykovouter},
we fix throughout $\delta=0.01$ and $\tau=0$.
This value of $\delta$ may be larger than that prescribed by the theory of the
preceding section, but represents a compromise to yield faster runtime.

We select $\lambda$ by minimizing an empirical estimate of
$\E[\|\hat{\mu}-\mu_0\|^2]$. Typically, cross-validation is used to obtain such 
an estimate. However, we observe that
naive cross-validation does not necessarily work well for all graphs and
signals. (Consider, for example, a case where the primary contribution to error
comes from
vertices $i$ near the boundaries of the constant pieces of $\mu_0$, and
estimation of these values $\mu_{0,i}$ is more difficult when $Y_i$ is removed.)
We instead use the following procedure based on \cite{tiantaylor,harris}:
\begin{enumerate}
\item Compute an estimate $\hat{\sigma}$ for $\sigma$. Set $\alpha=0.04$.
\item For repetitions $b=1,\ldots,B$:
\begin{enumerate}
\item Generate $z=(z_1,\ldots,z_n) \sim \Normal(0,\alpha\hat{\sigma}^2\Id)$,
and set $Y^*=Y+z$ and $Y^{**}=Y-z/\alpha$.
\item For each $\lambda$, compute $\hat{\mu}$ based on data $Y^*$,
and compute $\text{err}^{(b)}(\lambda)=\|\hat{\mu}-Y^{**}\|^2$.
\end{enumerate}
\item Choose $\lambda$ that minimizes the average error
$\overline{\text{err}}(\lambda)=\frac{1}{B}\sum_{b=1}^B
\text{err}^{(b)}(\lambda)$.
\end{enumerate}
This is motivated by the insight that
if $\hat{\sigma}=\sigma$, then $Y^*$ and $Y^{**}$
are independent, so $\E[\text{err}^{(b)}(\lambda)]=n\sigma^2(1+\alpha^{-1})
+\E[\|\hat{\mu}-\mu_0\|^2]$. Hence $\overline{\text{err}}(\lambda)$
estimates a constant plus the risk of $\hat{\mu}$
applied to data at the slightly elevated noise level
$\sigma(1+\sqrt{\alpha})=1.2\sigma$. Due to this elevation in noise level, this
procedure has a slight tendency to oversmooth.

For each edge $\{i,j\}$ where $\mu_{0,i}=\mu_{0,j}$, we have
$Y_i-Y_j \sim \Normal(0,2\sigma^2)$. Hence $\sigma$ may be estimated from
the edge differences $(Y_i-Y_j)_{\{i,j\} \in E}$ by identifying a
normal mixture component corresponding to this subset of values; we used
the method of \cite{efron} as implemented in the locfdr R package.
Increasing $B$ reduces the variability of the selection procedure. For the
smaller graphs (linear chain, Oldenburg, Gnutella P2P) we set $B=20$, and for
the larger graphs (2-D cow, San Francisco, Enron email) we set $B=5$.

We will report both the st.MSE achieved using this method, as well as the
best-attained st.MSE corresponding to retrospective optimal tuning of $\lambda$.
For (\ref{eq:objTV}), $\lambda$ may alternatively be selected by minimizing
Stein's unbiased risk estimate (SURE) using the simple degrees-of-freedom
formula derived in \cite{tibshiranitaylor}. We found results of the SURE
approach to be very close to those obtained using the above procedure.

\subsection{Empirical runtime}\label{subsec:runtime}
For Algorithm \ref{alg:boykovouter}, we computed minimum s-t cuts using the
method of \cite{boykovkolmogorov}. The outer loop required
no more than 15 iterations, and typically fewer than 10 iterations,
in all tested examples. Table \ref{tab:runtime} displays the average runtime
of this algorithm on our personal computer for computing $\hat{\mu}^{\Lz}$ with
a single value of $\lambda$. The runtime of this algorithm
for computing $\hat{\mu}^\W$ was comparable, although 
computing effective resistance weights required an additional a priori cost of
10 seconds, 3 hours, 45 seconds, and 30 minutes for the four networks in the
order listed, using the approxCholLap method of the Laplacians-0.2.0 Julia
package with error tolerance $0.01$. (The effective resistance computation is a
one-time cost per network, reusable across different $\lambda$ values and data
vectors $Y$.) Parameter tuning using the approach of Section
\ref{subsec:tuning} is slower as it requires running the method multiple times
over a range of $\lambda$ values, although this computation is easily
parallelized.

\begin{table}[b]
\begin{tabular}{l|cccccc}
Graph & 1-D & cow & Oldenburg & San Fran. & Gnutella & Enron \\
\hline
Runtime (seconds) & 0.13 & 45 & 0.7 & 40 & 4 & 240
\end{tabular}
\caption{Average computational time of Algorithm \ref{alg:boykovouter}
for one value of $\lambda$}\label{tab:runtime}
\end{table}

\subsection{Linear chain graph}
\begin{figure}
\includegraphics[width=0.33\textwidth]{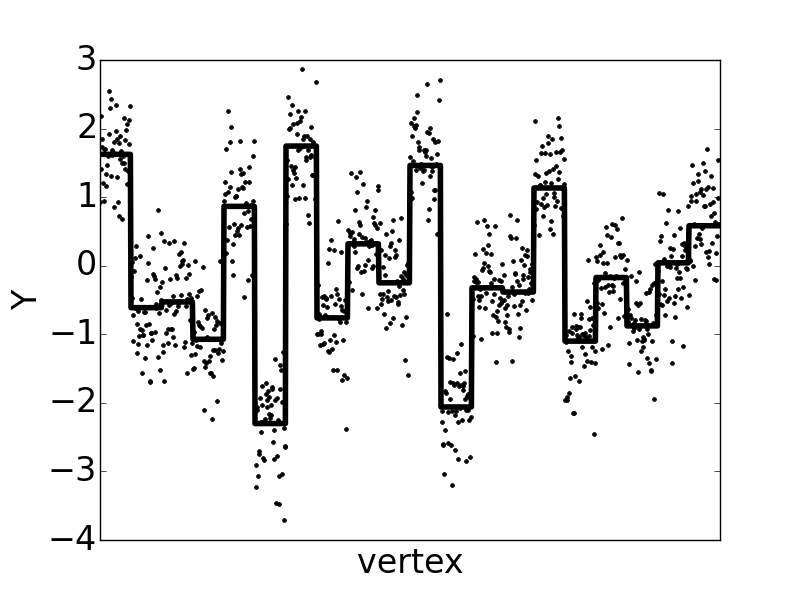}%
\includegraphics[width=0.33\textwidth]{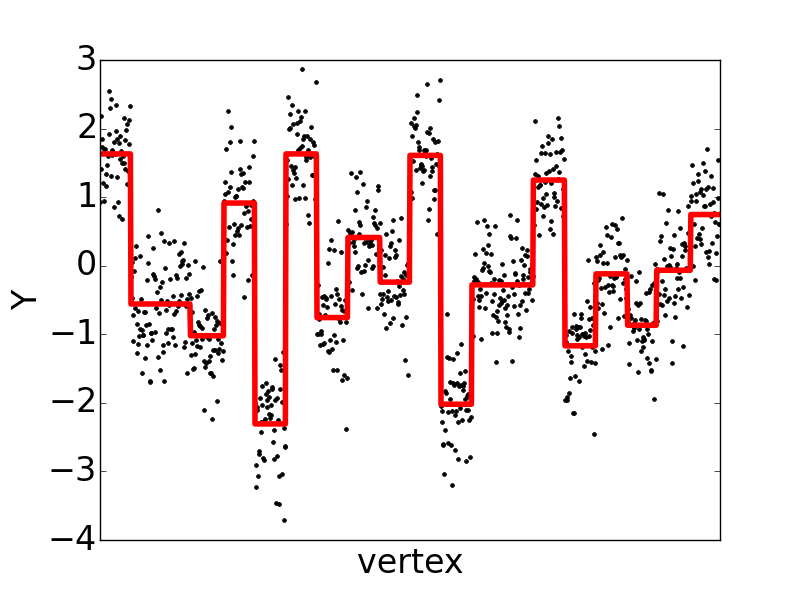}%
\includegraphics[width=0.33\textwidth]{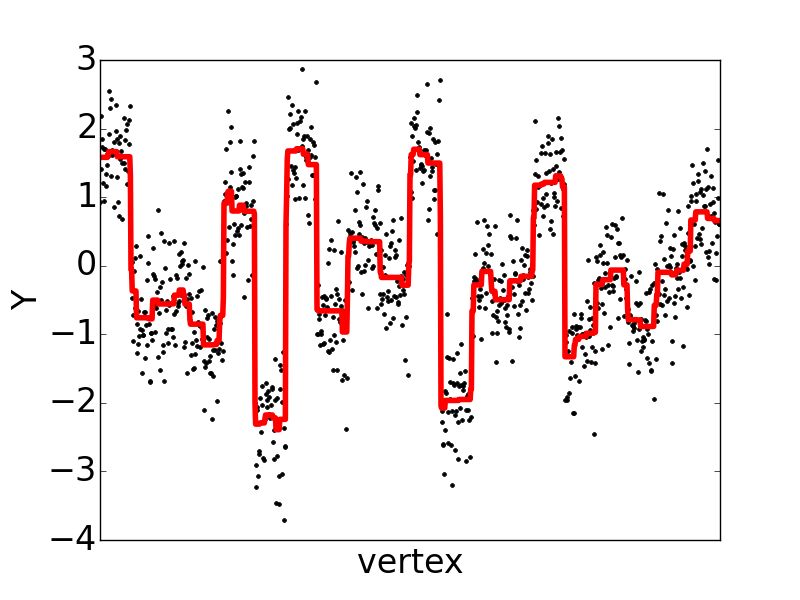}
\includegraphics[width=0.33\textwidth]{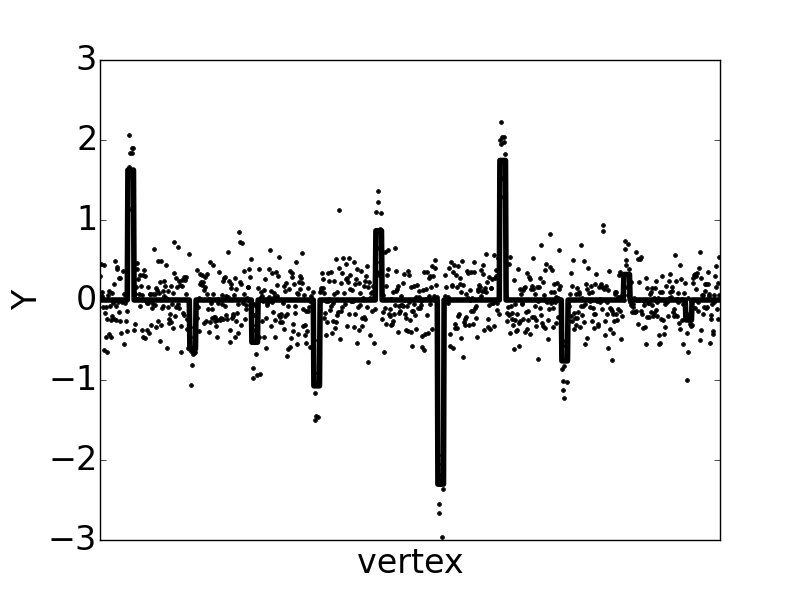}%
\includegraphics[width=0.33\textwidth]{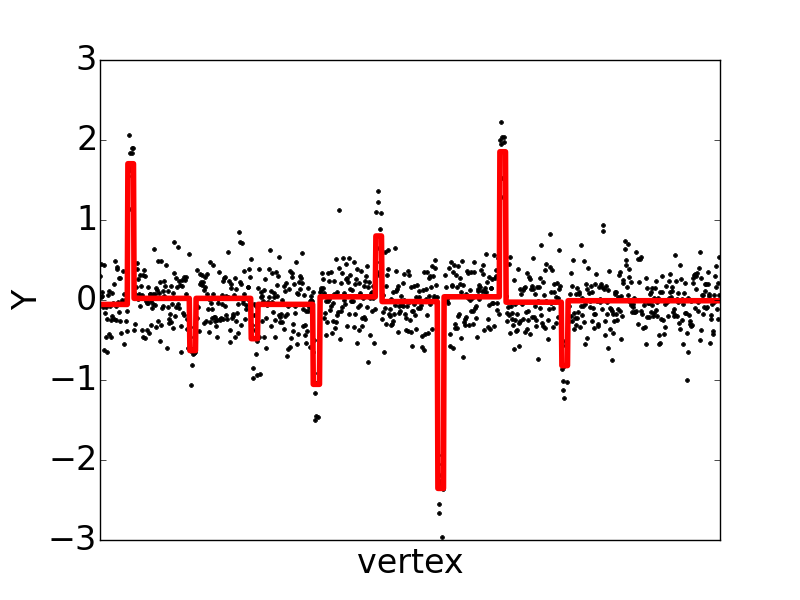}%
\includegraphics[width=0.33\textwidth]{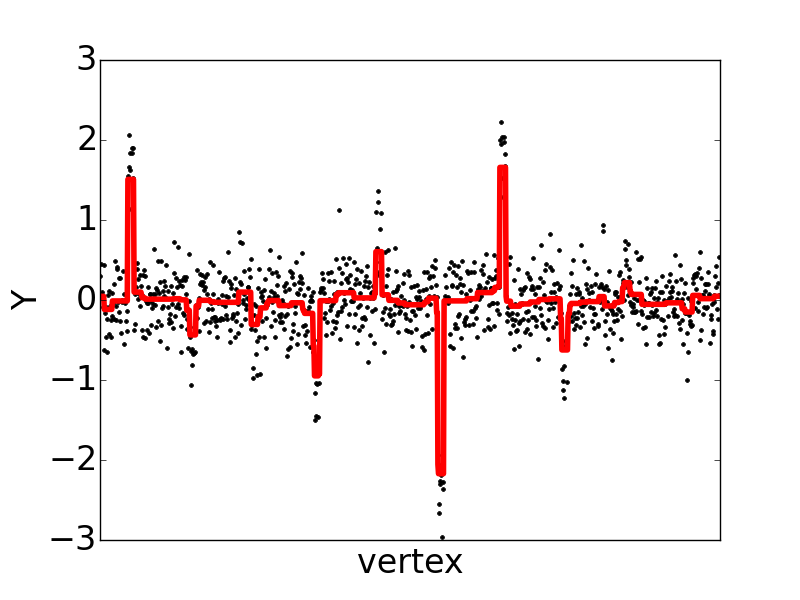}
\caption{Signals on a linear chain graph. Top: Equally-spaced breaks,
$\sigma=0.5$. Bottom: Unequally-spaced breaks, $\sigma=0.3$.
The true signal $\mu_0$ is displayed on the left, $\hat{\mu}^\Lz$
in the middle, and $\hat{\mu}^\TV$ on the right (both with data-tuned
$\lambda$).}\label{fig:1Dsignals}
\end{figure}
\begin{figure}
\includegraphics[width=0.5\textwidth]{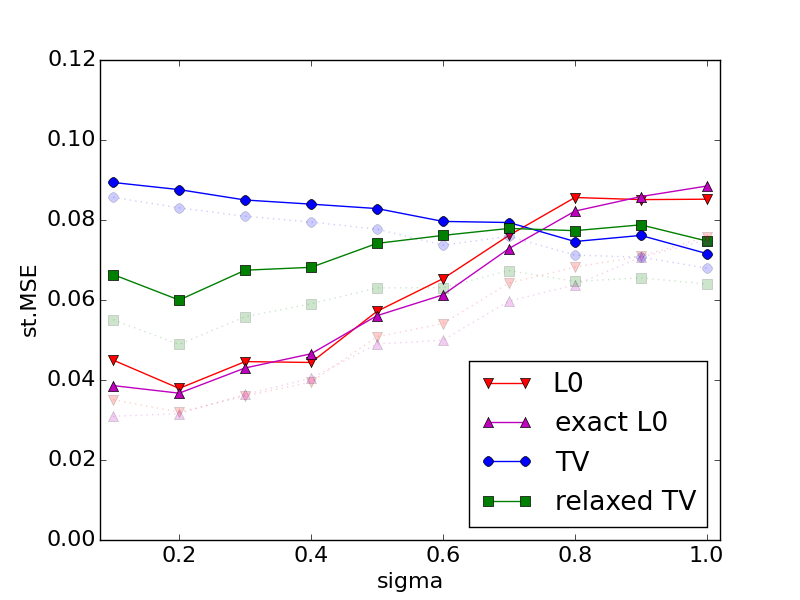}%
\includegraphics[width=0.5\textwidth]{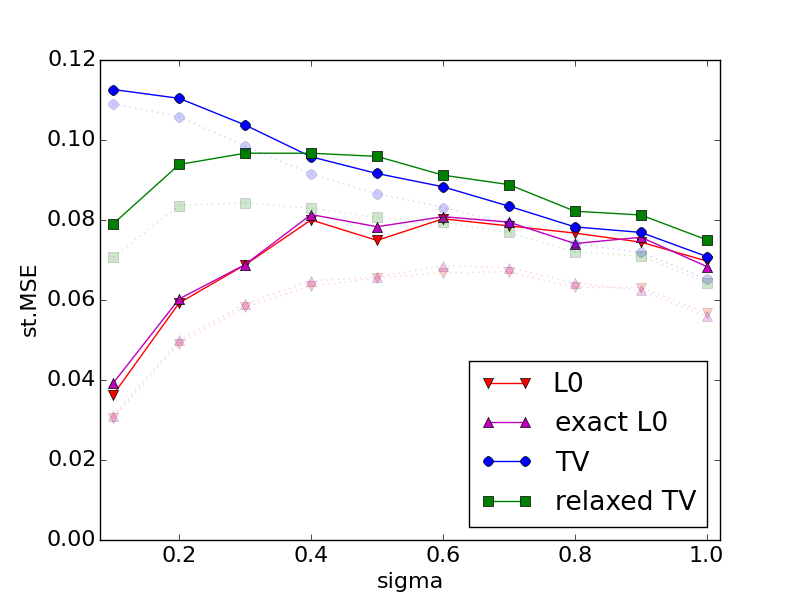}
\caption{Comparisons of st.MSE in (left) the equally-spaced and
(right) unequally-spaced examples of Figure \ref{fig:1Dsignals},
for $\hat{\mu}^\Lz$, $\hat{\mu}^\TV$, the exact minimizer of (\ref{eq:objl0}),
and $\hat{\mu}^{\TV,\text{relaxed}}$. Solid lines correspond to
data-tuned $\lambda$, and dashed transparent lines to best-achieved error.
All errors are averaged over 100 replicates of the
simulated noise.}\label{fig:1Derror}
\end{figure}
\begin{figure}
\includegraphics[width=0.5\textwidth]{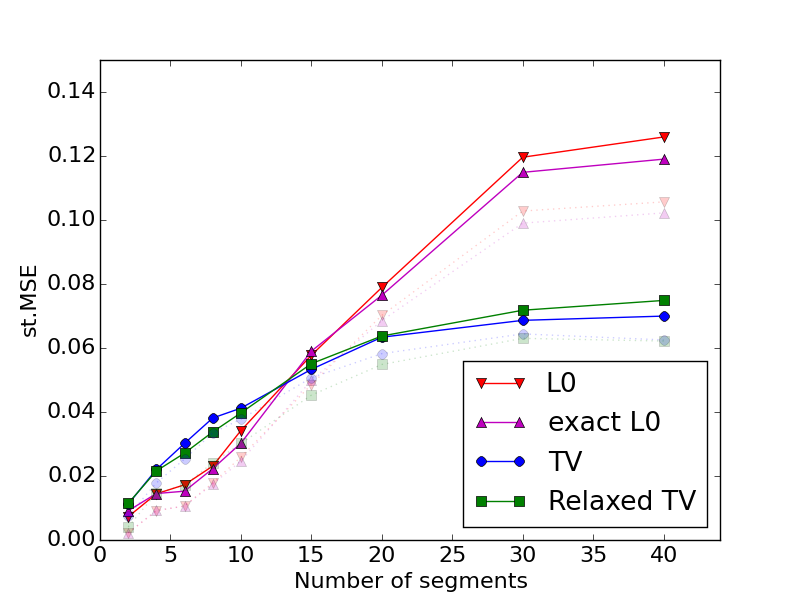}%
\includegraphics[width=0.5\textwidth]{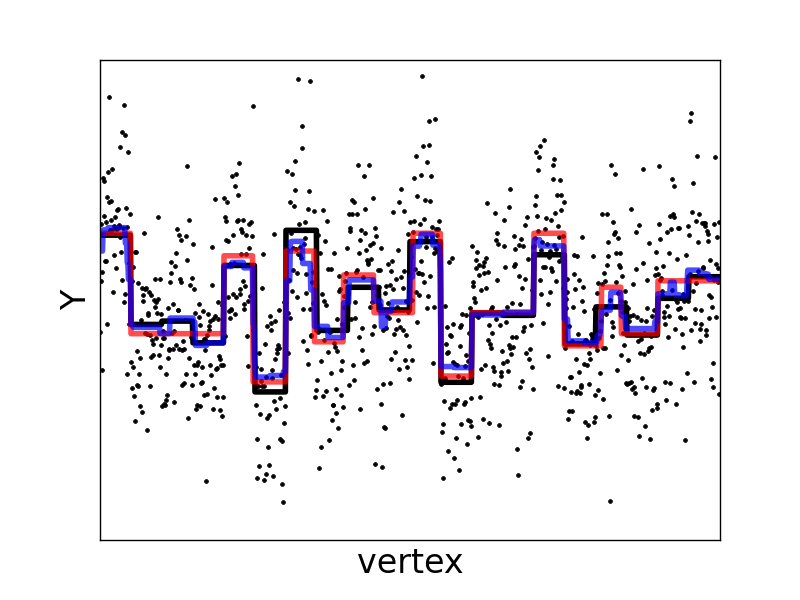}
\caption{(Left) Comparisons of st.MSE for signals of fixed total-variation
$\|D\mu_0\|_1/\sigma=20$ and increasing numbers of segments/decreasing jump
sizes. (Right) Raw data and true signal in black, $\hat{\mu}^\Lz$ in red, and
$\hat{\mu}^\TV$ in blue, for the signal with 20 segments.}\label{fig:1DconstL1}
\end{figure}

Two signals on a linear chain graph with $n=1000$ vertices are
depicted in Figure \ref{fig:1Dsignals}. The first signal
has 19 equally-spaced break points, while the second has 20 break points at
unequal spacing. We studied recovery for
noise levels $\sigma=0.1$ to $\sigma=1$. Figure \ref{fig:1Dsignals} displays
one instance of simulated noise and the resulting estimates
$\hat{\mu}^\Lz$ and $\hat{\mu}^\TV$. In both examples, for data-tuned
$\lambda$, $\hat{\mu}^\Lz$ tends to over-smooth (missing two and four
changepoints respectively) and $\hat{\mu}^\TV$ tends to undersmooth.

Figure \ref{fig:1Derror} plots st.MSE comparisons for $\hat{\mu}^\Lz$ and
$\hat{\mu}^\TV$. The $\hat{\mu}^\Lz$ estimate
achieves significantly smaller risk than $\hat{\mu}^\TV$ at higher
signal-to-noise regimes, for example those displayed in Figure
\ref{fig:1Dsignals}, while $\hat{\mu}^\TV$ becomes competitive or better
in lower signal-to-noise regimes, corresponding to lower values of
normalized total-variation $\|D\mu_0\|_1/\sigma$ for the true signal.
Figure \ref{fig:1DconstL1} presents a different example to further explore this
trade-off, in which the normalized total-variation of the signal is fixed at
$\|D\mu_0\|_1/\sigma=20$, and we increase the number of
changepoints while simultaneously decreasing the jump sizes. (Changepoints are
equally spaced, and distinct signal values are normally distributed.)
The estimate $\hat{\mu}^\Lz$ is better under strong edge-sparsity, while
$\hat{\mu}^\TV$ becomes better as we transition to weaker edge-sparsity.

For the linear chain, we may compare $\hat{\mu}^\Lz$ with the exact minimizer
of (\ref{eq:objl0}) (computed using PELT in the changepoint R package
\cite{killicketal}). Algorithm \ref{alg:boykovouter} achieves risk comparable
to the exact minimizer in all tested settings. One may ask, at the higher
signal-to-noise regimes, how much of the sub-optimality of $\hat{\mu}^\TV$ is
due to estimator bias incurred by shrinkage. To address this, we computed
also the ``relaxed'' TV estimate
\[\hat{\mu}^{\TV,\text{relaxed}}=\alpha \hat{\mu}^{\TV}+(1-\alpha)
\hat{\mu}^{\TV,\text{debiased}}\]
where $\alpha \in \{0,0.1,0.2,\ldots,1\}$ is an additional tuning parameter, and
where $\hat{\mu}^{\TV,\text{debiased}}$ replaces each
constant interval of $\hat{\mu}^\TV$ with the mean of $Y$ over this interval.
The error of $\hat{\mu}^\TV$ at high signal-to-noise is partially reduced,
but not to the same levels as $\hat{\mu}^\Lz$.

\subsection{2-D lattice graph}
\begin{figure}
\includegraphics[width=0.25\textwidth]{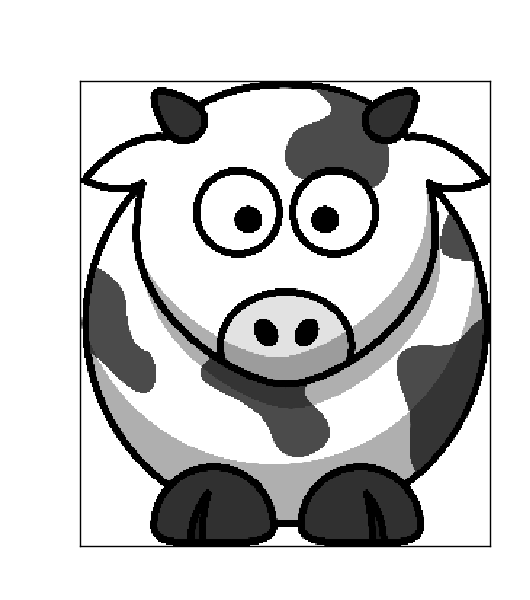}%
\includegraphics[width=0.25\textwidth]{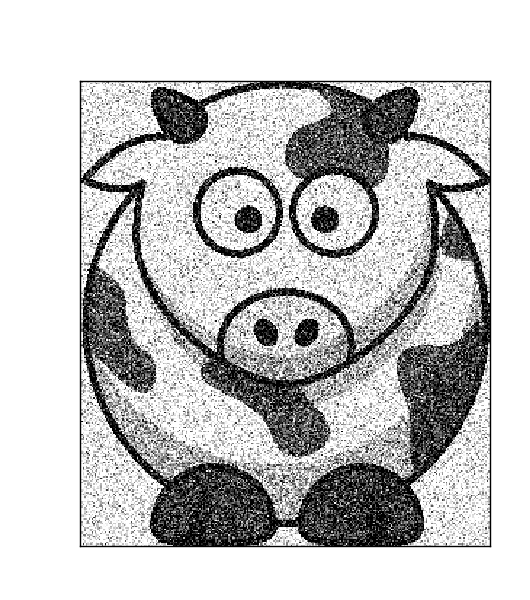}%
\includegraphics[width=0.25\textwidth]{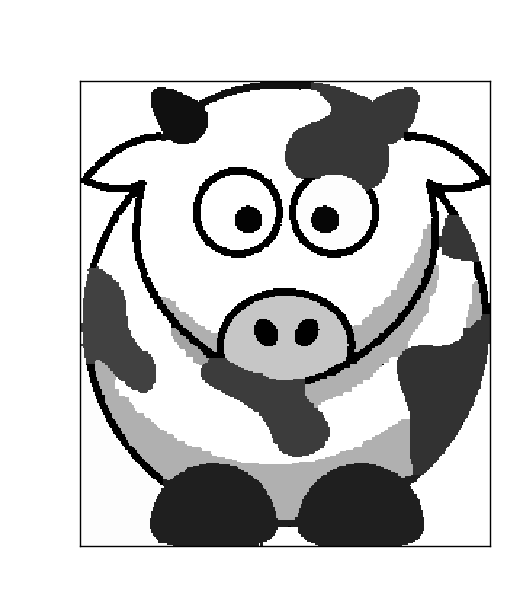}%
\includegraphics[width=0.25\textwidth]{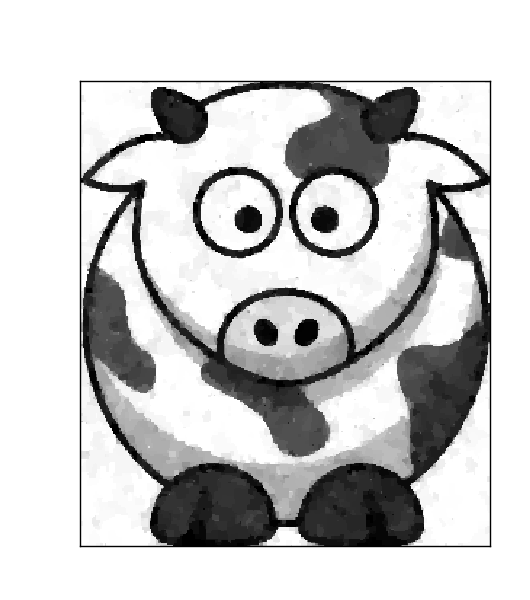}
\caption{(far left) Original image, with pixel values normalized to $[0,1]$.
(middle left) Noisy image, $\sigma=0.3$. (middle right) 
$\hat{\mu}^\Lz$ and (far right) $\hat{\mu}^\TV$, both with data-tuned
$\lambda$.}\label{fig:cow}
\end{figure}
Figure \ref{fig:cow} displays a cartoon gray-scale image of a cow, represented
by its pixel values on a 2-D lattice graph of size $320 \times 283$. Pure white
corresponds to $\mu_0=1$, and pure black to $\mu_0=0$. The figure also displays
$\hat{\mu}^\Lz$ and $\hat{\mu}^\TV$ when the image is contaminated by noise at
level $\sigma=0.3$. We again observe that $\hat{\mu}^\Lz$ oversmooths, missing
details in the cow's feet, right horn, and the shadows of the image.
In contrast, $\hat{\mu}^\TV$ undersmooths and returns a blotchy cow.

\begin{table}[b]
\begin{tabular}{c|ccccc}
sigma & 0.1 & 0.2 & 0.3 & 0.4 & 0.5 \\
\hline
\multirow{2}{*}{$\hat{\mu}^\Lz$ st.MSE} & 0.041 & 0.067 & 0.076 & 0.081 & 0.082 \\
 & (0.039) & (0.067) & (0.071) & (0.081) & (0.082) \\
\hline
\multirow{2}{*}{$\hat{\mu}^\TV$ st.MSE} & 0.083 & 0.075 & 0.065 & 0.061 & 0.054 \\
 & (0.083) & (0.075) & (0.065) & (0.059) & (0.053) \\
\end{tabular}

\caption{Comparison of st.MSE for the cow image of Figure \ref{fig:cow}.
Non-parenthesized values correspond to data-tuned $\lambda$, and parenthesized
values to best-attained error.}
\label{tab:cow}
\end{table}
Table \ref{tab:cow} reports st.MSE comparisons for $\sigma=0.1$ to $\sigma=0.5$.
At the level $\sigma=0.3$ displayed in Figure \ref{fig:cow}, the st.MSE of
$\hat{\mu}^\Lz$ is slightly greater than that of $\hat{\mu}^\TV$. At higher
signal-to-noise levels, $\hat{\mu}^\Lz$ is better, while it is worse at
lower signal-to-noise.

\subsection{Road and digital networks}
We tested signal recovery over four real networks: the Oldenburg
and San Francisco road networks from
\url{www.cs.utah.edu/~lifeifei/SpatialDataset.htm},
and the Gnutella08 peer-to-peer network and Enron email network from
\url{snap.stanford.edu/data}. Duplicate edges were removed, and
only the largest connected component of each network was retained.

\begin{table}
\begin{tabular}{c|ccc|cc|cc|cc}
Network & verts. & edges & res. var. & inf. $T_1$ & cut $T_1$ & inf. $T_2$ & cut $T_2$ & inf. $T_3$ & cut $T_3$ \\
\hline
Oldenburg & 6105 & 7029 & 0.118 & 57 & 16 & 515 & 75 & 2108 & 164\\
San Fran. & 174956 & 221802 & 0.203 & 8574 & 306 & 27724 & 562 & 70925 & 774\\
Gnutella & 6299 & 20776 & 0.826 & 19 & 91 & 67 & 477 & 271 & 3894\\
Enron & 33696 & 180811 & 1.297 & 179 & 7319 & 2564 & 74117 & 16868 & 29253\\
\end{tabular}

\caption{For each network: Number of total vertices, total edges, variability of
effective edge resistances (measured by standard deviation / mean), and numbers
of infected vertices and cut edges corresponding to the signal at three
observation times.}\label{tab:networks}
\end{table}

\begin{figure}
\includegraphics[width=0.33\textwidth]{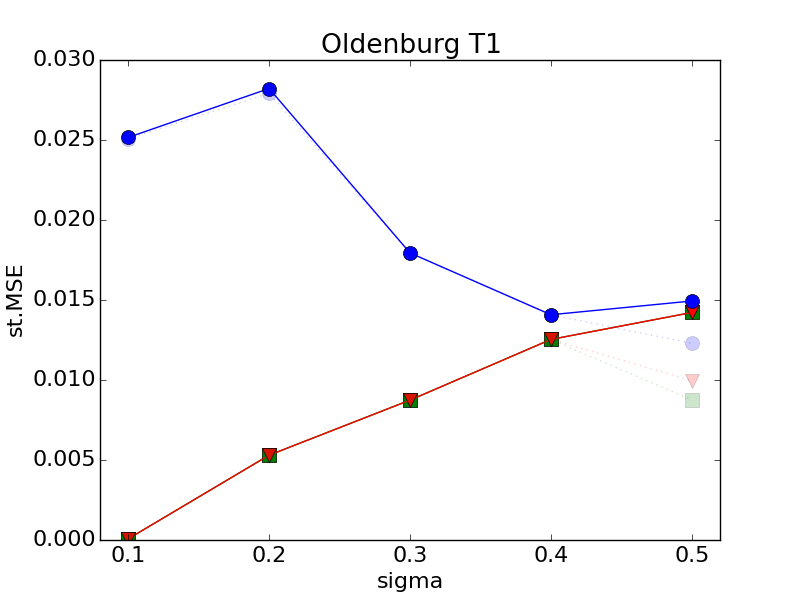}%
\includegraphics[width=0.33\textwidth]{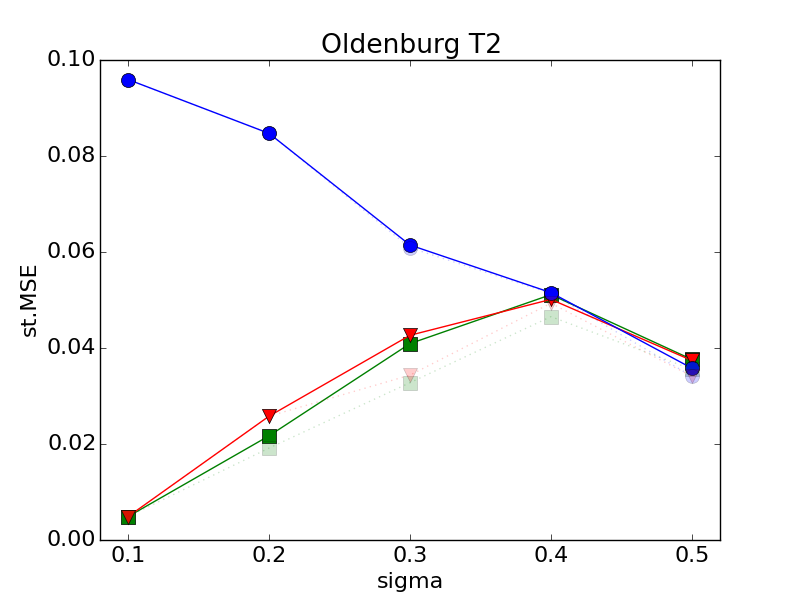}%
\includegraphics[width=0.33\textwidth]{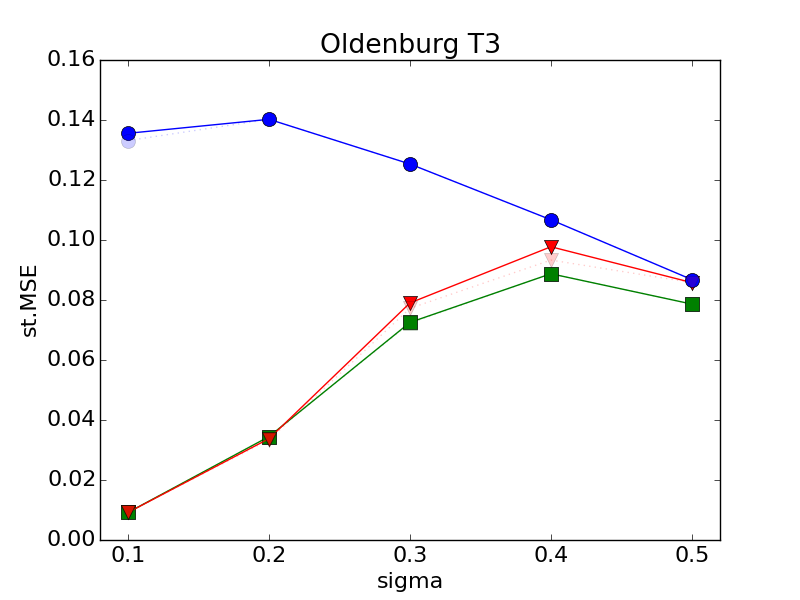}
\includegraphics[width=0.33\textwidth]{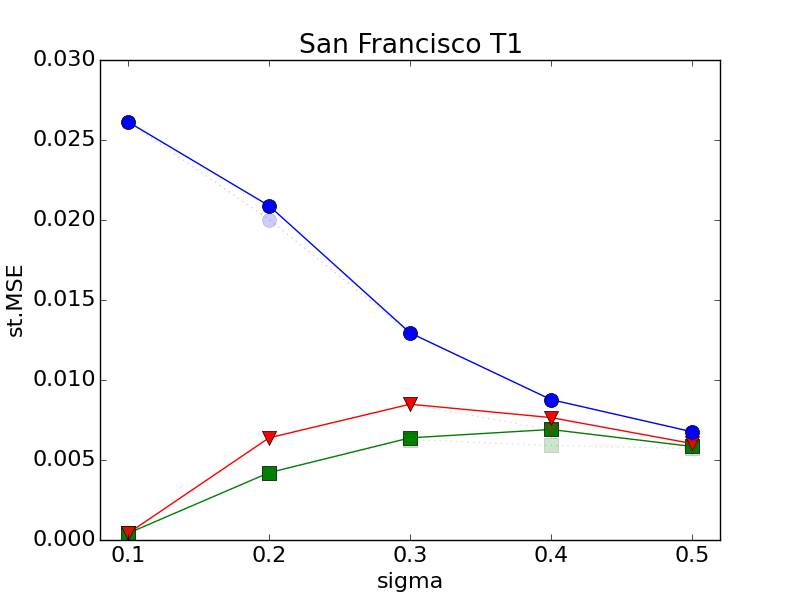}%
\includegraphics[width=0.33\textwidth]{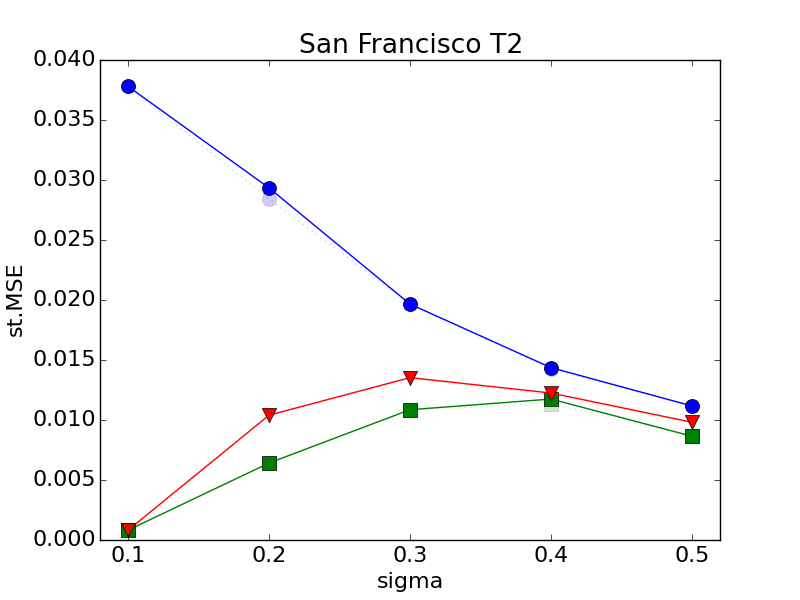}%
\includegraphics[width=0.33\textwidth]{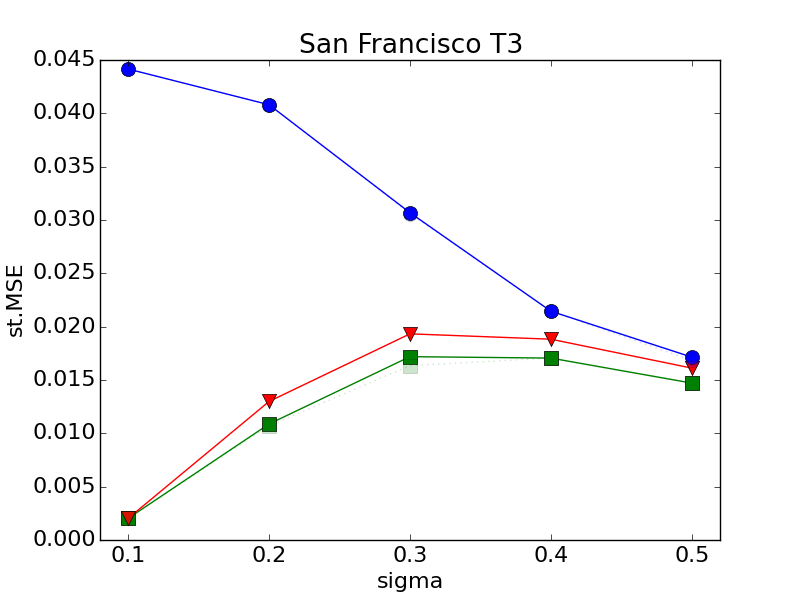}
\includegraphics[width=0.33\textwidth]{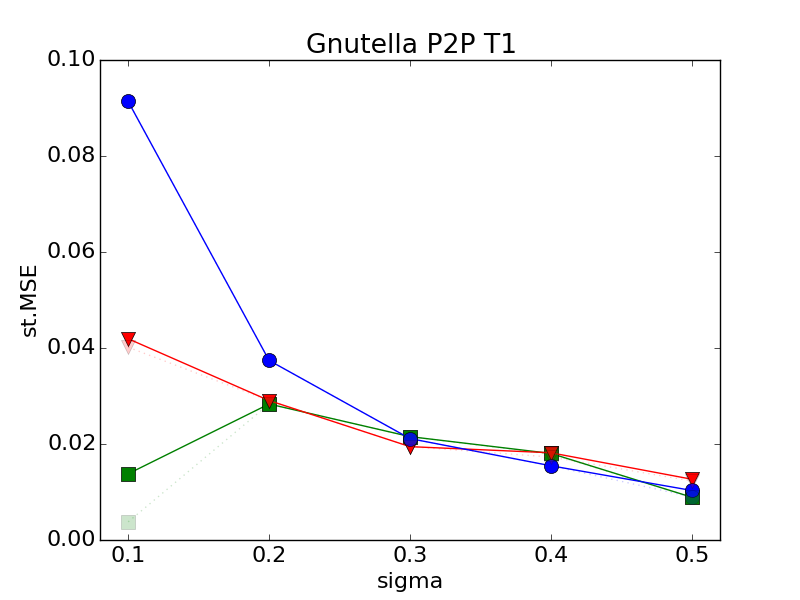}%
\includegraphics[width=0.33\textwidth]{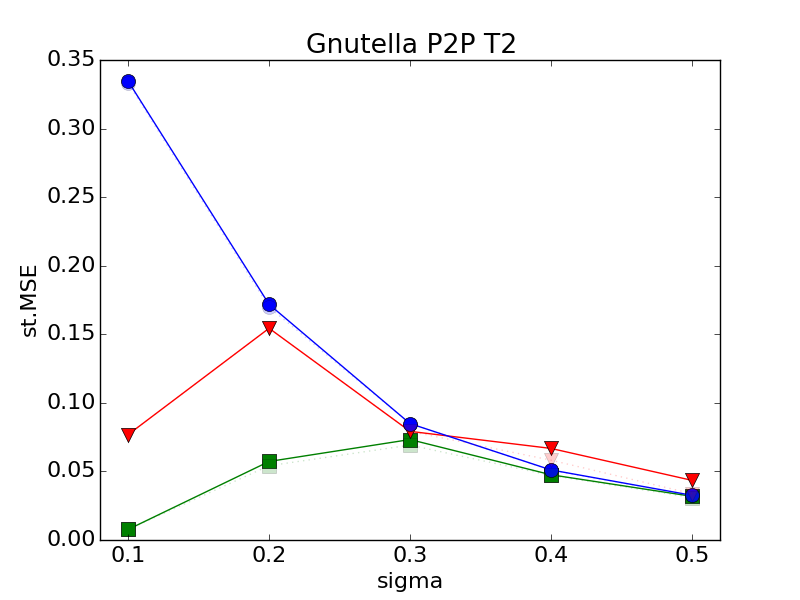}%
\includegraphics[width=0.33\textwidth]{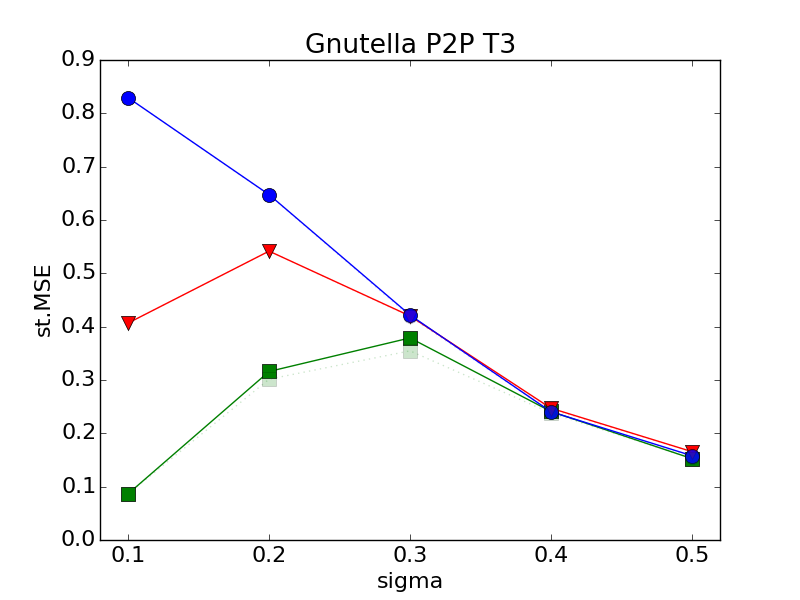}
\includegraphics[width=0.33\textwidth]{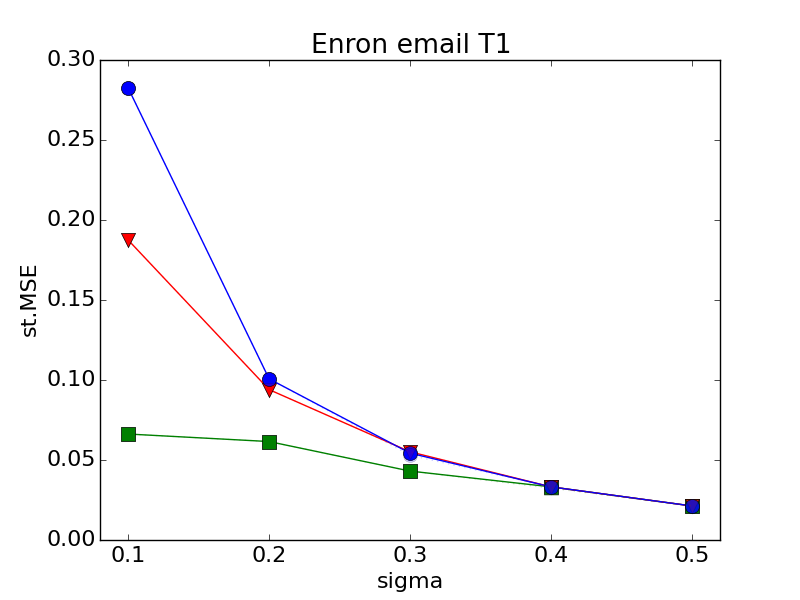}%
\includegraphics[width=0.33\textwidth]{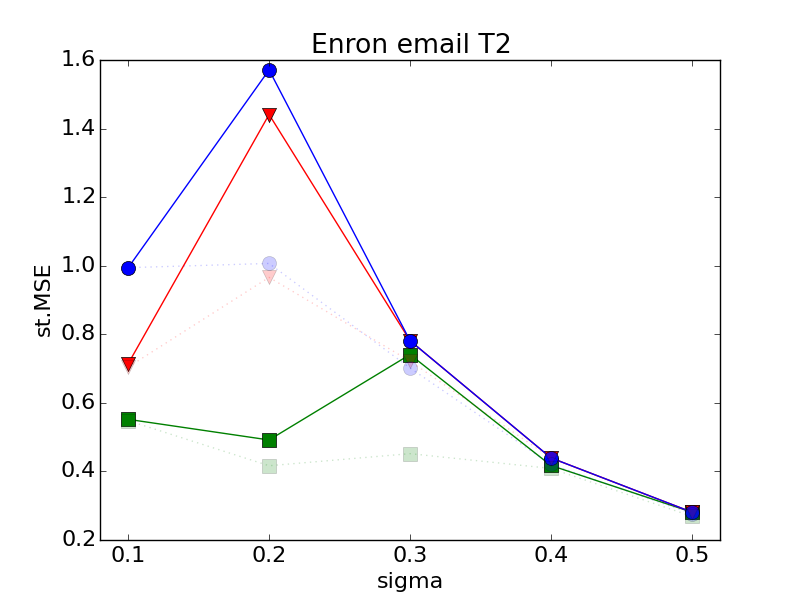}%
\includegraphics[width=0.33\textwidth]{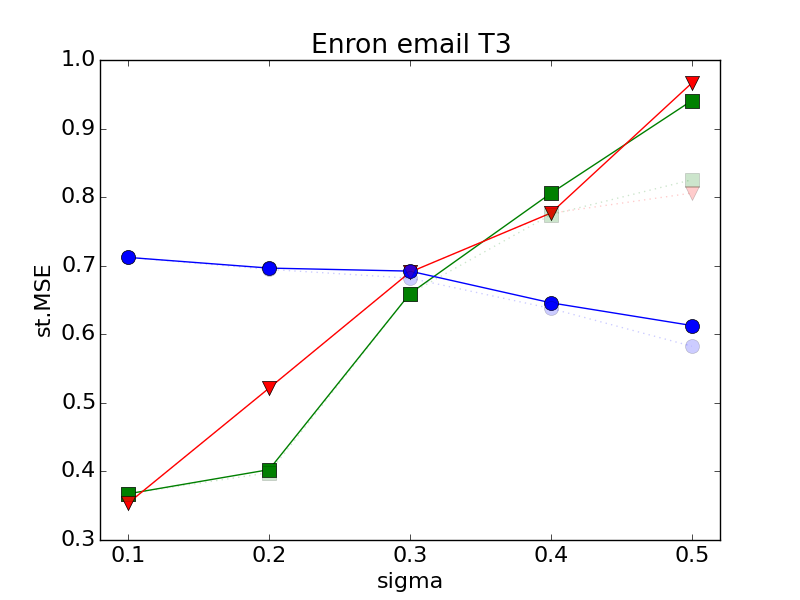}
\includegraphics[width=0.33\textwidth]{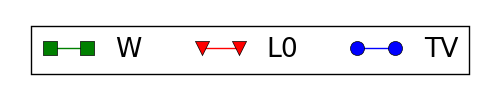}
\caption{Comparisons of st.MSE for recovery of epidemic signals on four
networks, for $\hat{\mu}^\Lz$, $\hat{\mu}^\TV$, and $\hat{\mu}^\W$ with
effective-resistance edge weights. Solid lines correspond to data-tuned
$\lambda$, and dashed transparent lines to best-achieved
error.}\label{fig:networkerrors}
\end{figure}

For each network, we simulated an epidemic according to a simple graph-based
discrete-time SI model \cite{moorenewman},
randomly selecting a source vertex to infect at time $t=0$ and, for each of
$T$ timesteps, allowing each infected vertex to independently infect each
non-infected neighbor with probability 0.5. We
associated the values $\mu_0=1.005$ and $\mu_0=0.005$ to infected and
non-infected vertices. For each
network, we considered three signals corresponding to observations of the
epidemic at three different times $T$. Various properties of these networks and
signals are summarized in Table \ref{tab:networks}.

\begin{figure}
\begin{minipage}{0.05\textwidth}
\includegraphics[width=\textwidth]{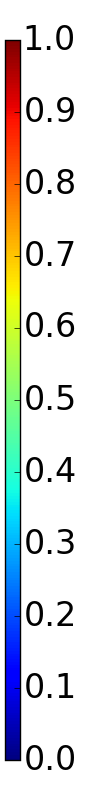}
\end{minipage}%
\begin{minipage}{0.2\textwidth}
\includegraphics[width=\textwidth]{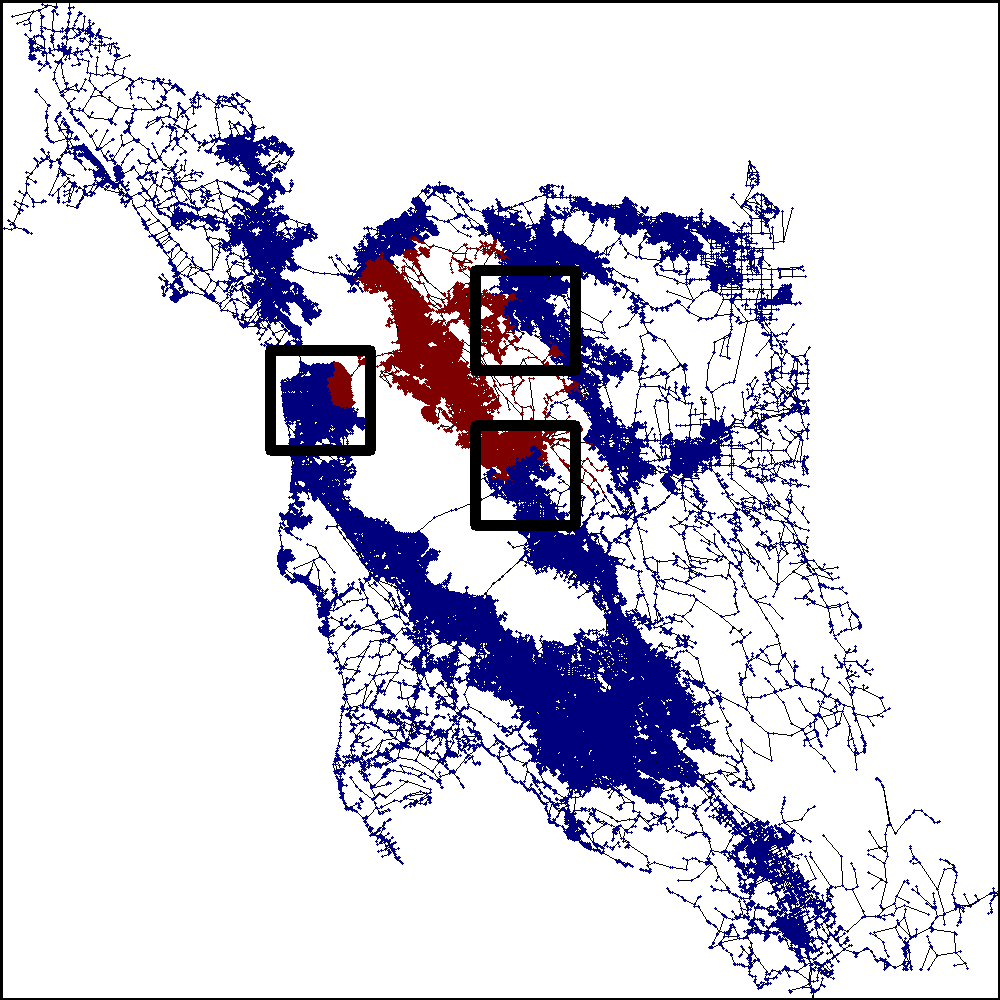}
\end{minipage}
\begin{minipage}{0.7\textwidth}
\includegraphics[width=0.33\textwidth]{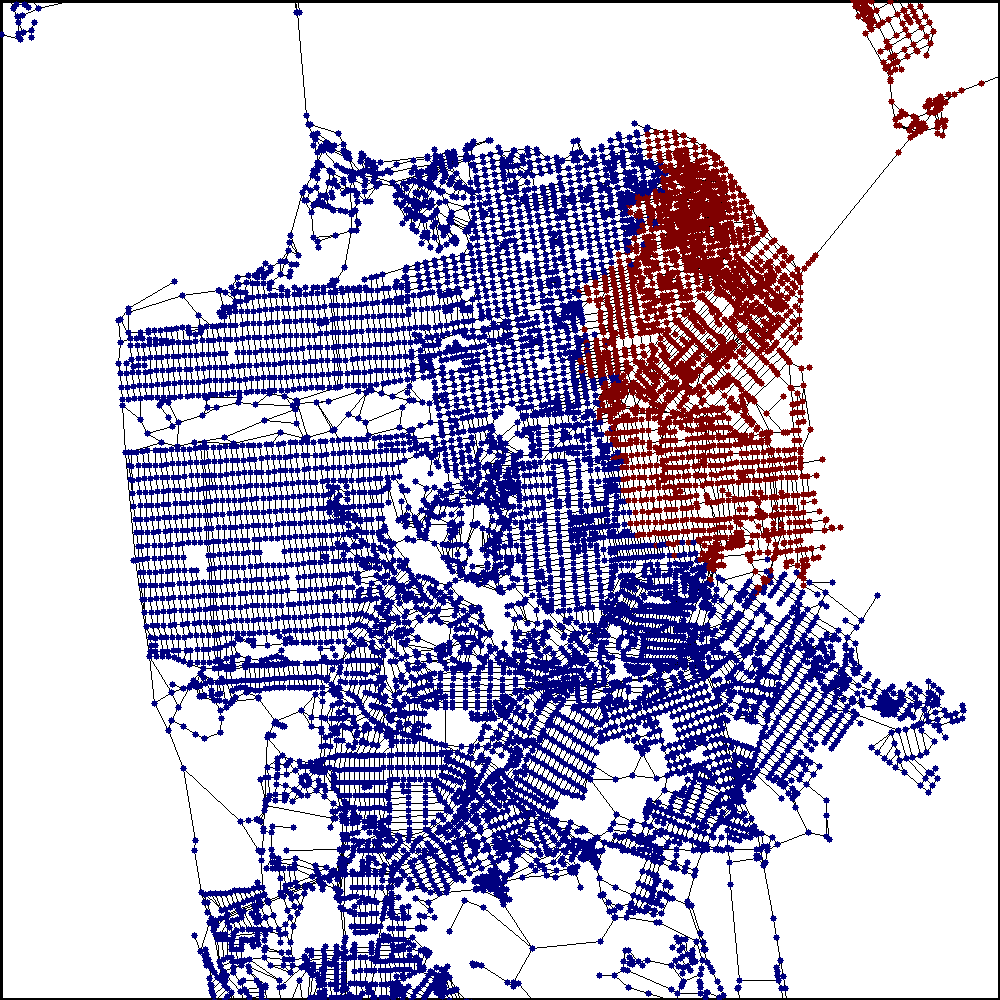}%
\includegraphics[width=0.33\textwidth]{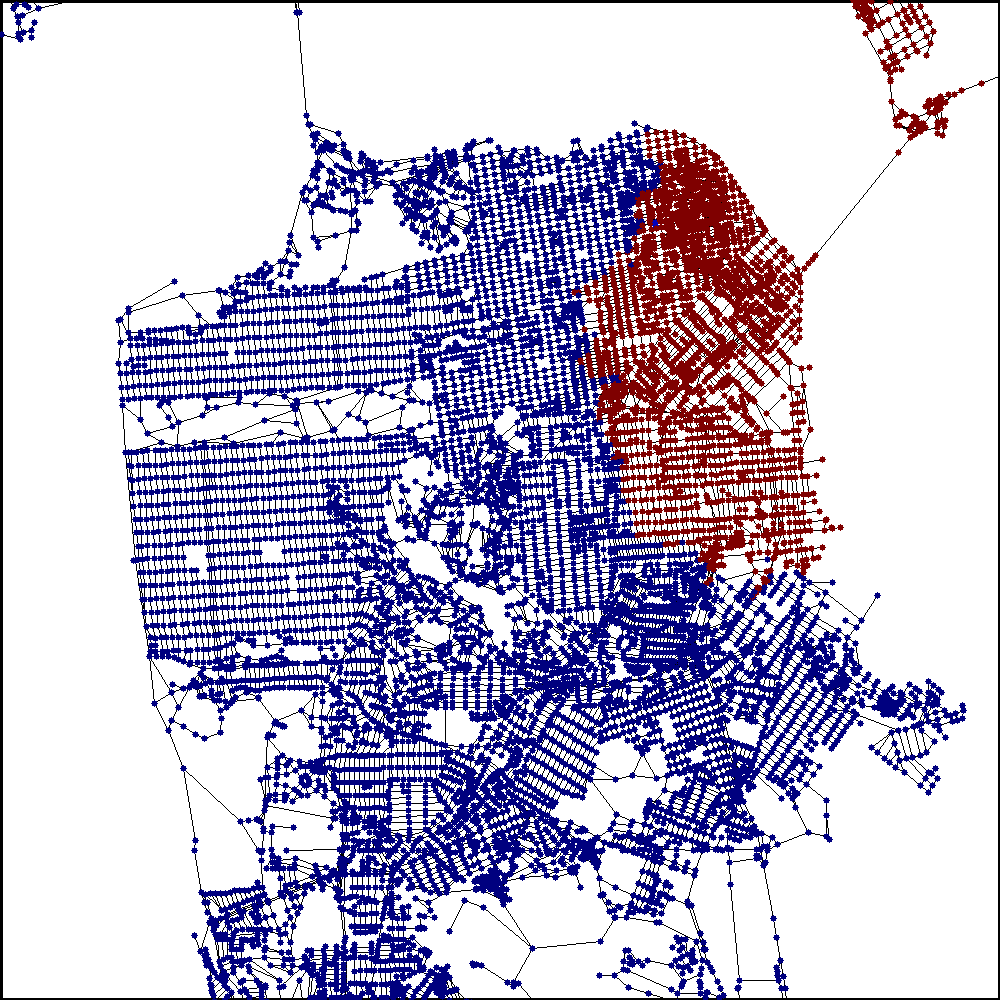}%
\includegraphics[width=0.33\textwidth]{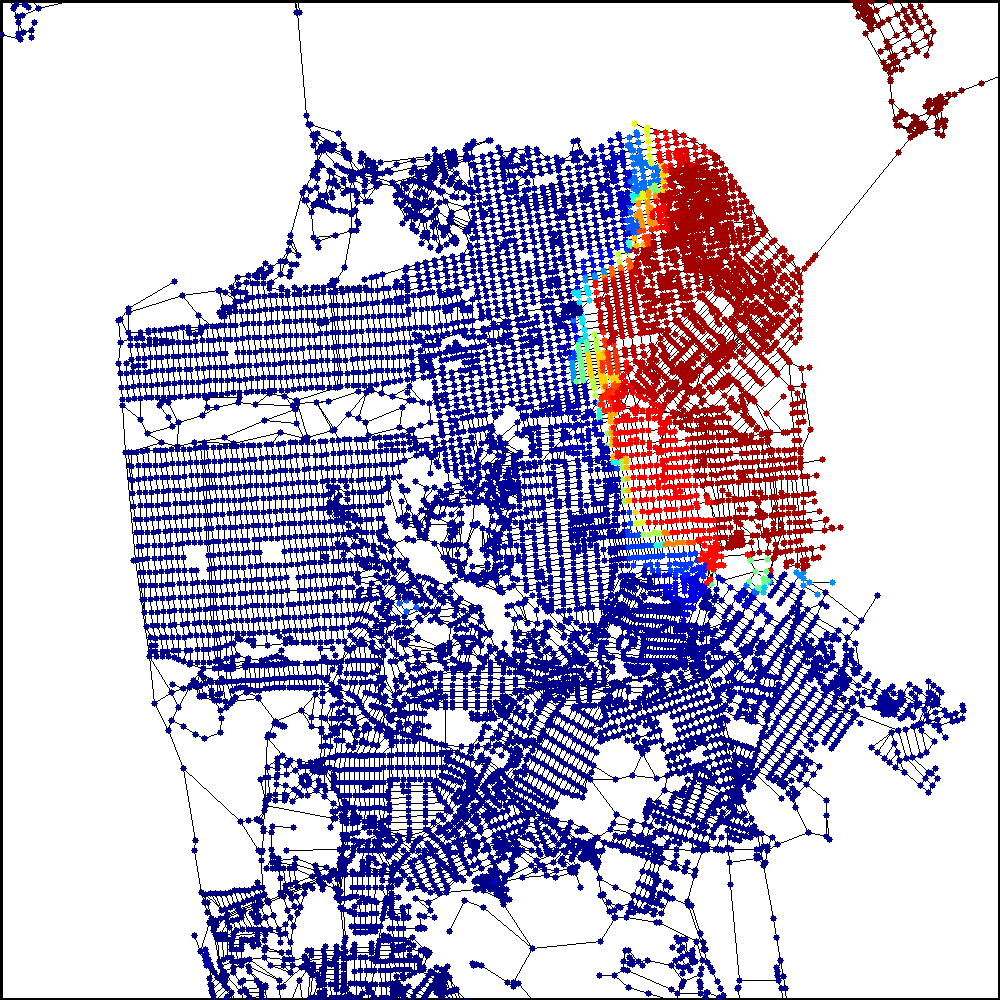}
\includegraphics[width=0.33\textwidth]{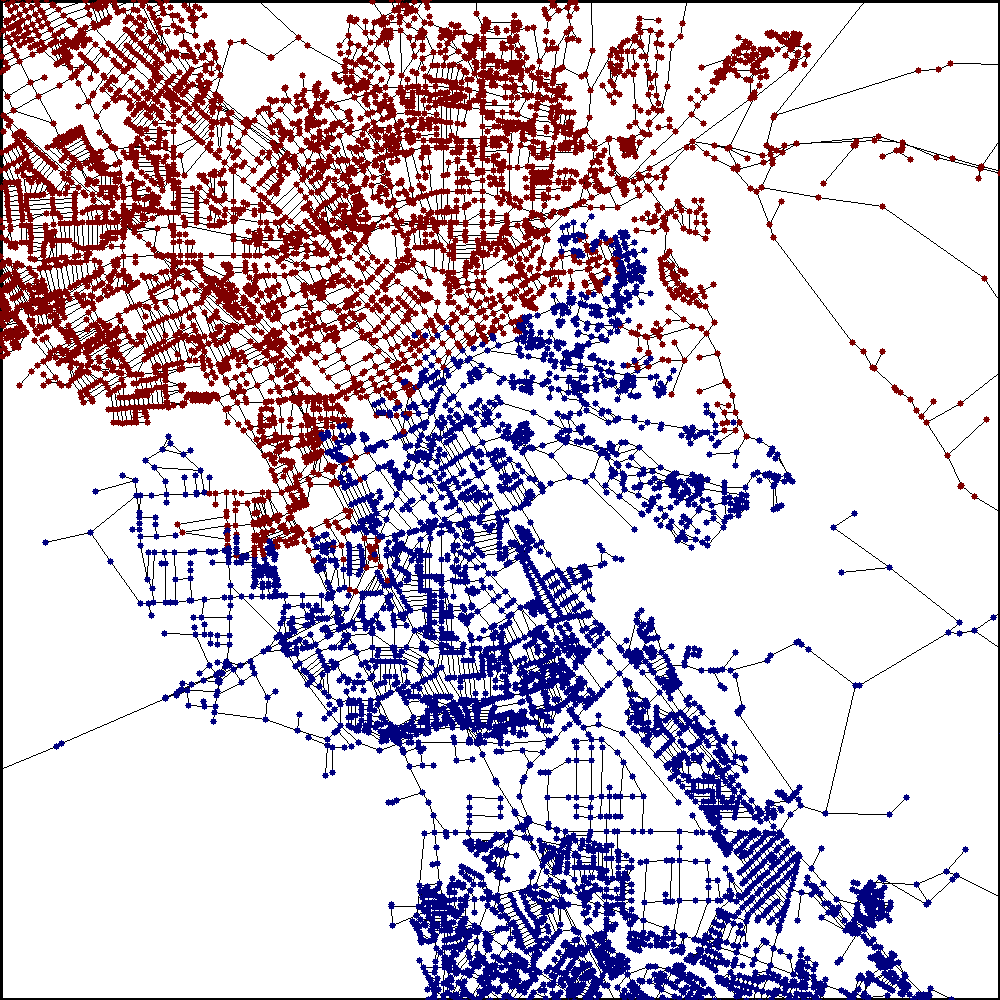}%
\includegraphics[width=0.33\textwidth]{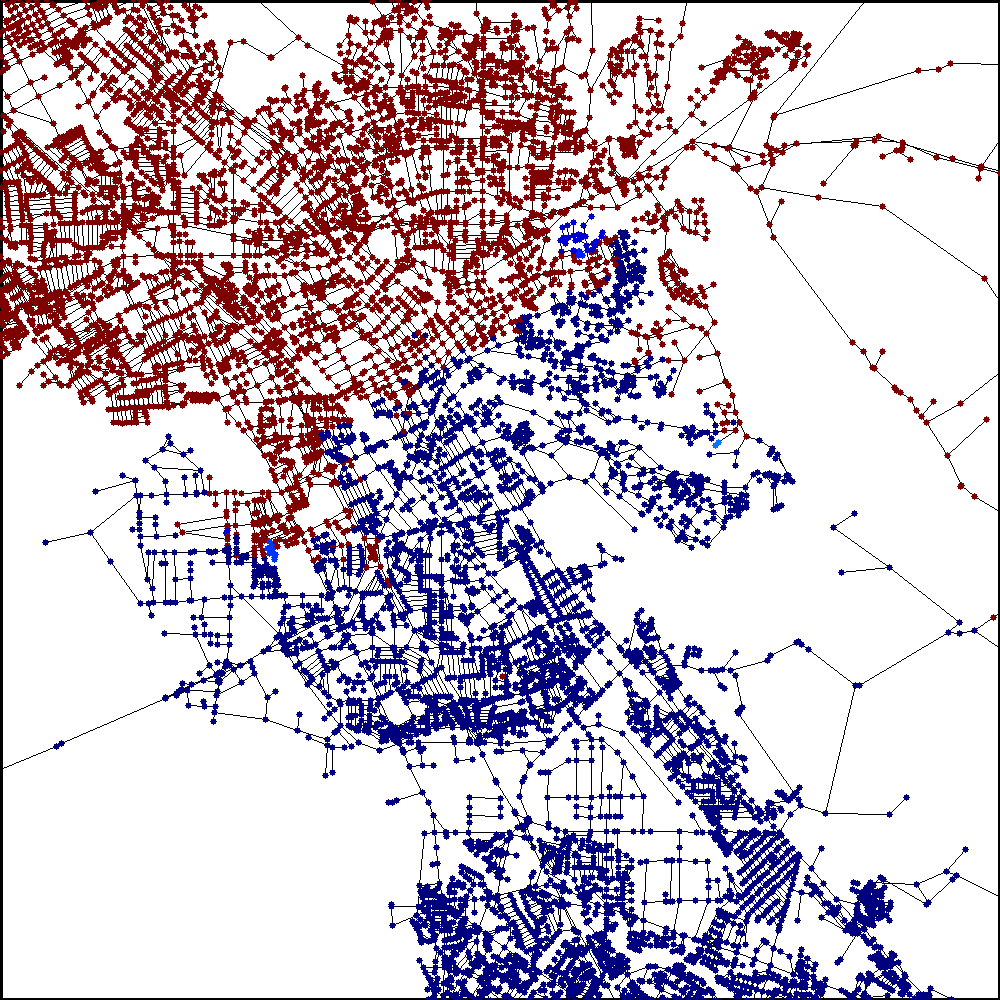}%
\includegraphics[width=0.33\textwidth]{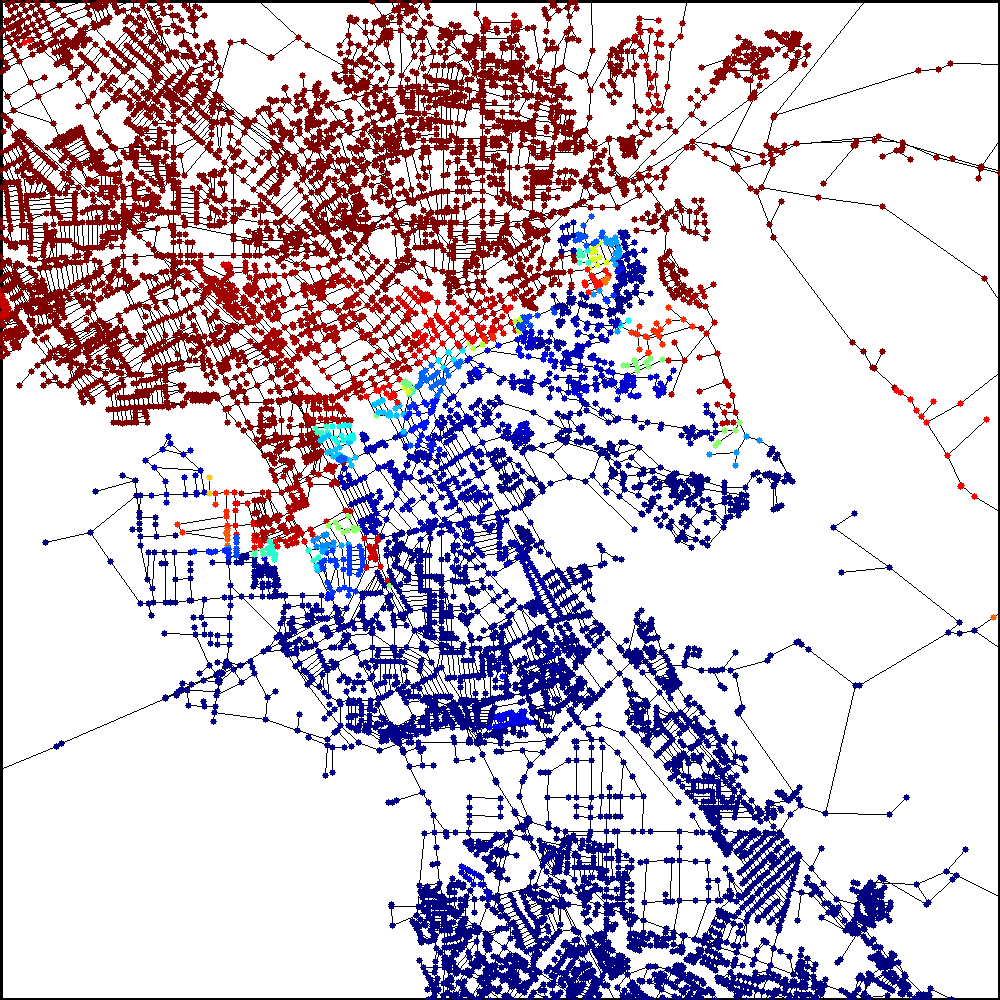}
\includegraphics[width=0.33\textwidth]{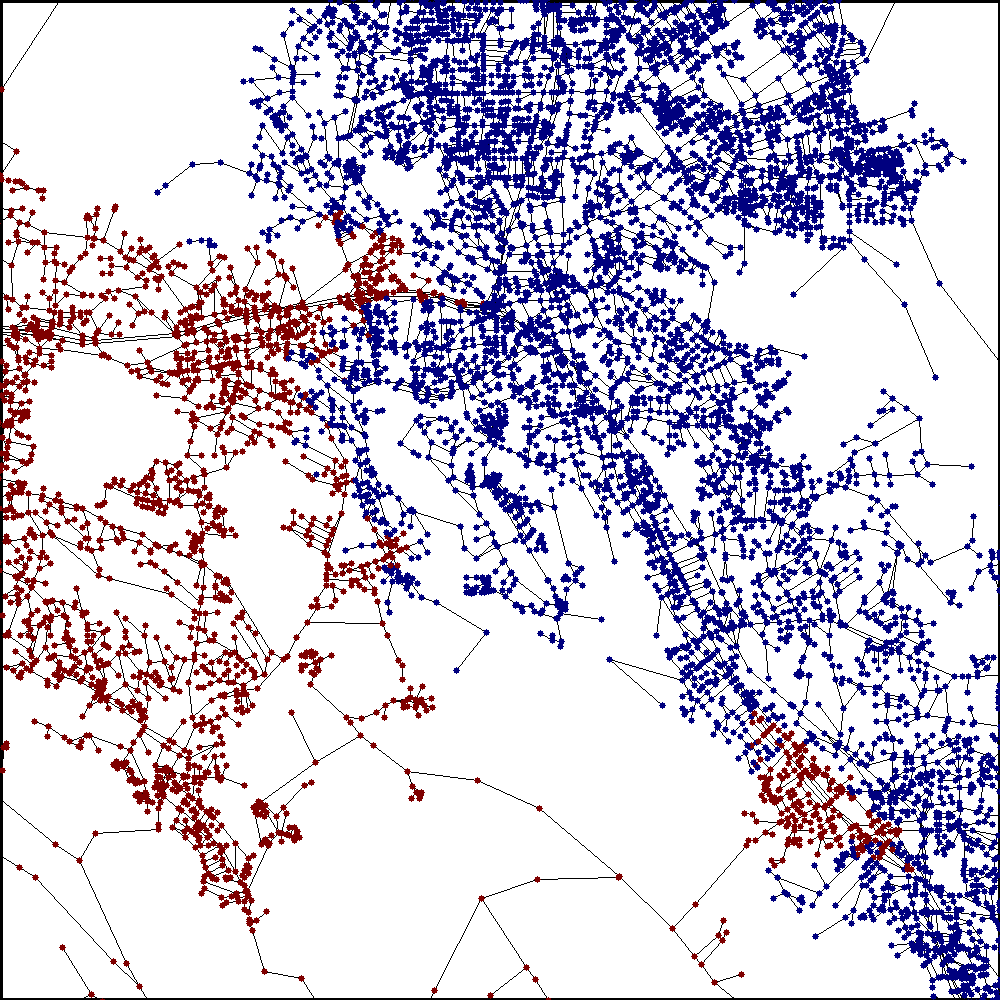}%
\includegraphics[width=0.33\textwidth]{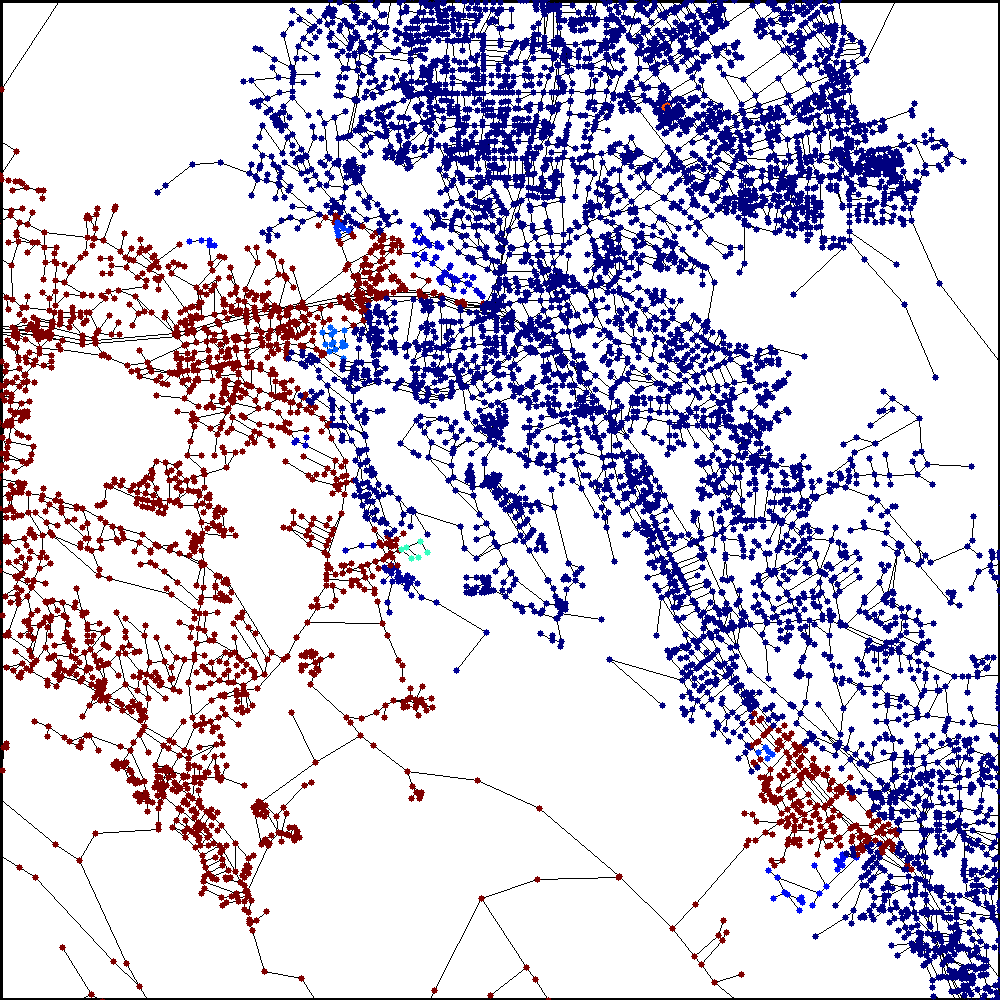}%
\includegraphics[width=0.33\textwidth]{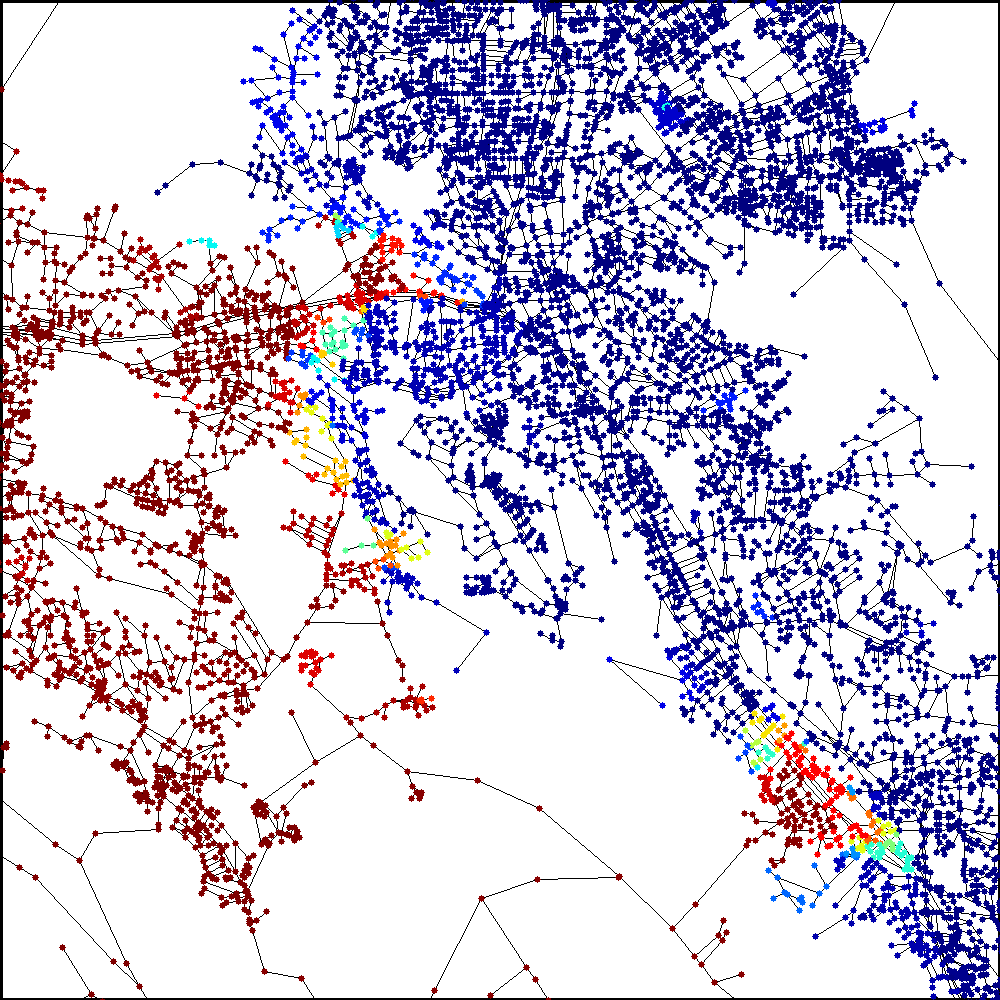}
\end{minipage}
\caption{A simulated epidemic signal over the San Francisco road network is
displayed on the far left. Three columns display (left) the true signal,
(middle) $\hat{\mu}^\W$, and (right) $\hat{\mu}^\TV$ for the three boxed areas
of the map. Both $\hat{\mu}^\W$ and $\hat{\mu}^\TV$ are estimated on noisy data
with $\sigma=0.3$, using data-tuned $\lambda$, and $\hat{\mu}^\W$ uses
effective-resistance edge weighting.}\label{fig:SF}
\end{figure}

For noise level $\sigma=0.3$, the simulated signal, $\hat{\mu}^\TV$, and
$\hat{\mu}^\W$ computed with effective-resistance edge weights are depicted in
Figure \ref{fig:SF} for the San Francisco road network at observation time
$T_2$. The most difficult regions to estimate are the signal
boundaries; we zoom in on three regions of the map where
$\hat{\mu}^\TV$ is inaccurate at these boundaries,
but $\hat{\mu}^\W$ is mostly correct. At this noise level, the st.MSE of
$\hat{\mu}^\TV$ exceeds $\hat{\mu}^\W$ by a factor of about 2.

Figure \ref{fig:networkerrors} displays st.MSE comparisons for $\hat{\mu}^\Lz$,
$\hat{\mu}^\TV$, and $\hat{\mu}^\W$ with effective-resistance weighting
at noise levels $\sigma=0.1$ to $\sigma=0.5$ in each example. We observe that
$\hat{\mu}^\W$ is not substantially worse than $\hat{\mu}^\Lz$ in any tested
setting, and that in the Gnutella and Enron digital networks where there is
large variation in effective edge resistances, $\hat{\mu}^\W$ is sometimes
substantially better. At the tested noise levels, these methods are
(with the exception of Enron $T_3$) not substantially worse than
$\hat{\mu}^\TV$, and can be substantially better in the lower noise settings.

\section{Conclusion}
We have studied estimation of piecewise-constant signals over arbitrary
graphs using an $l_0$ edge penalty, establishing minimax rate-optimal
statistical guarantees for the
local minimizer computed by an approximation algorithm for minimizing this
objective. We have shown theoretically that the same guarantees are not
necessarily achieved by total-variation denoising, and empirically that
$l_0$-penalization may be more effective in high signal-to-noise settings.
For application to networks with regions of varying connectivity, we have
proposed minimization of an edge-weighted objective, which
achieves better empirical performance in tested examples and leads
to theoretical guarantees that are spatially uniform over all graphs.

We note that while Algorithm \ref{alg:boykovouter} is provably polynomial-time,
discretization of the continuous parameter domain yields poor worst-case
runtime and may be computationally costly to extend to likelihood models with
multi-dimensional parameters. The development of faster
non-discretized algorithms is an interesting direction for future work.
Finally, our problem may be reformulated as sparse regression with particular
graph-based designs, and we believe it is an interesting question whether
similar computational ideas may be used to achieve better prediction error in
sparse regression for more general families of designs.

\appendix

\section{Properties of Algorithm \ref{alg:boykovouter}}
\label{appendix:algorithm}
In this appendix, we prove Propositions \ref{prop:runtime} and
\ref{prop:minimum}.

\begin{proof}[Proof of Proposition \ref{prop:runtime}]
The initial objective value is
$\frac{1}{2}\|Y-\bar{Y}^\delta\|^2=O(n(Y_{\max}-Y_{\min})^2)$. This value
decreases by at least $\tau$ in each outer loop of Algorithm
\ref{alg:boykovouter}, so there are at most $O(n(Y_{\max}-Y_{\min})^2/\tau)$
outer iterations. Within each outer iteration, there are at most
$(Y_{\max}-Y_{\min})/\delta$ inner loop iterations.
The augmented graph $G_{c,\hat{\mu}}$
of Algorithm \ref{alg:boykovinner} has $O(|E|)$ vertices and
edges, so solving minimum s-t cut using either Edmonds-Karp or Dinic's
algorithm requires time $O(|E|^3)$. This is the dominant cost of Algorithm
\ref{alg:boykovinner}, so combining these runtimes yields the proposition.
\end{proof}

\begin{proof}[Proof of Proposition \ref{prop:minimum}]
The objective (\ref{eq:objw}) may be written as
\[F_\w(\mu)=\frac{1}{2}\sum_{i=1}^n (Y_i-\mu_i)^2+\sum_{\{i,j\} \in E}
\lambda \w(i,j)\1\{\mu_i \neq \mu_j\},\]
and each edge-cost $\lambda \w(i,j)\1\{\mu_i \neq \mu_j\}$ satisfies the
triangle inequality in the sense
\[\lambda \w(i,j)\1\{a \neq c\}
\leq \lambda \w(i,j)\1\{a \neq b\}+\lambda \w(i,j)\1\{b \neq c\}\]
for any $a,b,c \in \R$. Hence, the first statement follows from
Theorem 5.4 and Corollary 5.5 of \cite{boykovetal}.
Termination of the algorithm implies that
$F_\w(\tilde{\mu})>F_\w(\hat{\mu})-\tau$ for every
$\delta\Z$-expansion $\tilde{\mu}$ of $\hat{\mu}$ with new value
$c \in \delta\Z \cap [Y_{\min},Y_{\max}]$, and hence also for every new value
$c \in \delta\Z$. Then $\hat{\mu}$ is a $(\tau,\delta\Z)$-local-minimizer.
\end{proof}

\section{Analysis of weighted $l_0$-denoising}\label{appendix:proofs}
In this appendix, we prove Theorems \ref{thm:oraclew}, \ref{thm:minimaxw},
and \ref{thm:nullw}. Recall from the end of Section \ref{sec:weighting} that
these imply, as direct corollaries, Theorems \ref{thm:oraclel0},
\ref{thm:minimaxupperl0}, \ref{thm:minimaxlowerl0}, and \ref{thm:nulll0}.
By rescaling, we will assume without loss of generality
that $\sigma^2=1$.

\begin{terminology*}
A {\bf multi-cut} $\bS:=\{S_1,\ldots,S_k\}$ of $G$ is a partition of
its vertices $V$ into non-empty disjoint vertex subsets
$S_1,\ldots,S_k$, which are called the {\bf elements} of the multi-cut.
In this appendix,
we do not require each element $S_\alpha$ of $\bS$ to be connected in $G$.
The multi-cut $\bS$ is {\bf trivial} if $k=1$, i.e.\ its only element is the
entire vertex set $V$. For two multi-cuts $\bS$ and $\bS'$, $\bS'$ is a
{\bf refinement} of $\bS$ if this holds in the usual sense of partitions:
For every element $S_\alpha$ of $\bS'$, there exists an element of $\bS$ that
contains $S_\alpha$ as a subset.
For any two multi-cuts $\bS$ and $\bS'$, their least common refinement is
denoted by $\bS \vee \bS'$.

For every $\mu \in \R^n$, the multi-cut {\bf induced by} $\mu$ is such that
any two vertices $i,j \in V$ are in the same
element of the cut if and only if $\mu_i=\mu_j$. (If $\mu$ is constant on $G$,
then it induces the trivial multi-cut.) A multi-cut {\bf cuts} an edge
$\{i,j\}$ if $i$ and $j$ belong to two different elements.
The set of all cut edges of $\bS$, i.e.\ the boundary of the multi-cut,
is denoted
\[\partial \bS:=\{\{i,j\} \in E: \bS \text{ cuts } \{i,j\}\}.\]
($\partial\bS$ is empty if $\bS$ is trivial.)
Recall $\w(E'):=\sum_{\{i,j\} \in E'} \w(i,j)$ for any $E' \subseteq E$. It is
then evident that if $\mu$ induces $\bS$, then
\[\|D\mu\|_\w=\w(\partial\bS).\]

Associated to each multi-cut $\bS=\{S_1,\ldots,S_k\}$ is a $k$-dimensional
subspace $K \subseteq \R^n$, such that $\mu \in K$ if and only if
$\mu$ takes a constant value on each element $S_\alpha$. (I.e., $\mu \in K$
if and only if $\bS$ is equal to, or is a refinement of, the
multi-cut induced by $\mu$.) We denote the orthogonal projection onto $K$
by $P^\bS:\R^n \to K$.
\end{terminology*}

We first recall a deterministic sub-optimality bound for any
$(\tau,\delta\Z)$-local-minimizer of (\ref{eq:objw}). Such a bound was stated
in \cite{boykovetal} as a factor-2 approximation,
$F_\w(\hat{\mu}) \leq 2F_\w(\mu^*)$, where $\hat{\mu}$ is any
$(0,\delta\Z)$-local-minimizer of (\ref{eq:objw}) and $\mu^*$ is the exact
minimizer in $(\delta\Z)^n$. In fact,
this factor of 2 applies only to the penalty term, the resulting bound
holds not only for $\mu^*$ but also for any $\mu \in (\delta\Z)^n$, and it
is easily extended to any $\tau \geq 0$. The proof is essentially the same,
and we provide it here for completeness.
\begin{lemma}\label{lemma:factor2approx}
Let $\delta>0$ and $\tau \geq 0$, and let $\hat{\mu}$ be any
$(\tau,\delta\Z)$-local-minimizer of $F_\w(\mu)$ in (\ref{eq:objw}). Then for
any vector $\mu^\delta \in (\delta\Z)^n$ whose entries take $k$ distinct values,
\[F_\w(\hat{\mu}) \leq \frac{1}{2}\|Y-\mu^\delta\|^2
+2\lambda\|D\mu^\delta\|_\w+k\tau.\]
\end{lemma}
\begin{proof}
Let $\bS=\{S_1,\ldots,S_k\}$ be the multi-cut induced by $\mu^\delta$. Fix
$\alpha \in \{1,\ldots,k\}$ and denote by $\mu^\delta_{S_\alpha}$ the
(constant) value
of $\mu^\delta$ on $S_\alpha$. Let $\tilde{\mu}$ be the $\delta\Z$-expansion
of $\hat{\mu}$ defined entrywise by
\[\tilde{\mu}_i=\begin{cases} \hat{\mu}_i & i \notin S_\alpha \\
\mu^\delta_{S_\alpha} & i \in S_\alpha.\end{cases}\]
Since $\hat{\mu}$ is a $(\tau,\delta\Z)$-local-minimizer of (\ref{eq:objw}),
by definition
\begin{align*}
F_\w(\hat{\mu})&=
\frac{1}{2}\sum_{i=1}^n (Y_i-\hat{\mu}_i)^2+\lambda\sum_{\{i,j\} \in E}
\w(i,j)\1\{\hat{\mu}_i \neq \hat{\mu}_j\}\\
&\leq \frac{1}{2}\sum_{i=1}^n (Y_i-\tilde{\mu}_i)^2+\lambda\sum_{\{i,j\} \in E}
\w(i,j)\1\{\tilde{\mu}_i \neq \tilde{\mu}_j\}+\tau=F_\w(\tilde{\mu})+\tau.
\end{align*}
As $\hat{\mu}_i=\tilde{\mu}_i$ for $i \notin S_\alpha$,
we may cancel terms from both sides to yield
\[\frac{1}{2}\sum_{i \in S_\alpha} (Y_i-\hat{\mu}_i)^2
+\lambda\!\!\!\mathop{\sum_{\{i,j\} \in E}}_{|\{i,j\} \cap S_\alpha| \geq 1}\!\!\!
\w(i,j)\1\{\hat{\mu}_i \neq \hat{\mu}_j\}
\leq \frac{1}{2}\sum_{i \in S_\alpha} (Y_i-\mu^\delta_{S_\alpha})^2
+\lambda\!\!\!\mathop{\sum_{\{i,j\} \in E}}_{|\{i,j\} \cap S_\alpha| \geq 1}\!\!\!
\w(i,j)\1\{\tilde{\mu}_i \neq \tilde{\mu}_j\}+\tau.\]
For the right side above,
note that $\1\{\tilde{\mu}_i \neq \tilde{\mu}_j\}=0$ when $i,j \in
S_\alpha$, and apply the trivial
bound $\1\{\tilde{\mu}_i \neq \tilde{\mu}_j\} \leq 1$
when $|\{i,j\} \cap S_\alpha|=1$, i.e.\ when $\{i,j\}$ is cut by $\bS$.
Then, summing the above over $\alpha=1,\ldots,k$,
\begin{align*}
\frac{1}{2}\|Y-\hat{\mu}\|^2+\lambda\sum_{\alpha=1}^k
\!\!\!\mathop{\sum_{\{i,j\} \in E}}_{|\{i,j\} \cap S_\alpha| \geq 1}\!\!\!
\w(i,j)\1\{\hat{\mu}_i \neq \hat{\mu}_j\}
&\leq \frac{1}{2}\|Y-\mu^\delta\|^2+\lambda\sum_{\alpha=1}^k\!\!\!
\mathop{\sum_{\{i,j\} \in E}}_{|\{i,j\} \cap S_\alpha|=1}\!\!\!
\w(i,j)+k\tau\\
&=\frac{1}{2}\|Y-\mu^\delta\|^2+2\lambda\|D\mu^\delta\|_\w+k\tau,
\end{align*}
where the second line follows from noting that each edge $\{i,j\}$ cut by $\bS$
appears twice in the sum.
The desired result then follows by noting that each edge $\{i,j\} \in E$
appears at least once in the double sum on the left side.
\end{proof}

Next, we control $\|P^\bS\eps\|^2$ using the weighted cut-size $\w(\partial\bS)$
uniformly over all multi-cuts $\bS$ of $G$. As the number of distinct
multi-cuts of $G$ may be very large, we will first
condition on a random spanning tree $T$ of $G$, consider the common
refinement of those multi-cuts that cut the same set of edges of $T$, and
take a union bound over these refinements. The distribution of the random
spanning tree $T$ that we take is the one specified by the
edge-weighting $\r \in \ST(G)$ such that $\w \geq \r$.

\begin{lemma}\label{lemma:badcut}
Let $\w:E \to \R_+$ be such that $\w \geq \r$ for some $\r \in \ST(G)$, and
let $\eps \in \R^n$ have coordinates
$\eps_i \overset{iid}{\sim} \Normal(0,1)$. Then there exist universal
constants $C,C'>0$ such that with probability at least $1-C'\w(E)^{-3}$, for
every multi-cut $\bS$ of $G$,
\[\|P^\bS \eps\|^2 \leq C\max(\w(\partial\bS),1)\log
\frac{e\cdot\w(E)}{\max(\w(\partial\bS),1)}.\]
\end{lemma}
\begin{proof}
If $\bS$ is the trivial multi-cut consisting of the single element
$\{1,\ldots,n\}$, then $\w(\partial\bS)=0$ and $\|P^\bS\eps\|^2 \sim
\chi^2_1$, so $\|P^\bS\eps\|^2 \leq C\log (e\cdot\w(E))$
with probability at least $1-C'\w(E)^{-3}$ by a standard Gaussian tail bound.

Consider, then, non-trivial multi-cuts $\bS$.
By (\ref{eq:wcut}), $\w(\partial\bS) \geq 1$. Observe that for any $c>0$,
the function $x \mapsto x\log (ec/x)$ is increasing over
$x \in (0,c)$. Then, as $\w \geq \r$,
\[C\w(\partial \bS)\log \frac{e\cdot\w(E)}{\w(\partial\bS)}
\geq C\r(\partial \bS)\log \frac{e\cdot\w(E)}{\r(\partial\bS)},\]
and it suffices to establish the stronger bound given by the right side above.

For any $\eta>0$, we may bound
\begin{align*}
p_n&:=\P\left[\text{ there exists a non-trivial multi-cut }\bS:
\|P^{\bS}\eps\|^2>C\r(\partial\bS)\log
\frac{e\cdot\w(E)}{\r(\partial\bS)}\right]\\
&=\P\left[\sup_{\bS} \|P^{\bS}\eps\|^2-C\r(\partial\bS)\log
\frac{e\cdot\w(E)}{\r(\partial\bS)}>0\right]\\
&=\P\left[\sup_{\bS} \exp\left(\eta\|P^{\bS}\eps\|^2-\eta C\r(\partial\bS)
\log \frac{e\cdot\w(E)}{\r(\partial\bS)}\right)>1\right]\\
&\leq \E\left[\sup_{\bS} \exp\left(\eta\|P^{\bS}\eps\|^2
-\eta C\r(\partial\bS)\log\frac{e\cdot\w(E)}{\r(\partial\bS)}\right)\right],
\end{align*}
where the suprema are over all non-trivial multi-cuts $\bS$
of $G$. Corresponding to the weighting $\r:E \to \R_+$ is a distribution
over random spanning trees $T$ of $G$, satisfying the property
(\ref{eq:expectedtreeedges}).
Let $\E_\eps$ and $\E_T$ denote the expectations over $\eps$ and over $T$.
Then the above yields
\[p_n\leq \E_\eps\left[\sup_{\bS} \exp\left(\eta\|P^{\bS}\eps\|^2
-\eta C\,\E_T[|T \cap \partial \bS|]
\log\frac{e\cdot\w(E)}{\E_T[|T \cap \partial \bS|]}\right)\right].\]
Observe that for any $c>0$, $x \mapsto x \log (ec/x)$ is concave over
$x \in (0,c)$, so $x \mapsto \exp(-\eta Cx \log (ec/x))$ is convex over
$x \in (0,c)$. Also, $|T \cap \partial\bS| \leq |T|=n-1=\r(E) \leq \w(E)$. Then
applying Jensen's inequality,
\begin{align}
p_n&\leq \E_\eps\left[\sup_{\bS} \E_T\left[
\exp\left(\eta \|P^{\bS}\eps\|^2-\eta
C|T\cap\partial\bS|\log\frac{e\cdot\w(E)}{|T\cap\partial\bS|}\right)\right]\right]
\nonumber\\
&\leq \E_\eps\E_T\left[\sup_\bS\exp\left(\eta \|P^{\bS}\eps\|^2-\eta
C|T\cap\partial\bS|\log\frac{e\cdot\w(E)}{|T\cap\partial\bS|}\right)\right].
\label{eq:pn}
\end{align}

Fix $T$, and note that $|T \cap \partial \bS| \geq 1$ for any
non-trivial multi-cut $\bS$. We now control the supremum over $\bS$ by a
union bound over all possible non-empty edge sets $T \cap \partial \bS$:
\begin{align*}
&\sup_\bS \exp\left(\eta \|P^{\bS}\eps\|^2-\eta
C|T\cap\partial\bS|\log\frac{e\cdot\w(E)}{|T\cap\partial\bS|}\right)\\
&\leq \sum_{E' \subseteq T:|E'| \geq 1}\;
\sup_{\bS:\,T \cap \partial \bS=E'}\;
\exp\left(\eta \|P^{\bS}\eps\|^2-\eta C|E'|\log \frac{e\cdot\w(E)}{|E'|}
\right)
\end{align*}
Consider any fixed non-empty subset $E' \subseteq T$,
and denote by $\bS^*:=\bS^*(E',T)$ the multi-cut
obtained in the following way: Remove the edges $E'$ from $T$, and let the
elements of $\bS^*$ be the remaining connected components in
the graph with edges $T \setminus E'$. If $\bS$ is any other multi-cut
satisfying $T \cap \partial \bS=E'$, then any two vertices in the same element
of $\bS^*$ are connected by a path of edges in $T$ that are not cut by
$\bS$, and hence these vertices belong to the same element of $\bS$.
Thus $\bS^*$ is a refinement of $\bS$. Then the range
of the projection $P^\bS$ is a linear subspace of the range of $P^{\bS^*}$,
and we consequently have $\|P^{\bS^*}\eps\|^2 \geq \|P^{\bS}\eps\|^2$
(deterministically for any $\eps$). Then
\[p_n\leq \E_T\left[\sum_{E' \subseteq T:|E'| \geq 1}
e^{-\eta C|E'|\log \frac{e\cdot\w(E)}{|E'|}}\E_\eps\left[e^{\eta
\|P^{\bS^*}\eps\|^2}\right]\right].\]
$\bS^*$ has $|E'|+1$ elements, so $\|P^{\bS^*}\eps\|^2 \sim \chi^2_{|E'|+1}$.
Then for $\eta \in (0,1/2)$, the chi-squared moment-generating-function yields
$\E_\eps[e^{\eta \|P^{\bS^*}\eps\|^2}]=(1-2\eta)^{|E'|+1}$, so
\[p_n \leq \E_T\left[\sum_{l=1}^{n-1} \sum_{E' \subseteq T:|E'|=l}
e^{-\eta Cl\log \frac{e\cdot\w(E)}{l}}(1-2\eta)^{l+1}\right].\]
Applying $|\{E' \subseteq T:|E'|=l\}|=\binom{n-1}{l} \leq \exp(l \log
\frac{e(n-1)}{l}) \leq \exp(l \log \frac{e\cdot\w(E)}{l})$ for every spanning tree
$T$,
\[p_n \leq \sum_{l=1}^{n-1} e^{(-\eta C+1)l\log \frac{e\cdot\w(E)}{l}}
(1-2\eta)^{l+1}
\leq \sum_{l=1}^{\lfloor \w(E) \rfloor}
e^{(-\eta C+1)l\log \frac{e\cdot\w(E)}{l}}(1-2\eta)^{l+1}.\]
Fixing $\eta$ to be a constant (say $\eta=1/4$), it is easily verified
that for $C>1/\eta$, the function
\[l \mapsto (-\eta C+1)l\log \frac{e\cdot\w(E)}{l}+(l+1)\log(1-2\eta)\]
is convex, so its maximum over $l \in [1,\w(E)]$ is attained at either $l=1$ or
$l=\w(E)$. Thus
\[p_n \leq \w(E)\max\left((e\cdot\w(E))^{-\eta C+1}(1-2\eta)^2,
e^{-(\eta C+1)\w(E)}(1-2\eta)^{\w(E)+1}\right) \leq C'\w(E)^{-3}\]
for sufficiently large constants $C,C'>0$, as desired.
\end{proof}

The preceding two lemmas yield risk bounds for $\hat{\mu}$ in probability
via a standard Cauchy-Schwarz argument, see
e.g.\ \cite[Theorem 2]{birgemassart2001}.
\begin{lemma}\label{lemma:oracleinequalityprob}
Let $\delta>0$ and $\tau \geq 0$,
and let $\w:E \to \R_+$ be such that $\w \geq \r$ for some $\r \in \ST(G)$.
\begin{enumerate}[(a)]
\item For any $\eta>0$, there exist constants $C,C',C''>0$ depending only on
$\eta$ such that for
any true signal $\mu_0 \in \R^n$, any (fixed) $\mu \in \R^n$, and
any $\lambda \geq C\log \w(E)$, with probability
at least $1-C'\w(E)^{-3}$ all $(\tau,\delta\Z)$-local-minima
$\hat{\mu}$ of (\ref{eq:objw}) satisfy
\[\|\hat{\mu}-\mu_0\|^2 \leq (1+\eta)\|\mu-\mu_0\|^2+C''\left((\lambda+\tau)
\max(\|D\mu\|_\w,1)+n\delta^2\right).\]
\item There exist universal constants $C,C',C''>0$ such that for any $s \in
[1,\w(E)]$, any true signal $\mu_0 \in \R^n$ with $\|D\mu_0\|_\w \leq s$, and
any $\lambda \geq C(1+\log\frac{\w(E)}{s})$, with probability at least
$1-C'\w(E)^{-3}$ all $(\tau,\delta\Z)$-local-minima $\hat{\mu}$ of
(\ref{eq:objw})
satisfy
\[\|\hat{\mu}-\mu_0\|^2 \leq C''((\lambda+\tau)s+n\delta^2).\]
\end{enumerate}
\end{lemma}
\begin{proof}
For part (a), let $\mu^\delta \in (\delta \Z)^n$ be the vector obtained by 
rounding each entry of $\mu$ to the nearest value in $\delta\Z$, and let
$\hat{\mu}$ be any $(\tau,\delta\Z)$-local-minimizer. Recall that if the entries
of $\mu^\delta$ take $k$ distinct values, then $\|D\mu^\delta\|_\w \geq k-1$
by (\ref{eq:wcut}). Then by Lemma \ref{lemma:factor2approx},
\[\frac{1}{2}\|Y-\hat{\mu}\|^2+\lambda \|D\hat{\mu}\|_\w
\leq \frac{1}{2}\|Y-\mu^\delta\|^2+2\lambda\|D\mu^\delta\|_\w
+\tau(\|D\mu^\delta\|_\w+1).\]
Writing $Y=\mu_0+\eps$ and canceling $\|\eps\|^2/2$ from both sides,
\begin{equation}
\frac{1}{2}\|\hat{\mu}-\mu_0\|^2+\langle \mu_0-\hat{\mu},\eps \rangle
+\lambda \|D\hat{\mu}\|_\w \leq \frac{1}{2}\|\mu^\delta-\mu_0\|^2
+\langle \mu_0-\mu^\delta,\eps \rangle+2\lambda\|D\mu^\delta\|_\w
+\tau(\|D\mu^\delta\|_\w+1).
\label{eq:oracleprob1}
\end{equation}

Suppose $\mu^\delta$ induces the multi-cut $\bS^\delta$ of $G$, and
$\hat{\mu}$ induces the multi-cut $\hat{\bS}$ of $G$.
Denote by $\bS^\delta \vee \hat{\bS}$ the least common refinement of
$\bS^\delta$ and $\hat{\bS}$, so
\[\langle \hat{\mu}-\mu^\delta,\eps \rangle
=\langle P^{\bS^\delta \vee \hat{\bS}}(\hat{\mu}-\mu^\delta),\eps \rangle
=\langle \hat{\mu}-\mu^\delta, P^{\bS^\delta \vee \hat{\bS}}\eps \rangle.\]
For any positive constants $c_1,c_2$ such that $c_1<\frac{1}{4}$ and
$c_1c_2=\frac{1}{4}$, applying $xy \leq c_1x^2+c_2y^2$ and $(x+y)^2 \leq
2(x^2+y^2)$ yields
\begin{align*}
\langle \hat{\mu}-\mu^\delta,\eps \rangle
&\leq (\|\hat{\mu}-\mu_0\|+\|\mu^\delta-\mu_0\|)\|P^{\bS^\delta \vee
\hat{\bS}}\eps\|\\
&\leq 2c_1(\|\hat{\mu}-\mu_0\|^2+\|\mu^\delta-\mu_0\|^2)
+c_2\|P^{\bS^\delta \vee \hat{\bS}}\eps\|^2.
\end{align*}
Applying this to (\ref{eq:oracleprob1}) and rearranging yields
\begin{equation}\label{eq:oracleprob2}
\left(\frac{1}{2}-2c_1\right)\|\hat{\mu}-\mu_0\|^2
\leq \left(\frac{1}{2}+2c_1\right)\|\mu^\delta-\mu_0\|^2
+c_2\|P^{\bS^\delta \vee \hat{\bS}}\eps\|^2
+2\lambda\|D\mu^\delta\|_\w+\tau(\|D\mu^\delta\|_\w+1)-\lambda\|D\hat{\mu}\|_\w.
\end{equation}

By Lemma \ref{lemma:badcut}, there exist constants $C,C'>0$ and an event of
probability $1-C'\w(E)^{-3}$ on which we are guaranteed
\[\|P^{\bS^\delta \vee \hat{\bS}}\eps\|^2
\leq C\max(\w(\partial(\bS^\delta \vee \hat{\bS})),1)
\log \frac{e\cdot\w(E)}{\max(\w(\partial(\bS^\delta \vee \hat{\bS})),1)}.\]
Note that $\partial(\bS^\delta \vee \hat{\bS}) \subseteq
\partial \bS^\delta \cup \partial \hat{\bS}$, so
\begin{equation}\label{eq:wrefinement}
\w(\partial(\bS^\delta \vee \hat{\bS}))
\leq \w(\partial \bS^\delta)+\w(\partial \hat{\bS})
=\|D\mu^\delta\|_\w+\|D\hat{\mu}\|_\w.
\end{equation}
Then on this event,
\[\|P^{\bS^\delta \vee \hat{\bS}}\eps\|^2
\leq C(\max(\|D\mu^\delta\|_\w,1)+\|D\hat{\mu}\|_\w)\log (e\cdot\w(E)).\]
If $\lambda \geq c_2C\log(e\cdot\w(E))$, then applying this to
(\ref{eq:oracleprob2}) yields (on this event)
\[\|\hat{\mu}-\mu_0\|^2 \leq \frac{1+4c_1}{1-4c_1}\|\mu^\delta-\mu_0\|^2
+C''\left((\lambda+\tau)\max(\|D\mu^\delta\|_\w,1)\right)\]
for some constant $C''>0$. Part (a) follows by noting $\|D\mu^\delta\|_\w
\leq \|D\mu\|_\w$,
\[\|\mu^\delta-\mu_0\|^2 \leq (1+c_1)\|\mu-\mu_0\|^2+(1+c_1^{-1})
\|\mu-\mu^\delta\|^2 \leq (1+c_1)\|\mu-\mu_0\|^2+(1+c_1^{-1})n\delta^2,\]
and taking $c_1$ sufficiently small.

For part (b), let us apply the preceding arguments for $\mu=\mu_0$, so that
$\mu^\delta$ is obtained by rounding the entries of the true signal $\mu_0$.
Then $\|\mu^\delta-\mu_0\|^2 \leq n\delta^2$, $\|D\mu^\delta\|_\w
\leq \|D\mu_0\|_\w \leq s$, and (\ref{eq:wrefinement}) implies
$\max(\w(\partial(\bS^\delta \vee \hat{\bS})),1) \leq s+\|D\hat{\mu}\|_\w$.
Observing that $s+\|D\hat{\mu}\|_\w \leq 2\w(E)$ and that $x \mapsto
x\log(ec/x)$ is increasing for $x \in (0,c)$, on an event of
probability $1-C'\w(E)^{-3}$ we have
\begin{align*}
\|P^{\bS^\delta \vee \hat{\bS}}\eps\|^2
&\leq C\max(\w(\partial(\bS^\delta \vee \hat{\bS})),1)
\log \frac{2e\cdot\w(E)}{\max(\w(\partial(\bS^\delta \vee \hat{\bS})),1)}\\
&\leq C(s+\|D\hat{\mu}\|_\w)\log \frac{2e\cdot\w(E)}{s+\|D\hat{\mu}\|_w}
\leq C(s+\|D\hat{\mu}\|_\w)\log \frac{2e\cdot\w(E)}{s}.
\end{align*}
If $\lambda \geq c_2C\log (2e\cdot\w(E)/s)$, then
applying this to (\ref{eq:oracleprob2}) yields
$\|\hat{\mu}-\mu_0\|^2 \leq C''((\lambda+\tau)s+n\delta^2)$ for a constant
$C''>0$, as desired.
\end{proof}

Finally, we obtain bounds in expectation by
applying H\"older's inequality and a crude bound on the $q$th power of
the squared-error risk for some $q>1$.
\begin{lemma}\label{lemma:qmoment}
Let $\delta>0$, $\tau \geq 0$, and $\hat{\mu}$ be any
$(\tau,\delta\Z)$-local-minimizer of (\ref{eq:objw}). Then for any $q>1$, there
exists a constant $C_q$ depending only on $q$ such that
\[\E[\|\hat{\mu}-\mu_0\|^{2q}] \leq C_q(\lambda\w(E)+n+n\delta^2+n\tau)^q.\]
\end{lemma}
\begin{proof}
Let $\mu^\delta \in (\delta\Z)^n$ denote the vector obtained by rounding each
entry of $\mu_0$ to the nearest value in $\delta\Z$. For each vertex
$i \in \{1,\ldots,n\}$, the vector $\tilde{\mu}$ obtained by replacing
$\hat{\mu}_i$ with $\mu^\delta_i$ and keeping all other coordinates
of $\hat{\mu}$ the same is a $\delta\Z$-expansion of $\hat{\mu}$. Since
$\hat{\mu}$ is a $(\tau,\delta\Z)$-local-minimizer, this implies
$F_\w(\hat{\mu}) \leq F_\w(\tilde{\mu})+\tau$ and hence (cancelling terms not
depending on the $i$th coordinate)
\[\frac{1}{2}(\hat{\mu}_i-Y_i)^2 \leq \frac{1}{2}(\mu^\delta_i-Y_i)^2
+\lambda\sum_{j:\{i,j\} \in E} \w(i,j)+\tau.\]
Summing over $i=1,\ldots,n$ and applying
$\|\mu^\delta-Y\|^2 \leq 2\|\mu_0-Y\|^2+2\|\mu^\delta-\mu_0\|^2
\leq 2\|\mu_0-Y\|^2+2n\delta$,
\[\frac{1}{2}\|\hat{\mu}-Y\|^2 \leq \frac{1}{2}\|\mu^\delta-Y\|^2
+2\lambda\w(E)+n\tau \leq \|\mu_0-Y\|^2+2\lambda\w(E)+n\delta^2+n\tau.\]
Then, applying $(a+b)^q \leq 2^{q-1}(a^q+b^q)$,
\begin{align*}
\|\hat{\mu}-\mu_0\|^{2q}
&\leq 2^{2q-1}(\|\hat{\mu}-Y\|^{2q}+\|\mu_0-Y\|^{2q})\\
&\leq 2^{2q-1}\left(2^{q-1}\left((2\|\mu_0-Y\|^2)^q
+(4\lambda\w(E)+2n\delta^2+2n\tau)^q\right)+\|\mu_0-Y\|^{2q}\right)\\
&\leq C_q(\|\mu_0-Y\|^{2q}+(\lambda \w(E)+n\delta^2+n\tau)^q)
\end{align*}
for a constant $C_q>0$. Since $\|\mu_0-Y\|^2 \sim \chi^2_n$, this implies
\[\E[\|\hat{\mu}-\mu_0\|^{2q}] \leq C_q'(\lambda\w(E)+n+n\delta^2+n\tau)^q\]
for a different constant $C_q'>0$.
\end{proof}

\begin{proof}[Proof of Theorem \ref{thm:oraclew}]
For $\mu \in \R^n$ and $\lambda,\eta,C''>0$, denote
\[M(\mu,\lambda,\eta,C''):=
(1+\eta)\|\mu-\mu_0\|^2+C''\lambda\max(\|D\mu\|_\w,1).\]
Note that $n\delta^2+\tau\max(\|D\mu\|_\w,1) \leq 2\lambda\max(\|D\mu\|_\w,1)$
whenever $\delta \leq 1/\sqrt{n}$, $\tau \leq 1$, and $\lambda \geq 1$.
Applying H\"older's inequality and Lemma \ref{lemma:oracleinequalityprob}(a),
there exist constants $C,C',C''>0$ such that
for any $\mu \in \R^n$, any $\lambda \geq C\log \w(E)$, and any
constants $p,q>1$ such that $\frac{1}{p}+\frac{1}{q}=1$,
\begin{align*}
\E[\|\hat{\mu}-\mu_0\|^2] &\leq M(\mu,\lambda,\eta,C'')
+\E\left[\|\hat{\mu}-\mu_0\|^2 
\1\{\|\hat{\mu}-\mu_0\|^2>M(\mu,\lambda,\eta,C'')\}\right]\\
&\leq M(\mu,\lambda,\eta,C'')+\left(\E[\|\hat{\mu}-\mu_0\|^{2q}]\right)^{1/q}
\left(\P[\|\hat{\mu}-\mu_0\|^2>M(\mu,\lambda,\eta,C'')]\right)^{1/p}\\
&\leq M(\mu,\lambda,\eta,C'')+\left(\E[\|\hat{\mu}-\mu_0\|^{2q}]\right)^{1/q}
(C'\w(E)^{-3})^{1/p}.
\end{align*}
Note that $\w(E) \geq \r(E)=n-1$, so $\lambda \w(E) \geq n+n\delta^2+n\tau$
for any $\lambda \geq 3$. Then, choosing $p=3$, $q=3/2$, and applying
Lemma \ref{lemma:qmoment},
\[\E[\|\hat{\mu}-\mu_0\|^2] \leq M(\mu,\lambda,\eta,C'')+C'''\lambda\]
for a constant $C'''>0$. Taking the infimum over $\mu \in \R^n$ yields the
desired result.
\end{proof}

\begin{proof}[Proof of Theorem \ref{thm:minimaxw}]
The proof of the upper bound is the same as that of Theorem \ref{thm:oraclew},
using Lemma \ref{lemma:oracleinequalityprob}(b) instead of
Lemma \ref{lemma:oracleinequalityprob}(a):
For universal constants $C,C',C'',C'''>0$, any $\delta \leq 1/\sqrt{n}$,
$\tau \leq 1$, $\lambda \geq C(1+\log\frac{\w(E)}{s})$, and $p=3$ and $q=3/2$,
\[\E[\|\hat{\mu}-\mu_0\|^2] \leq C''\lambda s+\left(\E[\|\hat{\mu}-\mu_0\|^{2q}]
\right)^{1/q}(C'\w(E)^{-3})^{1/p} \leq C'''\lambda s.\]

For the lower bound, denote $\w(i)=\sum_{j:\{i,j\} \in E} \w(i,j)$, and consider
the vertex subset $V'=\{i:\w(i) \leq 4\w(E)/n\}$. Then the class
$\{\mu_0 \in \R^n:\|D\mu_0\|_\w \leq s\}$ contains the class of sparse vectors
\[B_0:=\left\{\mu_0 \in \R^n:\supp(\mu_0) \subseteq V',\;\|\mu_0\|_0 \leq
\left\lfloor \frac{sn}{4\w(E)}\right\rfloor \right\},\]
where $\supp(\mu_0)$ denotes the set of vertices $i$ for which
$\mu_{0,i} \neq 0$. (This is because for any $\mu_0 \in B_0$, $\|D\mu_0\|_\w
\leq \sum_{i \in \supp(\mu_0)} \w(i) \leq \lfloor \frac{sn}{4\w(E)}\rfloor
\frac{4\w(E)}{n} \leq s$.) Then
\[\inf_{\hat{\mu}} \sup_{\mu_0:\|D\mu_0\|_\w \leq s} \E[\|\hat{\mu}-\mu_0\|^2]
\geq \inf_{\hat{\mu}} \sup_{\mu_0 \in B_0} \E[\|\hat{\mu}-\mu_0\|^2].\]
For $s \geq 4\w(E)/n$, the lower-bound for the sparse normal-means
problem (see e.g.\ Theorem 1(b) of \cite{raskuttietal} in the case
$X=\sqrt{n}\operatorname{Id}$) yields, for a universal constant $c>0$,
\[\inf_{\hat{\mu}} \sup_{\mu_0 \in B_0} \E[\|\hat{\mu}-\mu_0\|^2]
\geq c\left\lfloor \frac{sn}{4\w(E)}\right\rfloor
\log \frac{|V'|}{\lfloor \frac{sn}{4\w(E)}\rfloor}
\geq c\frac{sn}{8\w(E)}\log \frac{4\w(E)|V'|}{sn}.\]
Note that $\w(i)>4\w(E)/n$ for all $i \notin V'$, so 
\[2\w(E)=\sum_{i=1}^n \w(i)>\frac{4\w(E)}{n}(n-|V'|).\]
Then $|V'|>n/2$, and the lower bound follows.
\end{proof}

\begin{proof}[Proof of Theorem \ref{thm:nullw}]
Let $\mu^\delta \in (\delta\Z)^n$ denote the (constant) vector obtained
by rounding the value in $\mu_0$ to the closest value in $\delta\Z$.
As $\mu^\delta$ is a $\delta\Z$-expansion of $\hat{\mu}$, and $\hat{\mu}$ is
a $(\tau,\delta\Z)$-local-minimizer, we have
\[F_\w(\hat{\mu})=\frac{1}{2}\|Y-\hat{\mu}\|^2+\lambda\|D\hat{\mu}\|_\w
\leq \frac{1}{2}\|Y-\mu^\delta\|^2+\tau=F_\w(\mu^\delta)+\tau.\]
Denoting by $\bS$ the multi-cut induced by $\hat{\mu}$,
the same steps as leading to (\ref{eq:oracleprob2}) yield, for any positive
constants $c_1<\frac{1}{4}$ and $c_2>1$ such that $c_1c_2=\frac{1}{4}$,
\[\left(\frac{1}{2}-2c_1\right)\|\hat{\mu}-\mu_0\|^2
\leq \left(\frac{1}{2}+2c_1\right)n\delta^2+c_2\|P^{\bS}\eps\|^2+\tau
-\lambda\|D\hat{\mu}\|_\w.\]
Hence, for $\delta \leq 1/\sqrt{n}$ and $\tau \leq 1$,
\[\lambda\|D\hat{\mu}\|_\w \leq C''(1+\|P^\bS\eps\|^2)\]
for a constant $C''>0$. Lemma \ref{lemma:badcut} implies, with
probability at least $1-C'\w(E)^{-3} \geq 1-C'(n-1)^{-3}$,
\[\|P^\bS\eps\|^2 \leq C\max(\w(\partial\bS),1)\log(e\cdot\w(E))
=C\max(\|D\hat{\mu}\|_\w,1)\log(e\cdot\w(E)).\]
Recall by (\ref{eq:wcut}) that $\|D\hat{\mu}\|_\w \geq 1$ if
$\|D\hat{\mu}\|_\w \neq 0$. Thus, if $\lambda>C''(1+C\log(e\cdot\w(E)))$, then
the above two statements imply $\|D\hat{\mu}\|_\w=0$, and hence $\hat{\mu}$ is
constant with probability at least $1-C'(n-1)^{-3}$, as desired.
\end{proof}

\section{Total-variation lower bound for the linear chain}\label{appendix:TV}
In this appendix, we prove Theorems \ref{thm:TVnull} and
\ref{thm:TValternative}. By rescaling, we will assume $\sigma^2=1$. We say
that $k$ consecutive intervals $\bS:=(S_1,\ldots,S_k)$ partition $(1,\ldots,n)$
if there are integers $1=i_1<i_2<\ldots<i_k<i_{k+1}=n+1$ such that
$S_\alpha=(i_\alpha,i_\alpha+1,\ldots,i_{\alpha+1}-1)$ for each $\alpha$.
We associate to each such partition $\bS$ two orthogonal
projections $P^\bS:\R^n \to \R^n$ and $Q^\bS:\R^n \to \R^n$,
such that $P^\bS$ projects onto the $k$-dimensional subspace of
vectors assuming a constant value on each $S_\alpha$, and $Q^\bS$ projects
onto the $(n-k)$-dimensional
orthogonal complement of this subspace. More formally, for each
$\alpha$ and each $i \in S_\alpha$, $P^\bS$ is defined by
$(P^\bS \mu)_i=|S_\alpha|^{-1}\sum_{j \in S_\alpha} \mu_j$, and
$Q^\bS=\Id-P^\bS$. We denote by $P^\bS(\R^n)$ the range of $P^\bS$, and for any
$v \in P^\bS(\R^n)$ we denote by $v_{S_\alpha}$ the (constant) value of $v$ over
$S_\alpha$.

We will say that $\mu \in \R^n$ induces such a partition
$\bS=(S_1,\ldots,S_k)$ and a sign vector
$\bs=(s_1,\ldots,s_{k-1}) \in \{-1,1\}^{k-1}$ if $\mu \in P^\bS(\R^n)$ and
$\sign(\mu_{S_{\alpha+1}}-\mu_{S_{\alpha}})=s_\alpha$ for each
$\alpha=1,\ldots,k-1$. Recall the vertex-edge incidence matrix $D \in \R^{n
\times (n-1)}$: In this appendix, for the linear chain graph, let us fix the
sign convention for $D$ so that
$D\mu=(\mu_2-\mu_1,\mu_3-\mu_2,\ldots,\mu_n-\mu_{n-1})$.

The following lemma is an implication of the subgradient condition for
$\hat{\mu}$ minimizing (\ref{eq:objTV}); a similar result was stated as Lemmas
2.1 and A.1 of \cite{rinaldo}.
\begin{lemma}\label{lemma:subgradient}
Let $\lambda>0$. Fix $k$ consecutive intervals $\bS:=(S_1,\ldots,S_k)$ that
partition $(1,\ldots,n)$, and fix a sign vector $\bs:=(s_1,\ldots,s_{k-1})
\in \{-1,1\}^{k-1}$. Define $v \in P^\bS(\R^n)$
such that, for each $\alpha=1,\ldots,k$,
\begin{equation}\label{eq:v}
v_{S_\alpha}=\frac{\lambda(s_\alpha-s_{\alpha-1})}{|S_\alpha|},
\end{equation}
where we set $s_0=s_k:=0$. Define also
\begin{align*}
W&:=\{w \in [-1,1]^{n-1}:w_i=s_\alpha \text{ if }
i \in S_\alpha \text{ and }
i+1 \in S_{\alpha+1} \text{ for some } \alpha \in \{1,\ldots,k-1\}\},\\
K&:=\{x \in P^\bS(\R^n):\sign(x_{S_{\alpha+1}}-x_{S_{\alpha}})=s_\alpha
\text{ for all } \alpha=1,\ldots,k-1\}.
\end{align*}
Then the following two events are equivalent:
\begin{enumerate}
\item The minimizer $\hat{\mu}$ of (\ref{eq:objTV})
induces the partition $\bS$ and the signs $\bs$.
\item $P^\bS Y+v \in K$ and $Q^\bS Y-v \in \lambda D^TW$.
\end{enumerate}
Furthermore, if this event holds, then $\hat{\mu}=P^\bS Y+v$.
\end{lemma}
\begin{proof}
Suppose $\hat{\mu}$ minimizes (\ref{eq:objTV}) and
induces $\bS$ and $\bs$. Then $W$ is precisely the
subdifferential of the $l_1$-norm at $D\hat{\mu}$, so $\lambda D^TW$ is the
subdifferential of $\lambda \|D\hat{\mu}\|_1$ with respect to $\hat{\mu}$. Since
$\hat{\mu}$ minimizes (\ref{eq:objTV}), it satisfies the subgradient condition
$Y-\hat{\mu} \in \lambda D^TW$. Apply $P^\bS$ to both sides, noting that
$P^\bS \hat{\mu}=\hat{\mu}$ and
$P^\bS x=-v$ for any $x \in \lambda D^TW$. Then $\hat{\mu}=P^\bS Y+v$. Since
$\hat{\mu}$ induces the signs $\bs$, we have $P^\bS Y+v\in K$. Also,
$Y-\hat{\mu}=Q^\bS Y-v \in \lambda D^TW$, as desired.

Conversely, suppose $P^\bS Y+v \in K$ and $Q^\bS Y-v \in \lambda D^TW$.
Defining $\hat{\mu}=P^\bS Y+v$, the first condition implies that $\hat{\mu}$
induces $\bS$ and $\bs$, while the second implies $Y-\hat{\mu} \in \lambda
D^TW$. Since $\lambda D^TW$ is the subdifferential of $\lambda\|D\hat{\mu}\|_1$ 
with respect to $\hat{\mu}$, this implies that $\hat{\mu}$ minimizes
(\ref{eq:objTV}). (The objective (\ref{eq:objTV}) is strictly convex, so
$\hat{\mu}$ is the unique minimizer.)
\end{proof}

Using this characterization, we may lower-bound the number of intervals in the
partition induced by $\hat{\mu}$ when the true signal $\mu_0$ is 0:
\begin{proof}[Proof of Theorem \ref{thm:TVnull}(a)]
For a fixed partition $\bS:=(S_1,\ldots,S_k)$ into $k$ consecutive
intervals, let us first bound
\[p_\bS:=\P\left[Q^\bS Y \in \lambda Q^\bS D^T[-1,1]^{n-1}\right].\]
Denoting by $\sB_m$ the Euclidean unit ball in $\R^m$ centered at 0,
clearly
\[p_\bS \leq \P[Q^\bS Y \in \lambda\sqrt{n}Q^\bS D^T\sB_{n-1}].\]
The range of $D^T$ is the $(n-1)$-dimensional space orthogonal to the all-1's
vector, and $Q^{\bS}$ is a projection onto an $(n-k)$-dimensional subspace of
this range---hence
$Q^{\bS}D^T$ has rank $n-k$. Let $Q^{\bS}D^T=U\Gamma V^T$ denote the (reduced)
singular value decomposition of $Q^\bS D^T$, where
$\Gamma=\diag(\gamma_1,\ldots,\gamma_{n-k})$
contains the singular values $\gamma_1 \leq \ldots \leq \gamma_{n-k}$,
and $U \in \R^{n \times (n-k)}$ and $V \in \R^{(n-1) \times (n-k)}$ have
orthonormal columns. Then, as $V^T\sB_{n-1}=\sB_{n-k}$, we have
\[p_\bS \leq \P[U^TQ^{\bS}Y \in \lambda\sqrt{n}\Gamma \sB_{n-k}].\]
The set $\mathcal{E}:=\lambda\sqrt{n}\Gamma \sB_{n-k}$ is an ellipsoid in
$\R^{n-k}$, whose principal axes are aligned with the standard basis and
have lengths $\lambda\sqrt{n}\gamma_i$. The vector
$\eps':=U^TQ^\bS Y$ of length $n-k$
has i.i.d.\ entries distributed as $\Normal(0,1)$ (when
$\mu_0=0$.) Then, for any $l \in \{1,\ldots,n-k\}$,
letting $\eps_{1:l}'$ denote the first $l$ coordinates of $\eps'$
and $\mathcal{E}_{1:l}$ denote the projection of $\mathcal{E}$ onto the first
$l$ dimensions,
\[p_{\bS} \leq \P[\eps_{1:l}' \in \mathcal{E}_{1:l}]
\leq (2\pi)^{-l/2}\operatorname{Vol}(\mathcal{E}_{1:l})
=\frac{\left(\lambda \sqrt{\frac{n}{2}}\right)^l}
{\Gamma(\frac{l}{2}+1)}\prod_{i=1}^l \gamma_i.\]
Here, $\operatorname{Vol}(\cdot)$ denotes the volume in $\R^l$,
and we have used $\operatorname{Vol}(\sB_l)=\pi^{l/2}/\Gamma(\frac{l}{2}+1)$.
Note, by Cauchy interlacing, that $\gamma_i \leq s_{i+k-1}(D^T)$ where
$s_{i+k-1}(D^T)$ denotes the
$i+k-1^\text{th}$ smallest non-zero singular value of $D^T$. (Equality holds
if the $k-1$ directions in the column span of $D^T$ that are
projected out by $Q^\bS$ are exactly those corresponding to the smallest $k-1$
non-zero singular values of $D^T$.)
As $D^TD$ is the Laplacian of the linear chain graph, which has
non-zero eigenvalues $4\sin^2(i\pi/2n) \leq (i\pi/n)^2$
for $i=1,\ldots,n-1$ \cite{andersonmorley}, this implies
$\gamma_i \leq (i+k-1)\pi/n$ and hence, for any $l \in \{1,\ldots,n-k\}$,
\[p_{\bS} \leq \frac{\left(\frac{\lambda \pi}{\sqrt{2n}}\right)^l}
{\Gamma(\frac{l}{2}+1)}\frac{(l+k-1)!}{(k-1)!}
=\left(\frac{\lambda \pi}{\sqrt{2n}}\right)^l
\frac{l!}{\Gamma(\frac{l}{2}+1)}\binom{l+k-1}{k-1}.\]
Applying the Gamma function duplication formula
$\sqrt{\pi}l!=\Gamma(\frac{l}{2}+1)\Gamma(\frac{l}{2}+\frac{1}{2})2^l$,
and $\binom{l+k-1}{k-1} \leq \binom{n-1}{k-1} \leq \binom{n}{k}$, this yields
\begin{equation}\label{eq:psbound}
p_{\bS} \leq \left(\lambda \pi\sqrt{2/n}\right)^l
\frac{\Gamma(\frac{l+1}{2})}{\sqrt{\pi}}\binom{n}{k}.
\end{equation}

Fix $k \geq 1$. Suppose $\hat{\mu}$ induces the partition $\hat{\bS}$ having
$\hat{k}$ elements, and $\hat{k} \leq k$. Then
Lemma \ref{lemma:subgradient} implies $Q^{\hat{\bS}} Y-v
\in \lambda D^TW$, where $v$ and $W$ are defined in terms of $\hat{\bS}$ and
$\hat{\bs}$. Applying $Q^{\hat{\bS}}$ to both sides and noting
$Q^{\hat{\bS}} v=0$,
\[Q^{\hat{\bS}} Y \in \lambda Q^{\hat{\bS}} D^TW
\subseteq \lambda Q^{\hat{\bS}} D^T[-1,1]^{n-1}.\]
Thus there exists some partition $\bS$ into exactly $k$ consecutive intervals
(equal to or refining $\hat{\bS}$)
for which applying $Q^\bS$ to both sides above yields
$Q^\bS Y \in \lambda Q^{\bS} D^T[-1,1]^{n-1}$.
As there are $\binom{n-1}{k-1} \leq \binom{n}{k}$ such partitions, applying
a union bound, (\ref{eq:psbound}), and $\Gamma(x+1)
\leq \sqrt{2\pi}((x+0.5)/e)^{x+0.5}$ for any $x>0$ (see Theorem 1.5 of
\cite{batir}),
\begin{align*}
\P[\hat{k} \leq k]&\leq \binom{n}{k}^2\left(\lambda \pi\sqrt{2/n}\right)^l
\frac{\Gamma(\frac{l+1}{2})}{\sqrt{\pi}}\\
&\leq \sqrt{2}\exp\left(2k\log \frac{en}{k}+l\log \lambda\pi\sqrt{2/n}
+\frac{l}{2}\log\frac{l}{2}-\frac{l}{2}\right).
\end{align*}
Taking $l=\lfloor \frac{n}{\max(\lambda^2,1)\pi^2} \rfloor$
ensures $l \log \lambda \pi\sqrt{2/n}+\frac{l}{2}\log \frac{l}{2} \leq 0$,
and taking $k=\lfloor \frac{cn}{\max(\lambda^2,1)\log n} \rfloor$
for a sufficiently small constant $c>0$ then ensures
$\P[\hat{k} \leq k] \leq Ce^{-c'n/\max(\lambda^2,1)}$,
for constants $C,c'>0$, as desired.
\end{proof}

To show part (b) of Theorem \ref{thm:TVnull}, we will argue that on the event
where $\hat{\mu}$ induces a partition into at least $k$ intervals, the
squared-error of $\hat{\mu}$ is typically also at least $k$ (up to logarithmic
factors). We establish this using the next two lemmas.

Suppose that $\hat{\mu}$ induces the partition $\bS=(S_1,\ldots,S_k)$
and the signs $\bs=(s_1,\ldots,s_{k-1}) \in \{-1,1\}^{k-1}$. For any $\alpha \in
\{1,\ldots,k\}$, let us call $S_\alpha$ a local maximum interval if
$s_{\alpha-1} \geq 0$ and $s_\alpha \leq 0$, and a local minimum interval
if $s_{\alpha-1} \leq 0$ and $s_\alpha \geq 0$, where by convention
$s_0=s_k:=0$.
\begin{lemma}\label{lemma:localmaxmin}
Let $\lambda>0$, let $\hat{\mu}$ be the minimizer of (\ref{eq:objTV}), and
suppose $\mu_0$ is such that $\|\mu_0\|_\infty \leq M$. There exists a
constant $c>0$ such that with probability at least $1-2/n$, every local minimum
interval and local maximum interval of the partition and signs induced by
$\hat{\mu}$ has length at
least $c\min(\frac{\lambda^2}{\log n},\frac{\lambda}{M})$.
\end{lemma}
\begin{proof}
For any $v \in \R^n$ and any interval $S \subseteq (1,\ldots,n)$,
denote $\bar{v}_S=|S^{-1}|\sum_{i \in S} v_i$. Define the event
\[\mathcal{E}:=\left\{\text{there exists an interval } S \subseteq (1,\ldots,n):
|\bar{\eps}_S|>\sqrt{(6\log n)/|S|}\right\}.\]
By a Gaussian tail bound and union bound, $\P[\mathcal{E}]
\leq 2n^2e^{-3\log n}=2/n$.

Let $\bS=(S_1,\ldots,S_k)$ be the partition into consecutive intervals induced
by $\hat{\mu}$. Define
\[A:=\{\alpha:S_\alpha \text{ is a local maximum interval or local minimum
interval}\}.\]
Consider $\alpha \in A$ corresponding to the
smallest length $|S_\alpha|$. Let $\beta \in A$ be the element immediately
before or after $\alpha$ in $A$, and note that if $S_\alpha$ is a
local maximum, then $S_\beta$ is a local minimum, and vice versa.
Then, letting $v \in P^{\bS}(\R^n)$ be as defined in (\ref{eq:v}),
$|\bar{Y}_{S_\alpha}-\bar{Y}_{S_\beta}|$ exceeds
$|\hat{\mu}_{S_\alpha}-\hat{\mu}_{S_\beta}|$ by
$|v_{S_\alpha}|+|v_{S_\beta}|$, so we must have
\[|\bar{Y}_{S_\alpha}-\bar{Y}_{S_\beta}|
\geq |v_{S_\alpha}|+|v_{S_\beta}|
\geq \frac{\lambda}{|S_\alpha|}+\frac{\lambda}{|S_\beta|}
\geq \frac{\lambda}{|S_\alpha|}.\]
On the event $\mathcal{E}$,
\[|\bar{Y}_{S_\alpha}-\bar{Y}_{S_\beta}|
\leq |\bar{\eps}_{S_\alpha}-\bar{\eps}_{S_\beta}|+2M
\leq \sqrt{\frac{6\log n}{|S_\alpha|}}
+\sqrt{\frac{6\log n}{|S_\beta|}}+2M
\leq 2\sqrt{\frac{6\log n}{|S_\alpha|}}+2M,\]
where the last bound recalls that $S_\alpha$ has smallest length among $\alpha
\in A$. Then on $\mathcal{E}$,
\[0 \leq 2M|S_\alpha|+2\sqrt{6(\log n)|S_\alpha|}-\lambda,\]
which implies
\[\sqrt{|S_\alpha|} \geq 
\frac{\sqrt{24\log n+8M\lambda}-2\sqrt{6\log n}}{4M}
=\frac{2\lambda}{\sqrt{24\log n+8M\lambda}+2\sqrt{6\log n}}
\geq c\min\left(\frac{\lambda}{\sqrt{\log n}},\sqrt{\frac{\lambda}{M}}\right).\]
\end{proof}

\begin{lemma}\label{lemma:conebound}
Fix any $\mu \in \R^k$. Let $\xi \in \R^k$ be a random vector distributed as
$\xi \sim \Normal(\mu,\Id)$, and let $K \subseteq \R^k$ be any convex cone with
vertex 0 and non-empty interior. Then there exists a universal constant $c>0$
such that
\[\E[\|\xi\|^2 \mid \xi \in K] \geq \frac{ck^2}{\max(\|\mu\|^2,k)}.\]
\end{lemma}
\begin{proof}
Consider a change of variables to $S=\|\xi\|$ and
$\Theta$ a set of $k-1$ angular variables representing the angle of $\xi$.
Since $K$ is a cone with vertex 0, the condition $\xi \in K$ is a function only
of $\Theta$. Hence the joint density of $S$ and $\Theta$ conditional on $\xi \in
K$ is given by
\[p(s,\theta)=c(k,\theta,K)s^{k-1}e^{-\frac{\|\xi(s,\theta)-\mu\|^2}{2}},\]
for some quantity $c(k,\theta,K)$ independent of $s$. Note that for any
$s \geq t \geq 0$ and any $\theta$,
\begin{align*}
\|\xi(t,\theta)-\mu\|^2-\|\xi(s,\theta)-\mu\|^2
&=t^2-s^2+2\langle \xi(s,\theta)-\xi(t,\theta),\mu \rangle\\
&\geq t^2-s^2-2(s-t)\|\mu\|
=(t+\|\mu\|)^2-(s+\|\mu\|)^2.
\end{align*}
Letting $R$ denote a positive random variable with density function
\[q(r)=I(\mu)r^{k-1}e^{-\frac{(r+\|\mu\|)^2}{2}},\]
where $I(\mu)$ is the normalization constant such that $\int_0^\infty q(r)dr=1$,
this implies $p(s,\theta) \geq \frac{q(s)}{q(t)}p(t,\theta)$. Integrating over
$\theta$ and denoting by $p(s)$ the marginal distribution of $S$ conditional on
$\xi \in K$, this implies $p(s) \geq \frac{q(s)}{q(t)}p(t)$. The
inequality is reversed for $0 \leq s \leq t$, so
\begin{align*}
\P[\|\xi\| \geq t \mid \xi \in K]
=\frac{\int_t^\infty p(s)ds}
{\int_0^t p(s)ds+\int_t^\infty p(s)ds}
&\geq \frac{\frac{p(t)}{q(t)}\int_t^\infty q(s)ds}
{\int_0^t p(s)ds+\frac{p(t)}{q(t)}\int_t^\infty q(s)ds}\\
&\geq \frac{\frac{p(t)}{q(t)}\int_t^\infty q(s)ds}
{\frac{p(t)}{q(t)} \int_0^t q(s)ds+\frac{p(t)}{q(t)}\int_t^\infty q(s)ds}
=\P[R \geq t].
\end{align*}
That is, the distribution of $\|\xi\|$ conditional on $\xi \in K$ stochastically
dominates the (unconditional) distribution of $R$.

Thus, to conclude the proof of the lemma, it suffices to lower bound
\[\E[R^2]=\int_0^\infty r^{k+1}e^{-\frac{(r+\|\mu\|)^2}{2}}dr \Bigg/
\int_0^\infty r^{k-1}e^{-\frac{(r+\|\mu\|)^2}{2}}dr=:M_{k+1}/M_{k-1},\]
where $M_l$ is the $l^\text{th}$ moment of the $\Normal(-\|\mu\|,1)$
distribution truncated on the left at 0. Integration by parts yields the
recurrence $M_{k+1}=-\|\mu\|M_k+kM_{k-1}$, and Cauchy-Schwarz yields
$M_{k+1}M_{k-1} \geq M_k^2$. Then
\[\frac{M_{k+1}}{M_{k-1}}+\|\mu\|\sqrt{\frac{M_{k+1}}{M_{k-1}}}-k
\geq \frac{M_{k+1}}{M_{k-1}}+\|\mu\|\frac{M_k}{M_{k-1}}-k
\geq 0,\]
so
\[\sqrt{\frac{M_{k+1}}{M_{k-1}}} \geq \frac{\sqrt{\|\mu\|^2+4k}-\|\mu\|}{2}
=\frac{2k}{\|\mu\|+\sqrt{\|\mu\|^2+4k}},\]
and the result follows by taking the square.
\end{proof}

\begin{proof}[Proof of Theorem \ref{thm:TVnull}(b)]
For any positive constant $n_0$, it suffices to consider graphs with
$n \geq n_0$. (For $n<n_0$, the lower bound trivially holds by adjusting $c'$
based on $n_0$, as the risk is clearly non-zero for any graph.)

Let $\hat{\mu}$ induce the partition $\hat{\bS}$ and signs $\hat{\bs}$.
For any fixed partition $\bS=(S_1,\ldots,S_k)$ and signs $\bs$, 
by Lemma \ref{lemma:subgradient},
\[\E\big[\|\hat{\mu}\|^2 \mid \hat{\bS}=\bS,\hat{\bs}=\bs\big]
=\E\big[\|P^\bS Y+v\|^2 \mid P^\bS Y+v \in K,\,Q^\bS Y-v
\in \lambda D^TW\big],\]
where $v$, $K$, and $W$ are fixed and defined by $\bS$ and $\bs$.
Since the projections $P^\bS$ and $Q^\bS$ are orthogonal to each other, $P^\bS
Y$ is independent of $Q^\bS Y$, so we may drop the conditioning on the event
$Q^\bS Y-v \in \lambda D^TW$. Define $\xi \in \R^k$ by
$\xi_\alpha=|S_\alpha|^{1/2}(\bar{Y}_{S_\alpha}+v_{S_\alpha})$ for each
$\alpha=1,\ldots,k$, where
$\bar{Y}_{S_\alpha}=|S_\alpha|^{-1}\sum_{i \in S_\alpha} Y_i$.
Then $\|P^\bS Y+v\|^2=\|\xi\|^2$, $\xi \sim \Normal(w,\Id)$
for $w=(|S_1|^{1/2}v_{S_1},\ldots,|S_k|^{1/2}v_{S_k})$, and the condition
$P^\bS Y+v \in K$ is equivalent to $\xi \in \tilde{K}$ for some convex cone
$\tilde{K} \in \R^k$ with vertex 0 and non-empty interior.
Applying Lemma \ref{lemma:conebound}, for a constant $c'>0$,
\begin{equation}\label{eq:MSEboundTVnull}
\E\big[\|\hat{\mu}\|^2 \mid \hat{\bS}=\bS,\hat{\bs}=\bs\big]
=\E[\|\xi\|^2 \mid \xi \in \tilde{K}] \geq \frac{c'k^2}{\max(\|w\|^2,k)}.
\end{equation}

Let $\hat{k}$, $\hat{v}$, and $\hat{w}$ denote $k$, $v$, and $w$ as defined
above for the random induced partition and signs $\hat{\bS}$ and $\hat{\bs}$.
By part (a) of this theorem (already established), for some constant $n_0$ and
any $n \geq n_0$, on an event of probability at least $3/4$, 
$\hat{k} \geq \frac{c'n}{\max(\lambda^2,1) \log n}$.
Note that $\hat{v}$ is only non-zero on those intervals $\hat{S}_\alpha$
which are local maxima or local minima, and that
the (constant) value of $\hat{v}$ on any such interval has
magnitude at most $2\lambda/|\hat{S}_\alpha|$.
Applying Lemma \ref{lemma:localmaxmin} with $M=0$, on a different event of
probability at least $3/4$, every local minimum and local maximum interval of
$\hat{\bS}$ has length $|\hat{S}_\alpha| \geq c'\max(\lambda^2,1)/(\log n)$.
(For $\lambda<1$ this statement is trivial.) Then, on this event,
there are at most $\frac{n\log n}{c'\max(\lambda^2,1)}$ such intervals, and
hence
\[\|\hat{w}\|^2=
\sum_{\alpha=1}^{\hat{k}} |\hat{S}_\alpha|\hat{v}_{\hat{S}_\alpha}^2
\leq \sum_{\alpha:\hat{S}_\alpha \text{ local max or min}}
\frac{4\lambda^2}{|\hat{S}_\alpha|}
\leq \frac{n\log n}{c'\max(\lambda^2,1)}\frac{4\log n}{c'} \leq
\frac{4n(\log n)^2}{{c'}^2\max(\lambda^2,1)}.\]
Taking the full expectation of (\ref{eq:MSEboundTVnull}) over the intersection
of these two events of probability at least $1/2$, the theorem follows.
\end{proof}

Finally, we prove Theorem \ref{thm:TValternative} using
Theorem \ref{thm:TVnull}(b) and Lemma \ref{lemma:localmaxmin}:
\begin{proof}[Proof of Theorem \ref{thm:TValternative}]
Let $c_0$ denote the constant in Lemma \ref{lemma:localmaxmin}. If
\[\lambda^2 \leq \frac{3\log n}{c_0}\sqrt{\frac{n\Delta}{s}},\]
then the lower bound follows from Theorem \ref{thm:TVnull}(b)
by considering the risk at $\mu_0=0$. Otherwise, construct a signal $\mu_0$
by the following procedure:
\begin{enumerate}
\item Choose
\[k=\left\lfloor \min\left(\frac{s}{2},\frac{n}{3\Delta}\right)\right\rfloor\]
disjoint intervals $I_1,\ldots,I_k$ in $(1,\ldots,n)$, each of length $3\Delta$.
\item Set $M=c_0\lambda/(3\Delta)$. For each $\alpha=1,\ldots,k$, divide
$I_\alpha$ into three sub-intervals of length $\Delta$, and set $\mu_{0,i}=M$
for each $i$ belonging to the middle such sub-interval.
\item Set $\mu_{0,i}=0$ for all remaining indices.
\end{enumerate}
This signal satisfies $\|D\mu_0\|_0 \leq 2k \leq s$ and $\Delta(\mu_0) \geq
\Delta$. We have
\[\frac{\lambda^2}{\log n}>\frac{3}{c_0}\sqrt{\frac{n\Delta}{s}}
>\frac{3\Delta}{c_0}=\frac{\lambda}{M},\]
so Lemma \ref{lemma:localmaxmin} implies that on an event $\mathcal{E}$ of
probability at least $1-2/n$, every local maximum or local minimum interval for
the partition and signs induced by $\hat{\mu}$ has length at least $3\Delta$. On
this event $\mathcal{E}$, there cannot be a sub-interval of any $I_\alpha$
that is a local maximum or local minimum, so the estimate $\hat{\mu}$ must be
monotonically increasing or monotonically decreasing over each interval
$I_\alpha$. Then it is easily verified that
\[\|\hat{\mu}-\mu_0\|^2 \geq
\sum_{\alpha=1}^k \sum_{i \in I_\alpha} \|\hat{\mu}_i-\mu_{0,i}\|^2 \geq
ck\Delta M^2 \geq c'(\log n)\sqrt{\frac{sn}{\Delta}}\]
for constants $c,c'>0$, where the last inequality uses $\lambda^2 \gtrsim (\log
n)\sqrt{n\Delta/s}$ and $k \gtrsim s$.
The result follows upon taking the expectation over $\mathcal{E}$.
\end{proof}

\section*{Acknowledgments}
We would like to thank Professors Amin Saberi, James Sharpnack, and Emmanuel
Candes for helpful and inspiring discussions on this problem and related
literature.

\bibliography{references}{}
\bibliographystyle{alpha}
\end{document}